\documentclass[a4paper,USenglish]{lipics-v2018}
\usepackage{microtype}
\usepackage{tikz}
\usepackage{apxproof}
\usetikzlibrary{graphs}

\usepackage{xspace}
\usepackage{shuffle}
\DeclareFontFamily{U}{shuffle}{}
\DeclareFontShape{U}{shuffle}{m}{n}{ <-8>shuffle7 <8->shuffle10}{}

\usetikzlibrary{shapes.geometric}
\usetikzlibrary{positioning}
\usetikzlibrary{calc}
\usepackage{colonequals}
\usepackage{bm}
\usepackage{xcolor}

\nolinenumbers

\title{Topological Sorting with Regular Constraints}

\author{Antoine Amarilli}{LTCI, Télécom ParisTech, Université Paris-Saclay}{}{}{}
\author{Charles Paperman}{Université de Lille}{}{}{}
\authorrunning{A.\ Amarilli and C.\ Paperman} %

\acknowledgements{We thank Michaël Cadilhac and
Pierre McKenzie for their fruitful insights.}

\Copyright{Antoine Amarilli and Charles Paperman}%

\subjclass{}


\DOIPrefix{}

\theoremstyle{plain}
\newtheoremrep{theorem}{Theorem}[section]
\newtheoremrep{proposition}[theorem]{Proposition}
\newtheoremrep{lemma}[theorem]{Lemma}
\newtheoremrep{claim}[theorem]{Claim}
\newtheoremrep{conjecture}[theorem]{Conjecture}

\EventEditors{Ioannis Chatzigiannakis, Christos Kaklamanis, Daniel Marx, and Don Sannella}
\EventNoEds{4}
\EventLongTitle{45th International Colloquium on Automata, Languages, and Programming (ICALP 2018)}
\EventShortTitle{ICALP 2018}
\EventAcronym{ICALP}
\EventYear{2018}
\EventDate{July 9--13, 2018}
\EventLocation{Prague, Czech Republic}
\EventLogo{eatcs}
\SeriesVolume{80}
\ArticleNo{41}

\makeatletter
\newcommand*{\defeq}{\mathrel{\rlap{%
  \raisebox{0.3ex}{$\m@th\cdot$}}%
  \raisebox{-0.3ex}{$\m@th\cdot$}}%
  =}
\makeatother

\newcommand{\card}[1]{\left|{#1}\right|}

\newcommand{\calA}{\mathcal{A}}
\newcommand{\calK}{\mathcal{K}}

\newcommand{\NN}{\mathbb{N}}

\newcommand{\DA}{\textbf{DA}\xspace}
\newcommand{\A}{\textbf{A}\xspace}
\newcommand{\DO}{\textbf{DO}\xspace}
\newcommand{\DS}{\textbf{DS}\xspace}

\newcommand{\NP}{\textsc{NP}\xspace}

\renewcommand{\phi}{\varphi}

\newcommand{\ba}{\mathbf{a}}
\newcommand{\bb}{\mathbf{b}}

\newcommand{\PI}{\mathrm{PI}}

\newcommand{\relto}{\tikz[baseline=-\the\dimexpr\fontdimen22\textfont2\relax]{\node(t) {$\to$};\node at (t.north)[regular polygon, regular polygon sides=3,draw, scale=0.3,rotate=90]{};}}

\newcommand{\freq}{\mathrm{freq}}
\newcommand{\rare}{\mathrm{rare}}

\newcommand{\ccl}{\mathrm{CCl}}
\newcommand{\CCl}{\ccl}

\newcommand{\com}{\textrm{Com}}

\newcommand{\shuf}{\shuffle}

\newcommand{\CTS}{\mathrm{CTS}}
\newcommand{\CSh}{\mathrm{CSh}}
\newcommand{\pCTS}[1]{\mathrm{CTS}\!\left[#1\right]}
\newcommand{\pCSh}[1]{\mathrm{CSh}\!\left[#1\right]}

\newcommand{\powerset}[1]{\mathcal{P}(#1)}

\newcommand{\lfun}{\lambda}

\newcommand{\fst}{\mathsf{first}}
\newcommand{\nxt}{\mathsf{next}}

\newcommand{\myparagraph}[1]{\subparagraph*{#1.}}

\begin{document}

\maketitle

\begin{abstract}
  We introduce the \emph{constrained topological sorting problem} (CTS):
given a regular language~$K$ and a directed acyclic graph $G$ with labeled
vertices, determine if $G$ has a topological sort that forms a word in~$K$. This natural problem applies to several settings, e.g.,
scheduling with costs or verifying concurrent programs. 
We consider the problem~$\pCTS{K}$ where the target language~$K$ is fixed, and
study its complexity depending on~$K$. We show that $\pCTS{K}$ is tractable when
$K$ falls in several language families, e.g., \emph{unions of monomials}, which
can be used for pattern matching. However, we show that $\pCTS{K}$ is NP-hard
for $K = (ab)^*$ and introduce a \emph{shuffle
reduction} technique to show hardness for more languages. We also study the special case of the \emph{constrained shuffle problem}
(CSh), where the input graph is a disjoint union of strings, and show that
$\pCSh{K}$ is additionally tractable when~$K$ is a group language or a union of district
group monomials. We conjecture that a
dichotomy should hold on the complexity of $\pCTS{K}$ or $\pCSh{K}$ depending
on~$K$, and substantiate this by proving a coarser dichotomy under a different problem
phrasing which ensures that tractable languages are closed under common
operators.

\end{abstract}

\section{Introduction}
Many scheduling or ordering problems amount to computing a
\emph{topological sort} of a directed acyclic graph (DAG), i.e.,
a totally ordered sequence of the vertices that is
compatible with the edge relation: when we enumerate a vertex,
all its predecessors must have been enumerated first.
However, in some settings, we need a topological sort
satisfying additional constraints that cannot be 
expressed as edges.
We formalize this
problem as follows: the vertices of the DAG are labeled with some symbols from
a finite alphabet~$A$, and we want to find a topological sort that falls into a
specific regular language. We call this the \emph{constrained topological sort problem}, or CTS.
For instance, if we fix the language $K = ab^*c$, and consider the example
DAGs of Figure~\ref{fig:exa}, then $G_1$ and $G_2$ have a topological sort that falls
in~$K$.

CTS relates to many applications.
For instance, many \emph{scheduling} applications
use a dependency graph \cite{agrawal2016scheduling}
of tasks, and it is often useful to express other constraints,
e.g., 
some tasks must be performed by specific workers and we should not assign more
than~$p$
successive tasks to the same worker. We can express this as a CTS-problem: label
each task by the worker which can perform it, and consider the target
regular language~$K$ containing all words where the same symbol is not repeated more
than~$p$
times. In \emph{concurrency} applications, 
we may consider
a program with multiple threads, and want to verify that there
is no linearization of its instructions that exhibits some
unsafe behavior, e.g., executing a read 
before a write. To search for such a linearization, we can label each instruction with its type, and
consider CTS with a
target language describing the behavior that
we wish to detect.
CTS can also be used in uncertain data management tasks, to reason
about the possible answers of aggregate queries on uncertain ordered data~\cite{amarilli2017possible}.
It can also be equivalently phrased in the language of partial order theory: seeing the labeled DAG as a \emph{labeled partial
order} $<$, we ask if some \emph{linear extension} achieves a word in~$K$.

\begin{figure}
  \centering
\begin{tikzpicture}[>=stealth, baseline=(c.north)]
 \node (a)  {
             \tikz \graph[branch up=1.5em, fresh nodes] { a -> {b, b}->  c };
       };
  \node at ($(a.west) + (-1, .3)$) [anchor=north] {$G_1$};             

\node (b) at (a.south east) [anchor =  south west, xshift= 8em] {
  \tikz \graph[branch up=1.5em, fresh nodes] { a -> {b}, {b}  -> c };
       };
\node at ($(b.west) + (-1, .3)$) [anchor=north] {$G_2$};

\node (c) at (b.south east) [anchor = south west, xshift= 8em, yshift=.75em]{
            \tikz \graph[branch up=1.5em] { a -> c -> b};
       };
 \node at ($(c.west) + (-1, .3)$) [anchor=north] {$G_3$};             
\end{tikzpicture}
  \vspace{-.3cm}
  \caption{Example labeled DAGs on the alphabet $A = \{a, b, c\}$}
  \label{fig:exa}
  \vspace{-.3cm}
\end{figure}
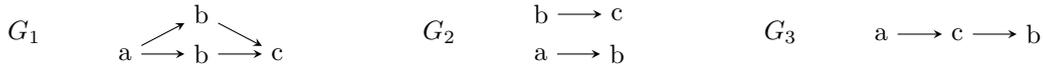

We thus believe that the CTS-problem is useful, and natural, but we are not
aware of previous work studying it, except for 
a special case called the
\emph{shuffle problem}.
This problem deals with the \emph{interleaving} of strings, as studied, e.g., 
in concurrent programming
languages~\cite{kimura1976algebraic,ogden1978complexity}, computational biology~\cite{kececioglu1998reconstructing},
and
formal
languages~\cite{eremondi2016complexity,buss2014unshuffling,rizzi2017recognizing}.
Specifically, we are given a tuple of strings, and we must decide if
they have some interleaving that falls in the target language~$K$.
This problem was known to be
NP-complete~\cite{mansfield1983computational,warmuth1984complexity,johnson1984np}
when the target language $K$ is given as input (in addition to the tuple of
strings), even when $K$ consists of just one target string. 
To rephrase this shuffle problem in our context, we call \emph{constrained
shuffle problem} (CSh) the special case of CTS where we require input DAGs to be a
union of directed path graphs (corresponding to the strings).

Our goal in this paper is to study the complexity of CTS and CSh.
We assume that the target regular language~$K$ is fixed, and call $\pCTS{K}$ and
$\pCSh{K}$ the corresponding problems, whose complexity is only a function of the
input DAG (labeled on the alphabet $A$ of~$K$). Our central question is: \emph{for
which regular languages $K$ are the problems $\pCTS{K}$ or $\pCSh{K}$ tractable?}
More precisely, for each of these problems, we conjecture a dichotomy on~$K$: the
problem is either in NL or it is NP-complete. However, the tractability boundary
is challenging to chart out, and we have not been able to prove these
conjectures in full generality. In this paper, we present the results that we
have obtained towards this end.
\myparagraph{Paper structure}
We formally define the CTS and CSh problems in Section~\ref{sec:mainres} and state the
conjecture. We then show the following results:
\begin{itemize}
  \item In Section~\ref{sec:hardness}, we present our hardness results. We
    recall the results of~\cite{warmuth1984complexity} on the shuffle
    problem, and present a general \emph{shuffle reduction} technique to show hardness for more
    languages. We use it in particular to show that $\pCSh{(ab)^*}$, hence $\pCTS{(ab)^*}$, are
    NP-hard, and extend this to several other languages.
  \item In Section~\ref{sec:sigma2}, we present tractability results. 
    We show that $\pCTS{K}$, hence $\pCSh{K}$, is in non-deterministic logspace (NL)
    when $K$ is a union of \emph{monomial languages}, i.e., of languages of the form $A_1^*
    a_1 \cdots A_{n-1}^* a_{n-1} A_n^*$, with the $a_i$ being letters and the
    $A_i$ being subalphabets. Such languages can be used for applications such
    as pattern
    matching, e.g., with the language $A^* u A^*$ for a fixed pattern~$u \in
    A^*$. We also show
    tractability for other languages that are not of this form, e.g. $(ab)^* + A^*
    aa A^*$ and variants thereof, using different techniques such as
    Dilworth's theorem~\cite{dilworth1950decomposition}.
  \item In Section~\ref{sec:dich}, we use our hardness and tractability results
    to show a coarser dichotomy result. Specifically, we give an
    alternative phrasing
    of the CTS and CSh problems using \emph{semiautomata} and DAGs with
    \emph{multi-letter} labels: this amounts to closing the
    tractable languages under intersection, inverse morphism, complement, and
    quotients. In this phrasing, when the semiautomaton is counter-free, 
    we can show that the problems are either in NL or NP-complete. This dichotomy
    is effective, i.e., the criterion on the semiautomaton is decidable, and it
    turns out to be the same for CTS and CSh.
  \item In Section~\ref{sec:group}, we focus on the constrained shuffle problem,
    and lift the counter-free assumption of the previous section. 
    We show that $\pCSh{K}$ is tractable when $K$ is a \emph{group language} or more
    generally a union of \emph{district group monomials}. This
    tractability result is the main technical contribution of the paper, with a
    rather involved proof.
    It implies, e.g., that the following problem is in NL for any fixed finite
    group $H$: given $g\in H$ and words $w_1, \ldots, w_n$ of elements of~$H$,
    decide whether there is an interleaving of the~$w_i$ which evaluates to~$g$
    according to the group operation.
\end{itemize}


\section{Problem Statement and Main Results}
\label{sec:mainres}
We give some preliminaries and define the two problems that we study.
We fix a finite alphabet~$A$, and call $A^*$ the set of all finite words on~$A$.
For $w \in A^*$, we 
write
$\card{w}$ for the
\emph{length} of~$w$,
and write $|w|_a$ for the number of occurrences
of~$a\in A$ in~$w$.
We denote the empty word by $\epsilon$.
A \emph{labeled DAG} on the alphabet $A$, or
\emph{$A$-DAG}, is a triple $G = (V, E, \lfun)$ where $(V, E)$ is a directed
acyclic graph with vertex set $V = \{1, \ldots, n\}$ and edge set $E \subseteq V \times V$, and
where $\lfun:V
\to A$ is a function giving a label in~$A$ to each vertex in~$V$. 
For $u \neq v$ in~$V$, we say that $u$ is an \emph{ancestor} of~$v$ if there is
a directed path from~$u$ to~$v$ in~$G$, we say that $u$ is a \emph{descendant}
of~$v$ if $v$ is an ancestor of~$u$, and otherwise we call $u$ and $v$ 
\emph{incomparable}.
A \emph{topological sort}
of~$G$ is a bijective function~$\sigma$ from~$\{1, \ldots,
n\}$ to~$V$ such that, for all $(u, v) \in E$, we have $\sigma^{-1}(u) <
\sigma^{-1}(v)$.
The word \emph{achieved} by~$\sigma$ is
$\lfun(\sigma) \colonequals \lfun(\sigma(1)) \cdots \lfun(\sigma(n)) \in A^*$.

The \emph{constrained topological sort problem} $\pCTS{K}$ for a fixed
language $K \subseteq A^*$ (described, e.g., by a regular expression)
is defined as follows: given an $A$-DAG $G$,
determine if there is a topological sort $\sigma$ of~$G$
such that
$\lfun(\sigma) \in K$ (in which case we say that $\sigma$ \emph{achieves}~$K$).

We now define the \emph{constrained shuffle problem} (CSh).
Given two words $u, v \in A^*$,
the \emph{shuffle}~\cite{warmuth1984complexity} of $u$ and $v$, written $u \shuf v$, is the set of words that
can be obtained by interleaving them. Formally, a word $w \in A^*$
is in $u \shuf v$ iff there is a partition $P \sqcup Q$ of $\{1,
\ldots, \card{w}\}$ such that $w_P = u$ and $w_Q = v$, where $w_P$ denotes the
sub-word of~$w$ where we keep the letters at positions in~$P$, and likewise
for $w_Q$. 
The \emph{shuffle} $\shuf(U)$ of a tuple of words~$U$
is defined by
induction as follows:
we set
$\shuf() \colonequals \{\epsilon\}$, set
$\shuf(u) \colonequals \{u\}$,
and set $\shuf(u_1, \ldots, u_n, u_{n+1}) \colonequals \bigcup_{v \in \shuf(u_1,
\ldots, u_n)} v \shuf u_{n+1}$.
The \emph{constrained shuffle problem}
$\pCSh{K}$
for a fixed language $K \subseteq A^*$
is defined as follows: given a tuple of words~$U$,
determine if $K \cap \shuf(U)$ is nonempty.
Of course, $\pCSh{K}$ is a special case of $\pCTS{K}$: we can code any tuple of words~$U$
as an $A$-DAG $G_U$ by coding
each $u \in U$ as a directed path graph $v_1 \rightarrow \cdots \rightarrow v_{\card{u}}$
with $\lfun(v_i) = u_i$ for all $1 \leq i \leq \card{u}$.
Thus, we will equivalently see inputs to CSh as tuples
of words (called \emph{strings} in this context) or as $A$-DAGs that are
unions of directed path graphs.

\begin{example}
  \label{exa:pb}
  The problem $\pCTS{(ab)^*}$ on an input $\{a,b\}$-DAG
$G$ asks if $G$ has a topological sort starting with an $a$,
ending with a $b$, and alternating
between elements of each label. 
The problem $\pCSh{(aa+b)^*}$ on a tuple $U$ of strings on~$\{a,b\}$ asks
if there is an interleaving $w \in \shuf(U)$ such that all $a^*$-factors in~$w$ are of
even length (e.g., $bbaabaaaa$, but not $baaabb$).
\end{example}

In this work, we study the complexity of the problems $\pCTS{K}$ and $\pCSh{K}$ depending
on the language $K$.
Clearly we can always solve these problems by guessing a topological sort (or an interleaving), and verifying that it
achieves a word in~$K$. Hence, the complexity is always in $\NP^K$,
that is, in non-deterministic PTIME with an oracle for the \emph{word problem} of
$K$, which we can call to test if an input word in is~$K$:

\begin{propositionrep}
  For any language $K$, the problems $\pCTS{K}$ and $\pCSh{K}$ are in $\NP^K$.
\end{propositionrep}

\begin{proof}
  As explained in the main text, we guess a permutation $\sigma$ of the input
  vertices, check that it respects the order constraints, and use the oracle for
  the word problem to check that the word achieved by~$\sigma$ is in~$K$.
\end{proof}

In particular, the problems are in $\NP$ when the language~$K$ is regular,
because the word problem for regular languages is in PTIME. We will study
regular languages in this work.
We believe that regular languages can be classified
depending on the complexity of these problems, and make the following \emph{dichotomy
conjecture}:

\begin{conjecture}
  \label{con:maincon}
  For every regular language $K$, the problem $\pCTS{K}$ is either in NL or
  NP-complete. Likewise, the problem $\pCSh{K}$ is either in NL or NP-complete.
\end{conjecture}

Towards this conjecture, we determine in this paper the complexity of CTS and
CSh for several languages and classes. We first show in the next
section that these problems are hard for some languages such as~$(ab)^*$, and we
then show tractability results in Section~\ref{sec:sigma2}, and a coarser dichotomy
result in Section~\ref{sec:dich} under an alternative phrasing of our problems.

\section{Hardness Results}
\label{sec:hardness}
Our hardness results are based on the \emph{shuffle problem}  of formal language
theory which asks, given a word $w \in A^*$ and a tuple $U$ of words of~$A^*$, whether $w
\in \shuf(U)$. This problem is known to be NP-hard already on the alphabet $\{a,
b\}$ (see~\cite{warmuth1984complexity}). The shuffle problem is different from
CSh, because the target word of the shuffle problem is given as input, 
whereas the target regular language of CSh is fixed.
However, the hardness of the shuffle problem directly implies the hardness of
$\CSh$, hence of $\CTS$, for a well-chosen target language:

\begin{toappendix}
  \subsection{Proof of Proposition~\ref{prp:easyhard}: Direct Hardness Result}
\end{toappendix}

\begin{propositionrep}
  \label{prp:easyhard}
  Let $K_0 \colonequals (a_1 a_2 + b_1 b_2)^*$. 
  The problem $\pCSh{K_0}$ is NP-hard.
\end{propositionrep}

\begin{proofsketch}
  We can reduce a shuffle instance $(w, U)$ to the instance $I \colonequals w_1
  \cup U_2$ for~$\pCSh{K_0}$, where $w_1$ is~$w$ but adding the subscript~$1$ to
  all labels, and $U_2$ is defined analogously. A topological sort of~$I$
  achieving~$K_0$ must then alternate between~$w_1$
  and~$U_2$, and enumerate letters with the same label (up to the subscript), witnessing that $w \in \shuf(U)$.
\end{proofsketch}

\begin{proof}
  We construct the CSh-instance~$I$ in PTIME from the input instance
  to the shuffle problem as explained in the proof sketch, and argue for
  correctness in more detail. It is clear that, to achieve~$K_0$, a topological
  sort $\sigma$ of~$I$ must enumerate alternatively a letter with subscript~$1$
  and a letter with subscript~$2$, so it must enumerate alternatively from~$w_1$
  and from~$U_2$, and the definition of~$K_0$ ensures that the two letters
  in~$w$ and~$U$ corresponding to the enumerated letters in~$w_1$ and~$U_2$ must
  have the same label in~$\{a, b\}$. Hence, considering the restriction
  $\sigma'$ of~$\sigma$ to~$U_2$, the interleaving of~$U$ that corresponds
  to~$\sigma'$ witnesses that $w \in \shuf(U)$.

  Conversely, if $w \in \shuf(U)$, starting from a witnessing topological
  sort~$\sigma'$ of~$U_2$, it is clear that we can construct a topological sort
  $\sigma$ of~$I$ that achieves~$K_0$, by enumerating the elements of~$w_1$
  alternatively with the elements of~$U_2$ according to~$\sigma'$. This shows
  correctness and concludes the proof.
\end{proof}

In this section, we will refine this approach to show hardness for more
languages.
We first recall another initial hardness result from~\cite{warmuth1984complexity}.
We then introduce a general \emph{shuffle reduction} technique to show the hardness of
languages by reducing from other hard languages. Last, we show that CTS and CSh
are hard for the language $(ab)^*$ and for other languages.

\begin{toappendix}
  \subsection{Proof of Lemma~\ref{lem:partition}: Initial Hard Family}
\end{toappendix}

\myparagraph{Initial hard family}
To bootstrap the hardness results of~\cite{warmuth1984complexity} on the shuffle problem (on input words) to our
CSh-problem (on fixed languages),
we generalize the definition of CSh
to a \emph{regular language family} $\calK$, i.e., a (generally infinite) family
of regular languages, each of
which is described as a regular expression.
The CSh-problem for~$\calK$, written
$\pCSh{\calK}$, asks, given a regular expression $K \in \calK$
and a set of strings $U$,
whether $K \cap \shuf(U)$ is nonempty. In other words, we no longer fix one single
target language but a family~$\calK$ of target languages, and the input
chooses one target language from the 
family~$\calK$.
The following is then shown in~\cite{warmuth1984complexity} by reducing from
UNARY-3-PARTITION~\cite{garey1975complexity}:

\begin{lemmarep}{\bfseries \emph{(\cite{warmuth1984complexity}, Lemma~3.2)}}
  \label{lem:partition}
  Let $\calK \colonequals \{(a^ib^i)^* \mid i \in \NN\}$. Then
  $\pCSh{\calK}$ is NP-hard.
\end{lemmarep}

\begin{proof}
  For completeness, we summarize here the proof of~\cite{warmuth1984complexity}:
  see the statement and proof of Lemma~3.2 in~\cite{warmuth1984complexity} for
  details.
  The reduction is from
  UNARY-3-PARTITION: given a tuple $E$ of $3m$ positive integers written in
  unary, such that $B \colonequals \frac{1}{m} \left(\sum_{1
  \leq i \leq 3m} n_i\right)$ is an integer,
  and such that $B/4 < e < B/2$ for each $e \in E$,
  decide whether $E$ can be partitioned into $m$
  triples, with each triple summing to~$B$. This problem is
  NP-hard~\cite{garey1975complexity}.
  Given a UNARY-3-PARTITION instance $(E, B)$, we create a $\CSh$ instance~$I$ by
  writing each integer $n$ as the string $a^n b^n$, and
  we choose the target language $K$ in~$\calK$ to be $(a^B b^B)^*$, which is clearly a PTIME
  reduction. Clearly, if $(E, B)$ can be partitioned in triples summing to~$B$, then we can define a
  topological sort of~$I$ by enumerating, for each triple, the $B$ copies of the
  $a$'s in that triple, and then the $b$'s, achieving a word of~$K$.
  Conversely, any topological sort achieving a word of~$K$ must start by
  enumerating $B$ copies of~$a$'s followed by the same number of~$b$'s, and the
  only way to free sufficiently many~$b$'s is to enumerate completely the
  initial $a$ segments of some strings: we know that the number of such
  strings is exactly $3$ by our assumption that $B/4 < e < B/2$ for all $e \in
  E$.
  Hence, by applying this argument repeatedly, a topological sort of~$I$
  achieving~$K$ must define 
  a solution to $(E, B)$, completing the proof of the reduction.
\end{proof}

\begin{toappendix}
  \subsection{Proof of Theorem~\ref{thm:reduction}: Shuffle Reduction}
\end{toappendix}
\myparagraph{Shuffle reduction}
Our goal in this section is to show the hardness of CTS and CSh for more
languages, but we do not wish to prove hardness for every
language from scratch. Instead, we will introduce a general tool called the
\emph{shuffle reduction} that allows us to leverage the hardness of a
language~$K$ to show that another language~$K'$ is also hard.
Specifically, if a language $K$ shuffle-reduces to a language~$K'$, this will imply that there
is a PTIME reduction from $\pCTS{K}$ to~$\pCTS{K'}$, and from $\pCSh{K}$
to~$\pCSh{K'}$.

The intuition for the shuffle reduction is as follows: to reduce 
from~$K$ to~$K'$,
given an input $A$-DAG $G$, we build an $A$-DAG $G'$ 
formed of~$G$ plus an additional directed path labeled by a word
$w$. Thus, any topological sort $\sigma'$ of~$G'$ must be the interleaving
of~$w$ and of a topological sort $\sigma$ of~$G$. Now, if we require that
$\sigma'$ achieves~$K'$, the presence of~$w$ can impose specific conditions
on~$\sigma$. Intuitively, if~$w$ is sufficiently long and ``far away'' from all words of~$K'$,
then~$\sigma'$ must ``repair'' $w$ to a word of~$K'$ by inserting symbols
from~$G$, so the insertions performed by $\sigma$ may need to be in a specific
order, i.e., $\sigma$ may be forced to achieve a word of~$K$. This means that
solving $\pCTS{K'}$ on~$G'$ allows us to solve
$\pCTS{K}$ on~$G$.
This
intuition is illustrated on Figure~\ref{fig:reduction}: to achieve a word
of~$K' \colonequals (ab)^*$ on the DAG~$G'$, a topological sort must enumerate elements
from~$G$ to insert them at the appropriate positions in~$w$, achieving a word
of~$K \colonequals (ba)^*b$. We call \emph{filter sequence} a family of words like~$w$ that
allow us to reduce any $\pCTS{K}$-instance to~$\pCTS{K'}$. Formally:

\begin{definition}[(Filter sequence)]
  \label{def:filter}
   Let $K$ and $K'$ be languages on an alphabet~$A$.
  A \emph{filter sequence} for $K$ and $K'$ is an infinite sequence $(f_n)$ of
  words of~$A^*$ having the following property: for every $n \in \NN$, for every
  word $v \in A^*$ 
  such that $\card{v} = n$,
  we have $v \in K$ iff $(v \shuf f_n) \cap K' \neq
  \emptyset$.
\end{definition}
\begin{figure}
  \centering
  \newcommand{\rv}[1]{\textcolor{green!60!black}{#1}}
\begin{tikzpicture}[>=stealth, remember picture, every node/.style={inner sep=0.5pt}]
  \node (a1)  {a};
  \node (a11) at (a1.west) [anchor=east, xshift=-1em] {$w$};
  \node (b1) at (a1.east) [anchor= west, xshift=2em] {a};
  \node (c1) at (b1.east) [anchor= west, xshift=2em] {b};
  \node (a2) at (c1.east) [anchor= west, xshift=2em] {b};
  \node (b2) at (a2.east) [anchor= west, xshift=2em] {a};
  \node (c2) at (b2.east) [anchor= west, xshift=2em] {a};
  \node (a3) at (c2.east) [anchor= west, xshift=2em] {b};
  \node (b3) at (a3.east) [anchor= west, xshift=2em] {b};
  \node (c3) at (b3.east) [anchor= west, xshift=2em] {a};
  \node (a4) at (c3.east) [anchor= west, xshift=2em] {a};
  \node (b4) at (a4.east) [anchor= west, xshift=2em] {b};

  \draw[->] (a1) -- (b1); 
  \draw[->] (b1) -- (c1);
  \draw[->] (c1) -- (a2);
  \draw[->] (a2) -- (b2);
  \draw[->] (b2) -- (c2);
  \draw[->] (c2) -- (a3); 
  \draw[->] (a3) -- (b3);
  \draw[->] (b3) -- (c3);
  \draw[->] (c3) -- (a4);
  \draw[->] (a4) -- (b4);
  \node (I) at (c2.south) [draw, anchor=north, yshift=-1.5em, xshift=-1.25em, color=green!60!black, draw=black, fill=white!95!black, inner sep=4pt]
  {
    \begin{tikzpicture}
      \node (Ia1) []{b};
      \node (Ib1) at (Ia1.east) [anchor=west, xshift=1.5em]{a};
      \node (Ic1) at (Ib1.east) [anchor=west, xshift=1.5em]{b};
      \node (Ia2) at (Ic1.east) [anchor=west, xshift=1.5em, yshift=1em]{a};
      \node (Ib2) at (Ic1.east) [anchor=west, xshift=1.5em, yshift=-1em]{b};

      \draw[->, color=black] (Ia1) -- (Ib1); 
      \draw[->, color=black] (Ib1) -- (Ic1);
      \draw[->, color=black] (Ic1) -- (Ia2);
      \draw[->, color=black] (Ic1) -- (Ib2);
      \node at (Ia1) [yshift=-1em, color=black,] {$G$};
    \end{tikzpicture}
  };
  \draw [->, dashed,color=red] (a1.south) .. controls ++(0,-0.5) and ++(-0.5,0) .. (Ia1.west);
  \draw [->, dashed,color=red] (Ia1.north) .. controls ++(0,0.3) and ++(0,-0.5) .. ([xshift=-2pt]b1.south);
  \draw  (b1.north) edge[->, bend left, dashed,color=red] (c1.north);  
  \draw [->, dashed,color=red] ([xshift=2pt]c1.south) .. controls ++(0,-0.3) and ++(0,0.5) .. ([xshift=-2pt]Ib1.north);
  \draw [->, dashed,color=red] ([xshift=2pt]Ib1.north) .. controls ++(0,0.3) and ++(1,-0.75) .. ([xshift=-2pt]a2.south);
  \draw  (a2.north) edge[->, bend left, dashed,color=red] (b2.north);  
  \draw [->, dashed,color=red] (b2.south east) .. controls ++(0.,-0.5) and ++(0,0.75) .. ([xshift=-2pt]Ic1.north);
  \draw [->, dashed,color=red] ([xshift=2pt]Ic1.north) .. controls ++(0.3,0.3) and ++(0,-0.) .. (c2.south);

  \draw  (c2.north) edge[->, bend left, dashed,color=red] (a3.north);  
  \draw [->, dashed,color=red] (a3.south east) .. controls ++(0.,-0.5) and ++(0,0.75) .. ([xshift=-2pt]Ia2.north);
  \draw [->, dashed,color=red] (Ia2.north east) .. controls ++(0.5,0.15) and ++(0,-0.75) .. ([xshift=-2pt]b3.south);
  \draw  (b3.north) edge[->, bend left, dashed,color=red] (c3.north);
  \draw [->, dashed,color=red] (c3.south) .. controls ++(-.5,-.5) and ++(2,1) ..
  ([yshift=-2pt]Ib2.north);
  \draw [->, dashed,color=red] ([yshift=2pt]Ib2.south east) .. controls ++(0.5,0.5) and ++(0,-0.75) .. (a4.south);
  
  \draw  (a4.north) edge[->, bend left, dashed,color=red] (b4.north);

  \node at (b2.south east) [xshift=17.5em, yshift=-4em, align=left, fill=white]
  {
    Global resulting word: $a\rv{b}ab\rv{a}ba\rv{b}ab\rv{a}ba\rv{b}ab$\\
    Resulting word on $G$: $\rv{babab}$
  };
\end{tikzpicture}
  \vspace{-.3cm}
  \caption{Example of a shuffle reduction from $K \colonequals (ba)^*b$ to $K'
  \colonequals (ab)^*$}
  \label{fig:reduction}
\end{figure}
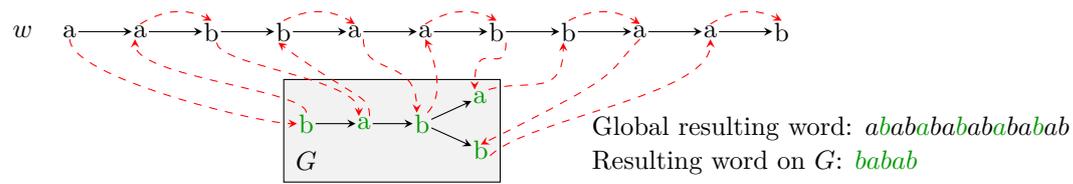
In Figure~\ref{fig:reduction}, we can choose $f_5 \colonequals w$ when defining
a filter
sequence for~$(ba)^*b$ and~$(ab)^*$: indeed, if we interleave $w$ with
any DAG $G$ of~$5$ vertices, then a topological sort $\sigma$ of~$G$
achieves~$K$ iff some interleaving $\sigma'$ of~$\sigma$ with~$w$ achieves~$K'$.
We can now define our reduction:
\begin{definition}[(Shuffle reduction)]
  \label{def:reduction}
  We say that a language $K$ \emph{shuffle-reduces} to a language~$K'$
  if there is a filter sequence $(f_n)$ for $K$ and $K'$ such that the
  function $i \mapsto f_i$
  is computable in PTIME (where $i$ is given in unary).

  We say that a
  regular language family $\calK$ \emph{shuffle-reduces} to~$K'$ if each $K$ does, and 
  if we can compute in PTIME the function
  $(K, i) \mapsto f_i^K$, which maps a regular expression $K$ of~$\calK$ and an
  integer $i$ in unary to the
  $i$-th word in a filter sequence $(f_n^K)$ for $K$ and~$K'$.
\end{definition}

\begin{theoremrep}
  \label{thm:reduction}
  For any regular language family $\calK$ and language~$K'$,
  if $\calK$ shuffle-reduces to~$K'$ 
  then 
  we can reduce in PTIME 
  from $\pCTS{\calK}$ to $\pCTS{K'}$, and from
  $\pCSh{\calK}$ to $\pCSh{K'}$.
\end{theoremrep}

\begin{toappendix}
  Of course, note that this result also applies to languages and not just to
  language families, because we can always take $\calK$ to be a singleton
  language family containing only one single language.
\end{toappendix}

\begin{proof}
  We show the result for the CSh-problem; the result for the CTS-problem is
  shown in exactly the same way.
  Fix the family $\calK$ and language $K'$.
  Let $K$ be the input language of~$\calK$, and let 
  $I$ be an input instance of the CSh-problem for~$K$.
  Let $(f^K_n)$ be the filter
  sequence for $K$ and $K'$. 
  Letting $n \colonequals\card{I}$,
  let us call $I'$ the
  instance of the CSh-problem for~$K'$ that contains $I$ and a separate
  string labeled with $f^K_n$: by our computability hypothesis on $(f^K_n)$, this is
  computable 
  in PTIME.
  We now argue that $I'$ is a positive
  instance to the CSh-problem for~$K'$ iff $I$ is a positive instance to the
  CSh-problem for~$K$. Indeed, assuming that there is a word $v$ of~$K$ achieved
  by a topological sort $\sigma$ of~$I$, we have $\card{v} = n$
  by definition,
  so by definition of $(f^K_n)$ we have $(v \shuf f^K_n) \cap K' \neq \emptyset$.
  Hence, let $v'$ be an element of this set. It is in $v \shuf f^K_n$, so it can
  be obtained as a topological sort of~$I'$ by shuffling $f^K_n$ with $\sigma$, 
  and it is in $K'$ so it witnesses that $I'$ is a
  positive instance to the CSh-problem for~$K'$. 
  
  Conversely, if there is a
  topological sort $\sigma'$ of~$I'$ achieving a word~$v' \in K'$,
  then~$\sigma'$ defines a topological sort  $\sigma$
  of~$I$ achieving a word~$v$ such that $v' \in v \shuf f_n$. As we have
  $\card{v} = n$ by definition, and $v'$ witnesses that $v \shuf f_n$ is
  non-empty, we must have $v \in K$, 
  so that $\sigma$ witnesses that $I$ is a positive instance to the
  CSh-problem for~$K$. This establishes correctness, and concludes the
  proof.
\end{proof}

\begin{toappendix}
Note that, for simplicity, we have shown Theorem~\ref{thm:reduction} for PTIME
reductions. This is because we only use the shuffle reduction in this paper to prove
  NP-hardness results. However, Theorem~\ref{thm:reduction} result can also be shown for
  NL reductions if we further assume that the filter sequences can be computed in
  logspace, i.e., the function mapping the
unary representation of~$n$ to the word~$f_n$ is computable by a logspace
transducer.
\end{toappendix}

\begin{toappendix}
  \subsection{Proof of Theorem~\ref{thm:abshuffle}: Hardness of $(ab)^*$}
\end{toappendix}
\myparagraph{Hardness for $(ab)^*$}
We now use the shuffle reduction and the language family of Lemma~\ref{lem:partition} to show the hardness of
$(ab)^*$. This
will be instrumental for our coarser dichotomy in Section~\ref{sec:dich}:

\begin{theoremrep}
  \label{thm:abshuffle}
  The problem $\pCSh{(ab)^*}$ (hence $\pCTS{(ab)^*}$) is NP-hard.
\end{theoremrep}

\begin{proofsketch}
  We shuffle-reduce from the language family $\calK$ of 
  Lemma~\ref{lem:partition}: for the language $K_B = (a^B b^B)^*$ of~$\calK$, we define the
  filter sequence for words of length $2 B n$ by $f^B_{2Bn} \colonequals (b^{B}
  a^{B} ab)^n$. This ensures that, when interleaving $f^B_{2Bn}$ with a word $v$
  of length~$2Bn$ to achieve a word of~$(ab)^*$, we must use~$v$ to insert
  in~$f^B_{2Bn}$ the letters written in bold: $((\ba b)^{B} (a \bb)^{B} ab)^n$.
  This can be done iff $v = (a^B b^B)^n$, i.e., iff $v \in K_B$. We conclude by
  Theorem~\ref{thm:reduction}.
\end{proofsketch}

\begin{proof}
  Let $\calK$ be the family of regular languages defined in
  Lemma~\ref{lem:partition}. We define a filter sequence $(f^B_n)$ for 
  each such language $K_B = (a^B b^B)^*$ for $B \in \NN$. We first explain how
  to define the filter sequence for
  word lengths of the form $2 B n$, for which there are words in~$K_B$ having
  the specified length. For such lengths,
  we define 
  $f^B_{2Bn} \colonequals (b^{B} a^{B} ab)^n$. For other word
  lengths $n' \in \NN$, for which there are no words in~$K_B$, we define
  $f^B_{n'} \colonequals a^{n'+1}$: this ensures that we can never obtain a word
  of~$K_B$ by interleaving $n'$ symbols with~$f^B_{n'}$, which is correct.
  The filter sequence is clearly computable in PTIME.
  So we only have to
  show that, for all $n \in \NN$, the word $f^B_{2Bn}$ is a filter sequence for
  word length $2Bn$.
  
  To do so, fix $n \in \NN$, and consider a word $v$ of length $2 B n$ in~$K_B$.
  For the forward direction, if $v$ is $(a^B b^B)^n$ which is the only word
  of~$K_B$ of length $2 B n$, then we can interleave $v$ with $f^B_{2Bn}$ to form a word
  of $(ab)^*$ by inserting the letters in
  bold: $((\ba b)^{B} (a \bb)^{B} ab)^n$.
  
  Conversely, for the backward direction, we are forced to insert at least these
  letters. More precisely, considering an interleaving of $f^B_{2Bn}$ with a
  word $v$ that achieves a word $w$ of~$(ab)^*$, we know that, in each of the
  $n$ occurrences of $b^B a^B ab$ in~$v$, each of the $B$ first $b$'s must be
  preceded by an~$a$ in~$w$ (so $B$ insertions of~$a$), then each of the $B$
  occurrences of~$a$ must be followed by a~$b$ in~$w$ (so $B$ insertions
  of~$b$). As we have $\card{v} = 2 B n$, we must perform these insertions in
  this order, and as they do not overlap, this completely specifies~$v$: so we know that $(v \shuf
  f^B_{2Bn}) \cap (ab)^*$ is non-empty iff $v = (a^B b^B)^n$. This shows that
  $f^B_{2B n}$ is indeed a filter sequence, which establishes that
  $\pCSh{(ab)^*}$ is NP-hard thanks to Theorem~\ref{thm:reduction} and
  Lemma~\ref{lem:partition}.
\end{proof}

\begin{toappendix}
  \subsection{Hardness Proofs for Other Languages}
\end{toappendix}
\myparagraph{Other hard languages}
From the hardness of $(ab)^*$, we can use the shuffle reduction to show
hardness for many other languages. For instance, we can show hardness for any
language~$u^*$, where $u \in A^*$ is a word with two different letters:

\begin{propositionrep}
  \label{prp:hardu}
  Let $u \in A^*$ such that $|u|_a > 0$ and $|u|_b > 0$ for $a \neq b$ in~$A$.
  Then $\pCSh{u^*}$ (hence $\pCTS{u^*}$) is NP-hard.
\end{propositionrep}

\begin{proofsketch}
  We shuffle-reduce from~$(ab)^*$ with the filter sequence $f_{2n} \colonequals (u u_{-a} u u_{-b} u)^n$, where $u_{-a}$
  (resp.\ $u_{-b}$) is $u$ but removing one occurrence of~$a$ (resp.\ of~$b$).
  If a word $v$ with $\card{v} = 2n$ has an interleaving $w$ with~$f_{2n}$ that
  falls in~$u^*$, then in~$w$ we must intuitively insert one $a$ from~$v$ in
  each $u_{-a}$ and one $b$ from~$v$ in each $u_{-b}$, so that $v = (ab)^n$.
  To formalize this, we first rotate
  $u$ to ensure that its first and last letters are
  different. We then observe that, as $w$ is in~$u^*$,
  any factor $w'$ of length~$\card{u}$ of~$w$
  must be such that $|w'|_a = |u|_a$ and $|w'|_b = |u|_b$. We then
  consider factors of~$w$ of length~$\card{u}$ centered on the
  $u_{-a}$ and $u_{-b}$ in~$f_{2n}$: we argue that in~$w$ we must have inserted
  at least one~$a$ in or around each~$u_{-a}$, and at least one~$b$ in or around
  each~$u_{-b}$, otherwise these
  factors do not have enough $a$'s and enough $b$'s.
\end{proofsketch}

\begin{proof}
  Fix $u \in A^*$ and the two witnessing letters $a$ and~$b$. We first make a
  straightforward preliminary observation: for any word $w$ of~$u^*$ and factor $z$ of $w$
  such that $\card{z} = \card{u}$, we must have $|z|_a = |u|_a$ and $|z|_b =
  |u|_b$. Indeed, when running $w$ through the obvious deterministic finite
  automaton for~$u^*$, we know that, while $z$ is read, the total number of
  $a$-transitions and $b$-transitions will be $|u|_a$ and $|u|_b$.

  We now write $u = xy$ such that the last letter of~$x$ is different from the
  first letter of~$y$; by assumption on~$u$, this is always possible. We can now
  write $u^* = \epsilon + x (u')^* y$, where $u' \colonequals yx$; this ensures
  that the first and last letters of~$u'$ are different.
  
  We now show that
  $(ab)^*$ shuffle-reduces to~$u^*$, by constructing a filter sequence $(f_n)$.
  To this end, we let $u'_{-a}$ be a word obtained by removing some~$a$ in~$u'$,
  and $u'_{-b}$ be defined likewise. 
  Now, to define the filter sequence, we
  first deal with odd numbers as in the proof of Proposition~\ref{prp:aabb}, by
  defining $f_{2n+1}$ for $n\in\NN$
  as something that can never be in~$u^*$ even when inserting $2n+1$ arbitrary
  symbols, e.g., $f_{2n+1}\colonequals a^{(2n+1)\times\card{u}+1}$, which is
  clearly computable in PTIME.

  For even numbers, we define $f_{2n} \colonequals
  x(u' u'_{-a} u' u'_{-b} u')^n y$ for $n\in\NN$: this is 
  clearly computable in 
  PTIME.
  We show that this is a filter sequence by
  picking $n \in \NN$ and letting $v$ be a word such that
  $\card{v} = 2n$.
  If $v = (ab)^n$, we can clearly interleave $v$ and $f_{2n}$ to obtain a word
  of~$u^*$ by inserting each $a$ of~$v$ in $u'_{-a}$ and each~$b$ of~$v$
  in~$u'_{-b}$. Conversely, for an interleaving of any word with~$f_{2n}$ to
  yield a word of~$u^*$, we know that we must at least insert one $a$ in or
  around each $u'_{-a}$, and one $b$ in or around each $u'_{-b}$. Specifically,
  consider a word~$v$ and 
  consider a candidate interleaving~$w \in (v \shuf f_{2n}) \cap u^*$ and assume by contradiction that there
  is a factor~$u'_{-a}$ of~$f_{2n}$ such that~$w$ did not insert any~$a$ from~$v$
  within this factor or adjacently to this factor (the case
  of~$u'_{-b}$-factors is symmetric). Now, consider the factor $w'$ of~$w$ that
  contains this factor~$u'_{-a}$, 
  the neighboring letter from the beginning or end of~$u'$ where we take one such
  letter which is not~$a$ (which is always possible by hypothesis on~$u'$), and
  all inserted elements (which by hypothesis are all $b$'s). The
  number of $a$'s in the factor~$w'$ is $|u'_{-a}|_a$, which is $|u|_a - 1$,
  but $\card{w'} \geq \card{u}$,
  so, by our preliminary observation, this is impossible because we had assumed
  that $w \in u^*$. 
  Hence, indeed, we must insert
  one $a$ in $u'_{-a}$ or adjacently to it, and likewise for the $u'_{-b}$: these
  insertions are distinct, and they use up all letters of~$v$, so for $f_{2n}
  \shuf v$ to intersect $u^*$ nontrivially, the only possibility is that $v = (ab)^n$. This
  shows that $(f_n)$ is indeed a filter sequence, and allows us to conclude by
  Theorem~\ref{thm:reduction} and Theorem~\ref{thm:abshuffle}.
\end{proof}

We can also use the shuffle reduction to show hardness for other languages, e.g., $(aa+bb)^*$:
\begin{propositionrep}
  \label{prp:aabb}
  Let $L \colonequals (aa+bb)^*$. The problem $\pCSh{L}$ (hence $\pCTS{L}$) is
  NP-hard.
\end{propositionrep}

\begin{proofsketch}
  We do again a shuffle reduction from~$(ab)^*$, with the filter
  sequence $f_{2n} = (ab)^n$. If a word $v$ with $\card{v} = 2n$ is such that
  $v \shuf f_{2n}$ intersects $(aa+bb)^*$ nontrivially, it must intuitively insert
  $a$'s and $b$'s in $f_{2n}$ alternatively, so it must be~$(ab)^n$.
  Note that a similar proof would also show hardness for the language $(a^i + b^j)^*$ for any
  choice of  $i, j \geq 2$.
\end{proofsketch}

\begin{proof}
  We show a shuffle reduction from $K \colonequals (ab)^*$ to $K' \colonequals
  (aa+bb)^*$, which concludes by Theorem~\ref{thm:reduction} and
  Theorem~\ref{thm:abshuffle}. We first define the filter sequence for even values of~$n$, and
  show correctness for them; then we explain how to handle the case of odd~$n$.

  For all even $n\in\NN$ we set $f_n \colonequals (ab)^n$,
  which is clearly computable in
  PTIME.
  Let us now show correctness.
  For the forward direction, it is clear that for every even $n \in \NN$, the
  only word of length $n$ of~$(ab)^*$ is $(ab)^{n/2}$ and we can interleave it
  with $f_n$ to form $(aabb)^{n/2}$.

  For the backwards direction, fix $n \in \NN$, take $v \in A^*$ such that
  $\card{v} = n$, and assume that $(v \shuf f_n) \cap K'$ contains some word~$w$.
  For each $1 \leq i \leq n$, consider the position where the $i$-th letter
  of~$f_n$ occurs in~$w$, and call $\phi_i$ the maximal factor of~$w$ which
  contains the $i$-th letter of~$f_n$ and consists only of occurrences of the
  same letter (i.e., is of the form $a^*$ or $b^*$). By definition, these
  factors must occur in~$w$ in the order $\phi_1 \leq \cdots \leq \phi_n$. Now,
  as any two consecutive letters of~$f_n$ are different, we know that the
  $\phi_i$ are disjoint (so we have~$\phi_1 < \cdots < \phi_n$) and that
  there is only one letter in each $\phi_i$ that was taken from~$f_n$, namely,
  the $i$-th letter of~$f_n$: the others were inserted from~$v$.
  Further, by definition of~$K'$, the $\phi_i$ must all be of even length. This
  means that $f_1$ contains at least one inserted $a$, that $f_2$ contains at
  least one inserted~$b$, etc. As we have $\card{v} = n$, this completely
  specifies~$v$, specifically as $n$ is even we must have $v = (ab)^{n/2}$. This
  is a word of $(ab)^*$, which concludes the backward direction and establishes
  correctness for even~$n$.

  There remains to define the filter sequence for odd numbers, i.e., $2n+1$ with $n
  \in \NN$. As there are
  no words of odd length in~$(ab)^*$, it suffices to define $f_{2n+1}$ to be
  something that can never be in~$K'$ even when inserting $n$ arbitrary symbols.
  For instance, we can take $f_{2n+1} \colonequals (ab)^{2n+2}$, which has the
  required property by a variant of the proof for the backward direction above.
  This concludes the proof of the proposition.
\end{proof}

We show a last result that does not use the shuffle reduction but an easy
consideration on the number of letter occurrences. This result will be useful in
Section~\ref{sec:dich}:

\begin{proposition}
  \label{prp:abb}
  The problem $\pCSh{(ab+b)^*}$ (hence $\pCTS{(ab+b)^*}$) is NP-hard.
\end{proposition}

\begin{proof}
  We describe an easy PTIME reduction from $\pCSh{(ab)^*}$ to $\pCSh{(ab+b)^*}$.
  Given an instance $I$, check if the number of $a$-labeled and
  $b$-labeled vertices is the same, and fail if it is not.
  Otherwise, then $I$ achieves a word of $(ab+b)^*$ iff it achieves one
  of~$(ab)^*$, because we must enumerate one $a$-labeled vertex with each
  $b$-labeled vertex.
\end{proof}

We believe that the shuffle reduction applies to many other languages, though
we do not know how to characterize them. In particular, we believe that the
following could be shown with the shuffle reduction, generalizing all the above hardness
results except Proposition~\ref{prp:aabb}:

\begin{conjecture}
  \label{con:fstar}
  Let $F$ be a finite language such that, for some letter $a \in A$, the
  language~$F$ contains no power of~$a$ but contains a word which contains~$a$.
  Then $\pCSh{F^*}$ is NP-hard.
\end{conjecture}

\section{Tractability Results}
\label{sec:sigma2}
\begin{toappendix}
  \label{apx:da}
  \subsection{Additional Explanations About $(ab)^* + A^*aaA^*$}
  \label{apx:notmonomial}
  Let $A \colonequals \{a, b\}$.
We first substantiate a claim made in the main text, namely:

\begin{claim}
The regular
language $K = (ab)^* + A^* a a A^*$ cannot be expressed as a union of
monomials.
\end{claim}

We have already mentioned that it is decidable to check if a given (regular)
language can be expressed as a union of monomials. We explain how this process
can be applied to~$K$ to prove the claim:

\begin{proof}
  It is shown in Theorem~8.7 of~\cite{pin-weil97} that a regular language $K$ can
  be expressed as a union of monomials
  (equivalently called ``languages of level 3/2'' in the statement of that
  result) if and only if the ordered syntactic monoid of~$K$ satisfies the
  \emph{profinite}
  identity:
  \begin{equation}
  \label{eqn:profinite}
    \text{For all~}x, y\in A^*\text{~having same content,~} x^\omega \geq x^\omega y x^\omega
    \end{equation}
   where ``$x$ and $y$ having the same content'' means that, for each letter $a\in A$, we
   have
   $|x|_a >0$ iff $|y|_a > 0$, and where $\omega$ denotes the idempotent
   power in the free profinite monoid (see~\cite{pin-weil97} for precise definitions).

  This can be rephrased in more elementary terms using
  the notion of \emph{syntactic order} $\leq_K$ induced by~$K$,
  which can be thought of as an ordered version of the Myhill-Nerode congruence.
  Formally, the order $\leq_K$ is defined as follows:
  for all $x, y \in A^*$, we have $x\leq_K y$ 
  iff for all $u,v\in A^*$, $uyv\in K$ implies $uxv\in K$.
  Equation~\ref{eqn:profinite} can then equivalently be rephrased to the following condition:
  for all words $x,y\in A^*$
with same content, and for all integers $n$ such that $x^n \leq_K x^{2n} \leq_K
x^n$, we have
$x^n\geq_K x^n y x^n$.

  For our choice of language $K$, we can show that this rephrased condition does
  not hold, by taking $x\colonequals ab$ and $y\colonequals bab$ and $n
  \colonequals 1$.
  Indeed, we have $(ab)^1 \leq_K (ab)^2 \leq_K (ab)^1$, but the right-hand-side
  of the implication is wrong: we have $x^1 = ab$ in $K$, so we can take $u = v =
  \epsilon$ in the definition of the syntactic order, however we then have
  $x^1yx^1=abbabab$ which is not in~$K$, so we have shown that 
  $x^n\not \geq_K x^n y x^n$.  
\end{proof}

\end{toappendix}

Having shown hardness for several languages,
we now present our tractability results.
We will also rely on some of these results to show our
coarser dichotomy result in the next section.
\subparagraph*{Closure under union.}
The first observation on tractable languages is that they
are closed under union, as follows (recalling the 
definition of CTS and CSh for language \emph{families}):

\begin{lemma}
 \label{lem:closeunion}
  For any finite family of languages $\calK$,
  there is a logspace reduction
  from $\pCTS{\bigcup \calK}$ to $\pCTS{\calK}$, and likewise from
  $\pCSh{\bigcup \calK}$ to $\pCSh{\calK}$.
\end{lemma}

\begin{proof}
  To solve a problem for the language~$\bigcup \calK$ on an input instance~$I$, simply enumerate the
  languages $K' \in \calK$, and solve the problem on~$I$ for each~$K'$. Clearly
  $I$ is a positive instance of the problem for~$\bigcup \calK$ iff~$I$ is a positive
  instance of the problem for one of the~$K'$.
\end{proof}

\begin{corollary}
  \label{cor:closeunion}
  For any finite family of languages $\calK$, if $\pCTS{K'}$ is in NL
  for each~$K' \in \calK$, then so is $\pCTS{\bigcup \calK}$. The same is true
  of the CSh-problem.
\end{corollary}

Clearly, tractability is also preserved under the
\emph{reverse operator}, i.e., reversing the order of words in a language;
however tractable languages are \emph{not} closed under many usual operators, as
we will show in Section~\ref{sec:dich}. Still, closure under union will often
be useful in the sequel.

\subparagraph*{Monomials.}
We will now show that CTS is tractable for an important family of languages (and
unions of such languages): 
the \emph{monomial} languages.
Having fixed the alphabet $A$, a
\emph{monomial} is a language 
of the form $A_1^*a_1A_2^*a_2\cdots a_{n}A_{n+1}^*$ with $a_i \in A$ and $A_i
\subseteq A$ for all~$i$. 
In particular, we may have $A_i = \emptyset$ so that $A_i^* = \epsilon$: hence,
for every word $u\in A^*$, the language $A^* u A^*$ is a monomial language,
which intuitively tests whether a word contains the pattern~$u$.
Several decidable algebraic and logical characterizations of these languages are
known; in particular, unions of monomials are exactly the languages that are
definable in the first-order logic fragment $\Sigma_2[<]$ of formulas with
quantifier prefix~$\exists^* \forall^*$, and it is decidable
to check if a regular language is in 
this class~\cite{pin-weil97,papermansemigroup}.
We show:

\begin{toappendix}
  \subsection{Proof of Theorem~\ref{thm:monomial}: Tractability for Monomials}
\end{toappendix}

\begin{theoremrep}
  \label{thm:monomial}
  For any monomial language $K$, the problem
  $\pCTS{K}$ is in NL.
\end{theoremrep}

\begin{proofsketch}
  Let $K$ be $A_1^*a_1A_2^*a_2\cdots A_n^*a_{n}A_{n+1}^*$.
  We can first guess in NL the vertices $v_1, \ldots, v_{n}$ to which the
  $a_1, \ldots, a_{n}$ are mapped, so all
  that remains is to check, for each such guess, whether we can match the
  remaining vertices to the~$A_i$. We
  proceed by induction on~$n$. The base case of~$n=0$ (i.e., $K = A_1^*$) is
  trivial. For the induction step,
  using the fact that NL = co-NL 
  (see~\cite{immerman1988nondeterministic,szelepcsenyi1988method}),
  we check that the descendants of the last element~$v_{n}$
  are all in~$A_{n+1}^*$, and then we compute the set $S$ of vertices that
  \emph{must} be enumerated before~$v_{n}$: they are the ancestors of
  the~$v_i$, and the ancestors of any vertex labeled by a letter
  in~$A\setminus A_{n+1}$. We then use the induction hypothesis to check in NL
  whether $S$ has
  a topological sort that achieves a word in~$A_1^* a_1 \ldots A_{n-1}^* a_{n-1} A_n^*$.
\end{proofsketch}

\begin{proof}
  Let $K$ be $A_1^*a_1A_2^*a_2\cdots A_n^*a_{n}A_{n+1}^*$.
  First, we can guess in NL the vertices $v_1, \ldots, v_n$ of~$G = (V, E,
  \lfun)$ to which the
  $a_1, \ldots, a_n$ are associated, and verify that indeed we have $\lfun(v_i) =
  a_i$ for all~$1 \leq i \leq n$. Hence, up to making such a guess and
  relabeling the vertices, we can assume without loss of
  generality what we call the \emph{fresh pivot} condition on the input
  $A$-DAG: for each $a_i$ in our target language, there is exactly one $v_i$ in
  the input instance such that $\lfun(v_i) = a_i$.

  We now prove by induction on~$n$ that,  for any monomial $K = A_1^* a_1 \cdots
  A_n^* a_n A_{n+1}^*$, given an input $A$-DAG satisfying the fresh pivot
  condition, we can decide in NL whether $A$ has a topological sort
  satisfying~$K$.

  The base case of $n=0$ is trivial because $K$ is of the form $A_1^*$: we
  simply check if all element labels are in~$A_1^*$.
  For the induction step on $n+1$, let $K = A_1^* a_1 A_2^* \cdots a_{n+1} A_{n+2}^*$ and
  $K' = A_1^* a_1 A_2^* \cdots a_{n} A_{n+1}^*$. Let $G = (V, E,
  \lfun)$ be the input $A$-DAG satisfying the fresh pivot condition, and let $v_1,
  \ldots, v_{n+1}$ be the uniquely defined vertices matched to~$a_1, \ldots,
  a_{n+1}$.
  We define the sub-$A$-DAG $G'$ to be the restriction of~$G$ on the following
  vertex set~$V'$:
  \begin{itemize}
    \item the ancestors of the $v_1, \ldots, v_n$, including $v_1, \ldots, v_n$;
    \item the ancestors of~$v_{n+1}$ except $v_{n+1}$ itself;
    \item for each $w$ incomparable
      to~$v_{n+1}$ such that $\lfun(w) \notin A_{n+2}$, the ancestors of~$w$
      (including itself).
  \end{itemize}
  We now claim the following:
  \subparagraph*{Claim.} $G$ is a positive
  instance to~$K$ iff all descendants $z$ of $v_{n+1}$ are such that $\lfun(z) \in
  A_{n+2}$ and $G'$ is a positive instance to~$K'$.\\
  
  Note that $G'$ is always
  computable in NL, and the condition on the descendants of~$v_{n+1}$ can be
  checked in co-NL, hence in NL by the Immerman-Szelepcs\'enyi theorem
  \cite{immerman1988nondeterministic,szelepcsenyi1988method}.
  Hence, once this claim is proved, we have an NL algorithm for
  $\pCTS{K}$ by running the NL algorithm on the descendants of~$v_{n+1}$ and
  running the algorithm given by the
  induction hypothesis on~$G'$, which has been implicitly computed in NL.

  What remains is to prove the claim.
  For the backward direction, if the condition of the claim is respected, then
  we build the topological sort~$\sigma$ of~$G$ satisfying~$K$ by concatenating
  the topological sort~$\sigma'$ of~$G'$ satisfying~$K'$ which exists by
  assumption, the vertex $v_{n+1}$ which achieves $a_{n+1}$, and any topological
  sort of~$G \setminus (G' \cup \{v_{n+1}\})$. We must argue that this a
  topological sort. Indeed, observe first that the condition of the claim and
  the fresh pivot condition ensures that no descendant of~$v_{n+1}$ has a label
  in $a_1, \ldots, a_n$, i.e,. $v_{n+1}$ is not an ancestor of any~$v_i$; in
  particular $v_{n+1}$ is not in~$V'$. However, by definition of~$G'$, all ancestors of~$v_{n+1}$
  are in~$V'$. So we know that we can indeed concatenate~$\sigma'$, $v_{n+1}$,
  and a topological sort of the remaining elements of~$G$, and the result~$\sigma$ is
  indeed a topological sort of~$G$. We now argue
  that~$\sigma$ achieves~$K$: this is because $\sigma'$ achieves~$K'$, $v_{n+1}$
  achieves $a_{n+1}$, and by assumption all remaining vertices are either
  descendants of~$v_{n+1}$ so their label is in~$A_{n+2}$, or they are
  incomparable to~$v_{n+1}$ so their label must be in~$A_{n+2}$ (they would be
  in~$G'$ otherwise). Thus, $\sigma$ is a topological sort of~$G$ that
  achieves~$K$, establishing the backward implication.

  For the forward direction, 
  consider a topological sort $\sigma$ of~$G$ that
  achieves~$K$. Thanks to the fresh pivot condition, we know that $v_{n+1}$ is
  matched to~$a_{n+1}$. Let $U$ be the elements enumerated before~$v_{n+1}$
  in~$\sigma$, and let~$\sigma'$ be the topological sort induced by~$\sigma$
  on~$U$:
  we know that $\sigma'$ satisfies~$K'$. We now claim that $V' \subseteq U$.
  Indeed, first, by the fresh pivot condition, $\sigma'$ must enumerate $a_i$ for all $1
  \leq i \leq n$, so $v_1, \ldots, v_n$ and their ancestors must be in~$V'$.
  Second, as $\sigma$ enumerates $v_{n+1}$ just after $\sigma'$, we know that
  $\sigma'$ must enumerate all ancestors of~$v_{n+1}$ except $v_{n+1}$ itself. Third, assuming by
  way of contradiction that $V'$ does not contain an ancestor of a vertex $w$
  incomparable to~$v_{n+1}$ such that $\lfun(w) \notin A_{n+2}$, we would have
  that $V'$ does not contain $w$ either, and as $w$ is incomparable to~$v_{n+1}$
  it is different from~$v_{n+1}$ so $w$ must be enumerated after $v_{n+1}$
  by~$\sigma$, but $\lfun(w) \notin A_{n+2}$, which is impossible because we are matching
  elements to~$A_{n+2}^*$ after~$v_{n+1}$. So indeed $V' \subseteq U$. Further,
  as~$V'$ contains the $v_1, \ldots, v_n$, we know that the topological sort
  $\sigma''$ of~$V'$ defined as the restriction
  of~$\sigma'$ to~$V'$ also achieves~$K'$: intuitively, given a topological sort
  that achieves~$K'$, we can remove any elements except those matched to the
  $a_i$ and the result still achieves~$K'$. So $\sigma''$ witnesses that $G'$ is a
  positive instance to~$K'$. Now, as $\sigma$ must enumerate all descendants~$z$
  of~$v_{n+1}$ after $v_{n+1}$ which achieves $a_{n+1}$, we know that they must
  be such that $\lfun(z) \in A_{n+2}$, so we have shown the condition and
  established the forward implication.

  We have shown our claim, which concludes the proof of
  Theorem~\ref{thm:monomial}.
\end{proof}

\begin{toappendix}
  \subsection{Proof of Proposition~\ref{prp:trivwidth}: Tractability Based on
  Width}
  \label{apx:widthproof}
\end{toappendix}
\subparagraph*{Tractability based on width.}
While unions of monomials are a natural class, it turns out that they do not
cover all tractable languages. In particular, we can show: 

\begin{proposition}
  \label{prp:abaaaa}
  Let $A \colonequals \{a, b\}$ and $K \colonequals (ab)^* + A^*aaA^*$.
  The problem $\pCTS{K}$
  (hence $\pCSh{K}$) is in NL.
\end{proposition}

This result is not covered by Theorem~\ref{thm:monomial}, because we can show
that $K$ cannot be expressed as a union of monomials (see
Appendix~\ref{apx:notmonomial}); and the proof technique is different.

\begin{proof}
  Let $G$ be an input $A$-DAG. We first check in NL if~$G$ contains two
  incomparable vertices
  $v_1 \neq v_2$ such that $\lfun(v_1) = \lfun(v_2) = a$.
  If yes,
  we conclude that $G$ is a positive instance, as we can clearly achieve~$K$ by
  enumerating $v_1$ and $v_2$ contiguously.

  If there are no two such vertices, we check in NL if there are two comparable
  $a$-labeled vertices $v_1 \neq v_2$ that can be enumerated contiguously, i.e.,
  there is an edge $v_1 \rightarrow v_2$ but no vertex $w$ that is
  \emph{between}~$v_1$ and $v_2$, i.e., is a descendant of~$v_1$ and an ancestor
  of~$v_2$. If there are two such vertices $v_1$ and~$v_2$, we conclude again
  that $G$ is a positive instance.

  Otherwise, our first test implies that $G$ induces a total order on the 
  $a$-labeled vertices, and our second test
  implies that any two consecutive $a$-labeled vertices in this order must have
  at least one $b$-labeled vertex between them. This ensures that no topological
  sort achieves $A^* aa A^*$, so it suffices to test whether one can achieve
  $(ab)^*$. Clearly this is the case iff all consecutive pairs of~$a$-labeled
  vertices have exactly one $b$-labeled vertex between them, and there is
  exactly one additional $b$-labeled vertex that can be enumerated after the
  last~$a$-labeled vertex. We can test this in NL, which concludes the proof.
\end{proof}

Intuitively, the language of Proposition~\ref{prp:abaaaa}
is tractable because it is easy to solve unless the
input instance has a very restricted structure, namely, all $a$'s are
comparable.
We do not know whether this result generalizes to $(ab)^* + A^* a^i A^*$ for $i
> 2$. However, following the intuition of this proof, we can show the
tractability of a similar kind of regular languages:

\begin{propositionrep}
  \label{prp:trivwidth}
  Let $A \colonequals \{a, b\}$, let~$K'$ be a regular language, let $i \in
  \NN$, and let $K \colonequals K' + A^* (a^i + b^i) A^*$.
  The problem $\pCTS{K}$ (hence $\pCSh{K}$) is in NL.
\end{propositionrep}

As in
Proposition~\ref{prp:abaaaa}, CTS is trivial for the languages in this
proposition unless the input $A$-DAG $G$ has a
restricted shape. Here, the requirement is on the \emph{width} of~$G$, i.e., the
maximal cardinality of a subset of pairwise incomparable vertices (called an
\emph{antichain}), so we can show Proposition~\ref{prp:trivwidth}  by
distinguishing two cases depending on the width of~$G$:

\begin{proofsketch}
  We test in NL whether the input $A$-DAG $G$ contains an antichain $C$ of
  size~$2i$: if it does, then at least $i$ vertices in~$C$ must have the same label,
  and we can enumerate them in succession to achieve $A^* a^i A^*$ or $A^* b^i A^*$, so
  $G$ is a positive instance. Otherwise, $G$ has
  width~$< 2i$, and Dilworth's
  theorem~\cite{dilworth1950decomposition} implies that its elements can be
  partitioned into chains, so that CTS can be solved in NL following a dynamic algorithm
  on them. 
\end{proofsketch}

\begin{toappendix}
To show this result, we will need several preliminary definitions.
Recall from the main text that an \emph{antichain} is a
set $S \subseteq V$ of vertices which are pairwise incomparable, and
the \emph{width} of a DAG is the size of its largest antichain. The main claim
is then the following:

\begin{proposition}
  \label{prp:dynamic}
  For any regular language $K$, the problem $\pCTS{K}$ can be solved in space $O(k \log
  n)$, where $k$ is the width of the input DAG and $n$ is its total size.
  The same bound holds for $\pCSh{K}$ where $k$ is the number of
  input strings.
\end{proposition}

  We note that a similar task was already known to be in PTIME
  by~\cite[Theorem~17]{amarilli2017possible}, but showing the space bound given
  here will introduce several additional technicalities.
  From Proposition~\ref{prp:dynamic},
  it is easy to show Proposition~\ref{prp:trivwidth}:

  \begin{proof}[Proof of Proposition~\ref{prp:trivwidth}]
  We follow the proof sketch: we test in NL if the input DAG contains an
  antichain of size~$2i$. If it does, as explained in the sketch, the DAG is a
  positive instance. So the only remaining case is when the input DAG has width
  $\leq 2i$, so we can conclude by Proposition~\ref{prp:dynamic}.
\end{proof}

  So all that remains is to show Proposition~\ref{prp:dynamic}, which we do in
  the rest of Appendix~\ref{apx:widthproof}.
  The high-level idea of the proof is to use Dilworth's
theorem~\cite{dilworth1950decomposition}, which essentially shows that the
  \emph{width} of any DAG~$G$ is equal to the minimal cardinality of a
  \emph{chain partition} of $G$, i.e., a partition of~$G$ into disjoint chains,
  where we may additionally have arbitrary edges between the chains.
  We will then perform a logspace algorithm
  following such a partition to guess an accepting path of a 
  (fixed) automaton for~$K$.

We present the complete proof of Proposition~\ref{prp:dynamic} in the rest of Appendix~\ref{apx:widthproof}.
We first define formally the notion of chain partition. 
Let $G = (V, E)$ be a DAG, and let $(V, E')$ be its transitive closure.
A \emph{chain partition} of~$G$
is a partition $V_1 \sqcup \cdots \sqcup V_n$ of~$V$, such
that, for all $1 \leq i \leq n$, the restriction of $E'$ to $V_i \times V_i$ is
the transitive closure of a directed path graph: equivalently, for each pair of
vertices $v \neq v'$ of~$V_i$, either $v$ has a directed path to~$v'$ in~$E$ or $v'$ has
a directed path to~$v$ in~$E$.
We call each of the~$V_i$ a \emph{chain} of~$G$.
Note that, in addition to the edges between vertices of the chain, there may
also be arbitrary edges in~$G$ between $V_i$ and $V_j$
for $i \neq j$.
The \emph{width} of a chain partition is the number of chains that it
contains. The following is then known from partial order theory:

\begin{theorem}[\cite{dilworth1950decomposition}]
  \label{thm:dilworth}
  For any DAG $G$, the width of~$G$ is~$k$ iff there exists a chain partition
  of width~$k$ of~$G$.
\end{theorem}

However, to show our desired space bound,
we need to look closely into the complexity of computing
a chain partition. This task is known to be in PTIME \cite{fulkerson1955note}
but we are unaware of an existing proof to show that it can be done in NL.
Because of this, we must give a custom scheme to compute implicitly a specific chain partition
that meets our logspace requirements. One difficulty will be to ensure that, as
we compute the chain partition implicitly in NL, we are always looking at the
same chain partition each time we recompute it implicitly (i.e., we are not
looking at some random chain partition that was nondeterministically chosen for
this implicit computation). To fix a canonical
choice of chain partition, we look at the \emph{minimal} one in an order that we
will define.

We will see a width-$k$ chain partition as a labeling function $\chi$ from~$V$
to~$\{1, \ldots, k\}$ such that, letting $V_i \colonequals \{v
\in V \mid \chi(v) = i\}$, then $V_1 \sqcup \cdots \sqcup V_k$ is indeed a chain
partition.
Given a DAG $(V, E)$, 
the vertices of~$V$ are integers, each of them represented in binary by a
sequence of size $\log
\card n$, and we
let~$<$ denote the corresponding total order relation on~$V$.
We can then talk
about the topological sort $\sigma$, equivalently seen as a total order
$<_\sigma$, which is minimal according to the lexicographic order defined
by~$<$: namely, $\sigma$ is constructed by picking, at each step, the smallest
possible vertex according to~$<$ which can be picked (i.e., it has not been
picked yet, but all its ancestors
have): we write the
vertices of~$V$ in the order of~$<_\sigma$ as $v_1 < \cdots < v_{\card{V}}$.
We then
lift the total order~$<_\sigma$ on~$V$ to a total order relation on chain
partitions:
we write each chain partition as the word $\chi(v_1) \cdots \chi(v_{\card{V}})$, 
and $<_\sigma$ defines an order on the chain partitions given by the
lexicographic order on words of~$\{1, \ldots, k\}^{\card{V}}$.
Now, we can talk
about the chain partition $\chi_0$ which is \emph{minimal} according to this
total order relation $<_\sigma$ on chain partitions. We will explain how this minimal chain partition
can be computed implicitly in logspace. Again, the reason why we are concerned about minimality is
simply to ensure that, when using the implicitly-computed chain partition
within our logspace algorithm for $\pCTS{K}$, then the chain partition that we
follow is well-defined, i.e., it is the same over all calls to the implicit
nondeterministic logspace chain partition oracle. The specific definition of
minimality that we use does not matter much.

We now describe the specific implicit representation 
that we want for the minimal chain partition $\chi_0$.
We want to show that we can evaluate efficiently two functions: one function
$\nxt$, which takes as input a vertex $v\in V$ and returns the \emph{next vertex} of
the chain of~$v$ in~$\chi_0$, and one function $\fst$, which
takes as input a chain number $1 \leq i \leq k$ and returns the \emph{first
vertex} of~$i$ in~$\chi_0$.
Formally, $\nxt(v)$ for $v \in V$ is defined as follows:
letting $c \colonequals \chi_0(v) \in \{1,
\ldots, k\}$ be the chain to which $v$ belongs in~$\chi_0$, return 
the vertex $\nxt(v)\in V$ which is the successor of~$v$ on chain~$c$ in~$\chi_0$, if
any, or $\top$ if $v$ is the last vertex of chain~$c$. More formally,
$\nxt(v)$ is the
vertex of~$V$ such that $\chi_0(\nxt(v)) = c$,
the edge $v \rightarrow \nxt(v)$ is in~$E'$,
and
there is no $z \in V$ such that $\chi_0(z) = c$ and the edges
$v \rightarrow z$ and $z \rightarrow \nxt(v)$ are in~$E'$.
As for the function $\fst$, for any
chain number $1 \leq c \leq k$, we let $\fst(c)$ be the first
element of the chain~$c$ in~$\chi_0$, that is, 
we have $\chi_0(\fst(c)) = c$ and there is no $z \in V$ such that $\chi_0(z) =
c$ and the edge $z \rightarrow \fst(c)$ is in~$E'$.
We can now claim:
\begin{lemma}
  \label{lem:nextfun}
  For any input to the functions $\nxt$ and $\fst$, we can evaluate them
  in space $O(k \log n)$.
\end{lemma}

We will do two things in the sequel: prove this lemma, and use it to prove
Proposition~\ref{prp:dynamic}. To do this, we need to define the notion of a
\emph{configuration}, which will be useful in our algorithms on chain partitions. A \emph{configuration} is a $k$-tuple $X =
(v_1, \ldots, v_k)$, where each $v_i$ is either an element of~$V$ or~$\bot$.
Intuitively, $X$ describes the lowest element of each chain, with $\bot$
indicating that no element has been assigned to this chain so far; when we
consider a configuration $X$ in an algorithm, we assume that the
\emph{ancestors} of~$X$, meaning all vertices $w$ such that for some~$v_i$ the edge $w
\rightarrow v_i$ is in~$E'$,
have already been assigned to a chain in some fashion. We say that a
configuration $X$ is \emph{continuable} if there exists a chain partition
$\chi$ which is consistent with~$X$, meaning that $\chi(v_i) = i$ for all $1
\leq i \leq k$ such that $v_i \neq \bot$. One useful lemma will be the
following:

\begin{lemma}
  \label{lem:continuation}
  There is an algorithm to decide, given a configuration $X$, whether $X$ is
  continuable, in space $O(k \log n)$.
\end{lemma}

We will first show how to use this lemma to prove Lemma~\ref{lem:nextfun}. We
will then explain how to prove Lemma~\ref{lem:continuation}. Last, we will prove
Proposition~\ref{prp:dynamic} from Lemma~\ref{lem:nextfun}.

We start by proving Lemma~\ref{lem:nextfun}.
The intuition is that we can use the  continuation check of
Lemma~\ref{lem:continuation} as a way to compute implicitly the minimal
chain partition, by considering all vertices in the minimal topological sort
$<_\sigma$, and assigning each vertex to the smallest possible
chain such that the resulting configuration is continuable. Formally, we show:

\begin{proof}[Proof of Lemma~\ref{lem:nextfun}]
  We maintain a configuration $X = (v_1, \ldots, v_k)$, initially $(\bot,
  \ldots, \bot)$, and extend it \emph{deterministically} at each step using
  the (nondeterministic) oracle for continuation checking described in
  Lemma~\ref{lem:continuation}. Specifically, at each step of the algorithm,
  we call $S$ the set of vertices which are ancestors of elements in~$X$, and
  we consider the vertex $v$ which is as small as possible according
  to~$<_\sigma$
  and which is not in $S$ but all its strict ancestors are in~$S$: we can
  find this vertex in NL. Now, for each $1 \leq i \leq k$ such that the edge
  $v_i \rightarrow v$ is in~$E'$ or $v_i = \bot$, we check whether the configuration $X_i$ obtained by
  replacing $v_i$ by~$v$ is continuable. We pick the smallest $i$ such that it
  is, and continue the algorithm with $X_i$: specifically, we guess a suitable
  $i$, and guess in co-NL that there is no $i' < i$ which is suitable: this is
  still in NL overall, thanks to the Immerman-Szelepcs\'enyi theorem
  \cite{immerman1988nondeterministic,szelepcsenyi1988method}. At the end of the process, we
  have memorized the successor of the vertex of interest on its chain (i.e.,
  the input to~$\nxt$), or the first vertex of the chain of interest
  (i.e., the input to~$\fst$), and we return this.

  We will soon explain why the algorithm does not get stuck, in the sense
  that, for each vertex~$v$ that we consider, there is a choice of~$i$ for
  which the conditions are respected. However, notice first that, if the
  algorithm does not get stuck, then the algorithm considers all vertices of~$V$
  exactly once, following the order $<_\sigma$ of the minimal topological sort. Indeed, at each step, the set $S$ contains all
  vertices that have been seen so far: the only thing to notice is that,
  whenever we remove a vertex $z$ from the configuration, we replace it by a
  vertex $z'$ such that all of its ancestors are in~$S$ and $z$ is an ancestor
  of~$z'$, so that the new value of $S$ becomes $S \cup \{z\}$. This ensures
  that we are indeed picking at each step the next vertex that $<_\sigma$ has
  picked.

  We now explain why the algorithm does not get stuck, which we show by
  induction. Initially, the configuration is $(\bot, \ldots, \bot)$, and this
  configuration is continuable, as we know by Dilworth's theorem
  (Theorem~\ref{thm:dilworth}). Now, at each step of the algorithm, the
  current configuration $X$ is continuable by induction
  hypothesis, because it was chosen to be continuable at the previous step of the algorithm.
  Now, as $X = (v_1, \ldots v_k)$ is continuable, letting $\chi$ be a
  witnessing chain partition, letting $v$ be the next vertex that we consider,
  we know that, if $v_{\chi(v)} = \bot$, then we can take $i = \chi(v)$. If
  $v_{\chi(v)} \neq \bot$, then $v_{\chi(v)}$ must be an ancestor of~$v$ in
  chain~$i$, and by the condition on the ancestors of~$v$, we know that $v$
  must be the first descendant of $v_{\chi(v)}$ on the chain, justifying that
  the edge $v_{\chi(v)} \rightarrow v$ must exist in~$E'$. Hence, $\chi$ witnesses that the
  algorithm does not get stuck.

  Last, we argue that the values computed by the algorithm are correct. To do
  so, we show by induction that all choices performed by the algorithm
  actually follow $\chi_0$, in the sense that, at each step of the algorithm,
  the current configuration is consistent with $\chi_0$, and, for each vertex
  $v$ that we consider, we take $i \colonequals \chi_0(v)$. We do this by
  mutual induction on these two claims.
  The base case is trivial because $(\bot, \ldots, \bot)$ is
  of course consistent by~$\chi_0$. Now, assuming consistency of the
  configuration, as $\chi_0$ is defined to be minimal following~$<_\sigma$, by
  minimality of the vertex~$v$ picked by both $<_\sigma$ and the algorithm,
  we know that $\chi_0(v)$ is the minimal value
  such that the resulting configuration is continuable. Indeed, if it were not,
  then by taking a smaller continuable value, and taking any
  witnessing continuation afterwards, we would obtain a chain partition which
  would be smaller in the lexicographic order, contradicting the minimality
  of~$\chi_0$. So we have shown that our algorithm actually computes $\nxt$ and
  $\fst$ following $\chi_0$, proving the result.
\end{proof}

We now come back to the proof of Lemma~\ref{lem:continuation}:

\begin{proof}[Proof of Lemma~\ref{lem:continuation}]
  The proof follows similar ideas as in Lemma~\ref{lem:nextfun}: we have a
  current configuration, we consider the vertices following a topological
  order, and we try to assign them to a chain, updating the configuration. The
  only difference is that, instead of assigning the minimal chain number
  following a continuation check, we simply nondeterministically guess a chain
  to which we assign them. When the nondeterministic guesses succeed, we can
  show exactly as in Lemma~\ref{lem:nextfun} (but without worrying about
  minimality) that these guesses witness the
  existence of a chain partition which is consistent with the input
  configuration $X$, so that $X$ is indeed continuable; and conversely,
  whenever such a chain partition exist, these is a sequence of
  nondeterministic guesses which make the algorithm succeed.
\end{proof}

Thanks to Lemma~\ref{lem:nextfun}, we now know that we can implicitly compute
the minimal chain partition within the prescribed time bounds. We are now
ready to prove Proposition~\ref{prp:dynamic}:

\begin{proof}[Proof of Proposition~\ref{prp:dynamic}]
  We fix an automaton $\calA$ for the regular language $K$: remember that, as $K$ is
  fixed, we can compute $\calA$ in constant time, and the size of its state set
  $Q$ and transition relation $\delta \subseteq Q \times A \times Q$ is constant.

  Our state at any stage of the algorithm will consist of a
  \emph{configuration}. Remember that this is a $k$-tuple $X = (v_1, \ldots, v_k)$ such that each $v_i$ is either $\bot$
  or an element of~$V$, which intuitively codes the lowest element for each
  chain, or $\bot$ if no element of the chain has been seen so far: initially
  the configuration is $(\bot, \ldots, \bot)$. The state also contains one state
  $q \in Q$ of the automaton, which is initially some initial state, chosen
  nondeterministically.

  At each stage of the algorithm, we nondeterministically guess one
  chain $1 \leq i \leq k$ to extend. We then replace the current configuration
  $X$ with the new configuration $X_i$ defined as follows: if $v_i = \bot$, then
  we replace $v_i$ in~$X_i$ by $v_i' \colonequals \fst(v_i)$; if $v_i \neq \bot$, then we replace
  $v_i$ in~$X_i$ by~$v_i' \colonequals \nxt(v_i)$ if it is different from~$\top$; otherwise we
  cannot choose this value of~$i$. We also cannot choose a value of~$i$ when
  the $v_i'$ that we have defined cannot be enumerated yet, i.e., if it is not
  the case that all strict ancestors of~$v_i'$ are in $X$ or are ancestors
  of vertices in~$X$. Once we have made an appropriate choice for~$i$, we also replace the current state $q$ with
  some element $q'$ such that $(q, \lfun(v_i'), q') \in \delta$, nondeterministically chosen.
  Intuitively, this means that the automaton processes the letter which is the
  label of the new element $v_i'$ which is read along the chain~$i$.

  The algorithm concludes when we can no longer perform a step, meaning that
  $v_i \neq \bot$ and $\nxt(v_i) = \top$ for each $1 \leq i \leq k$. Then, the
  algorithm accepts if the current state $q$ is final.

  It is clear that, whenever the algorithm succeeds, then the sequence of
  guesses witnesses the existence of a topological sort of~$G$, obtained
  following the vertices that are chosen at each step: the definition of the
  steps that we perform ensure that this sequence indeed respects the edge
  relation of~$G$, for similar reasons as in the proof of
  Lemma~\ref{lem:nextfun}. Conversely, whenever there is a witnessing
  topological sort, then we can decompose it along the minimal chain
  partition~$\chi_0$ defined earlier. Specifically, the sequence of vertices
  given by this topological sort can be expressed as a sequence of operations
  where we enumerate the first vertex of a chain, or enumerate the next vertex
  of a chain from the preceding one. The definition of the algorithm ensures
  that these steps can be mimicked by a sequence of nondeterministic guesses (in
  particular, following these guesses, the algorithm does not ``get stuck'' and
  can always pick the right $v_i'$ at each step),
  and likewise the accepting path in the automaton can be mimicked by
  nondeterministic choices of the states in the transition relation. This
  establishes the correctness of the algorithm, and concludes the proof.
\end{proof}

\end{toappendix}

\begin{toappendix}
  \subsection{Proof of Proposition~\ref{prp:aab}: Other Tractable Case}
\end{toappendix}

\subparagraph*{Other tractable case.}
We close the section with another example of a regular language which is
tractable for the CSh-problem for what appears to be a unrelated reason.

\begin{propositionrep}
  \label{prp:aab}
  Let $A \colonequals \{a, b\}$ and $K \colonequals (aa+b)^*$.
  The problem $\pCSh{K}$ is in NL.
\end{propositionrep}

This is in contrast to $(aa+bb)^*$, for which we showed intractability
(Proposition~\ref{prp:aabb}). We do not know the complexity of the CTS-problem
for $(aa+b)^*$, or the
complexity for either problem of languages of the form~$(a^i+b)^*$ for $i>2$.

\begin{proofsketch}
  We show that the existence of a suitable topological sort can be rephrased to
  an NL-testable equivalent condition, namely, there is no string in the input
  instance whose number of odd ``blocks'' of~$a$-labeled elements dominates the
  total number of $a$-labeled elements available in the other strings. If the condition
  fails, then we easily establish that no suitable topological sort can be
  constructed: indeed, eliminating each odd block of~$a$'s in the dominating
  string requires one~$a$ from the other strings.
  If the condition holds, we can simplify the input strings and show that a
  greedy algorithm can find a topological sort by picking pairs of $a$'s in the
  two current heaviest strings.
\end{proofsketch}

\begin{proof}
  We can first check in NL whether the total number of $a$-elements is even; if
  not, clearly there is no suitable topological sort, so in the sequel we assume
  that it is.

  Note that, if any string consists only of~$b$'s, then we can clearly enumerate
  these $b$'s first, and the result is equisatisfiable; so without loss of
  generality we can always remove
  any input string that consists only of $b$'s as soon as they appear, so we
  never consider such strings.
  Now, if there are less than 3 input strings, then we can conclude in NL by
  Proposition~\ref{prp:dynamic}, so we assume that there are at least 3 strings
  in the input instance (which contain some~$a$ by the assumption that we
  just made).

  Given an input instance $I$ to the CSh-problem for~$K$,
  we call a \emph{block} a maximal contiguous
  sub-sequence of $a$-labeled elements in a string, and call it an \emph{even} or
  \emph{odd} block depending on the number of such elements. The
  \emph{$a$-weight} of a string is its total number of $a$-labeled elements, and
  the \emph{$a$-alternation} of a string is its total number of odd $a$-blocks.

  We claim that $I$ does \emph{not} have a topological sort satisfying~$K$ if
  and only if there is a string whose $a$-alternation is greater than the sum of
  the $a$-weights of all other strings. This condition can clearly be checked in
  NL: compute the maximal $a$-alternation of a string, and compute the $a$-weight
  of the other strings and compare. Hence, all that remains is to show this
  condition.

  The easy direction is the backward one. If there is a string $C$ whose
  $a$-alternation is greater than the sum of the $a$-weights of all other
  strings, we know that any topological sort satisfying $K$ must enumerate one
  element of every odd block of~$C$ together with an $a$-element of another
  string of~$C$: indeed, when enumerating two $a$-labeled elements from $C$, they
  must be in the same block because of the $b$-elements between blocks, so this
  cannot change the parity of a block of~$C$. Hence, under our assumption, a
  topological sort would have to enumerate more $a$-elements in the other
  strings
  than their total $a$-weight, which is impossible; this concludes the backward
  direction.

  To show the forward direction, we show the contrapositive: if, for each string
  $C$, the $a$-alternation of~$C$ is no greater than the total $a$-weight of the
  other strings (which we call assumption (*)),
  then there exist a suitable topological sort.

  We first make a simplifying observation. Given an instance $I$, for any choice
  of two contiguous $a$-elements in a string of~$I$, we let $I'$ be the result of
  removing these two elements. If $I'$ has a suitable topological sort, then so
  does~$I$, because we can just mimic the topological sort on~$I$ and enumerate
  the two adjacent $a$-elements when they become available. Hence, to show that
  there is a suitable topological sort, we can always decide to remove any two
  contiguous $a$'s in a block (even at a stage where they are not available).
  We call this a
  \emph{simplification}. Note, however, that we cannot apply this simplification
  blindly, as the converse implication to the above does not hold in general
  (consider $\{ababa, aaa\}$ vs $\{ababa, a\}$).

  We will define a second assumption (**), and show two things: (i.) given any input
  instance satisfying (*) with an even number of~$a$'s and with at least 3
  strings (containing some $a$), we can rewrite it through simplifications (and
  removal of strings containing only~$b$) to an instance satisfying
  (**), and (ii.) that given an instance satisfying (**) and our preliminary
  assumptions, we can build a suitable
  topological sort. Condition (**) says: for each string $C$, the
  \emph{$a$-weight} of
  $C$ is no greater than the total $a$-weight of the other strings. (Notice the
  difference with (*).)

  We first show (ii.): under our preliminary assumptions on~$I$, any instance satisfying (**) has a suitable topological
  sort. We do so by describing a greedy algorithm which enumerates elements in a
  way that achieves a suitable topological sort. Namely:

  \begin{enumerate}
    \item If we can enumerate a $b$-element, then enumerate it.
    \item Otherwise, pick the two strings whose non-enumerated elements have
      largest $a$-weight and enumerate one $a$ from each of these two strings.
  \end{enumerate}

  If this algorithm does not get stuck, then it clearly constructs a topological
  sort satisfying~$K$. Now, the only way for this algorithm to get stuck is if
  there is only one string left, but this would violate (**).
  Hence, it suffices to show that the algorithm preserves assumption (**).
  Clearly step~1 preserves it, so we focus on step~2. By assumption (**) there
  are at least two strings left: if there are exactly two strings left, then
  condition (**) is preserved as the $a$-weight of both strings is decreased. Assume
  now that there are at least three strings left before applying step~2, and let
  $C, C', C''$ be the strings with the largest $a$-weight (in terms of
  unenumerated elements) and let
  $n \geq n' \geq n''$ be their respective $a$-weights. After step~2, the
  $a$-weights are
  $n-1$, $n'-1$, and $n''$. It is clear that, as condition (**) held of~$C$
  and~$C'$ before step~2, then the condition still holds, as the $a$-weight of each of
  these two strings and the total $a$-weight of the other strings has been decremented,
  then condition (**) still holds of these strings. We must show that it holds of
  the other strings, and clearly it suffices to focus on~$C''$, which has the
  largest $a$-weight in terms of non-enumerated elements. There are three cases,
  depending on the relationship of~$n''$ to~$n$.

  \begin{itemize}
    \item If $n'' < n-1$, then as (**) is still satisfied for~$C$ after the step
      and the $a$-weight of $C''$ is still smaller than~$C$ after the step, then (**) is satisfied
      for~$C''$ too.
    \item If $n'' = n-1$, then after performing the step, $C$ and $C''$ have
      same $a$-weight, and it is obvious that if condition (**) holds of a string~$C$
      then it holds of a string with the exact same $a$-weight (as the
      $a$-weight of the two
      strings is the same, and so is the $a$-weight of the other strings).
    \item If $n'' = n$, then we have $n'' = n' = n$. Now, the only problematic
      case would be if, after performing the step, $n''$ were strictly greater
      than the $a$-weight of all other strings, in particular, we would have $n'' >
      (n-1) + (n'-1)$. But substituting in this inequality we get $n > 2n-2$,
      hence $n <2$. Hence, the only bad situation is when all strings have
      $a$-weight at most~1, but then, remembering that the number of~$a$'s was
      initially even and clearly remains even throughout the enumeration, we
      have at least 2 strings left after the step in this case that all have 
      $a$-weight exactly~$1$, so condition (**) is always
      respected.
  \end{itemize}

  Hence, we have shown that, on any input instance satisfying condition (**) in
  addition to our preliminary requirements, the above algorithm succeeds and
  produces a suitable topological sort.

  The only thing left to show is (i.): given an instance satisfying (*) and our
  preliminary requirements, in particular that of having at least 3 strings
  containing an $a$-element, then we can rewrite it using simplifications to an
  instance satisfying (**). To do so, let us observe that, for any string with
  $a$-alternation $n$ and $a$-weight $m$, we can clearly perform simplifications
  to rewrite it to a string of $a$-weight $p$ for any value $n \leq p \leq m$ of
  the same parity as~$m$ (or of $n$, as $m$ and $n$ have same parity). So let us
  simplify the string $C$ with the greatest $a$-alternation to make its $a$-weight
  equal to its $a$-alternation $n$, and let us rewrite all strings in the
  following way: if the string has $a$-weight $\leq n+1$, we do not change it;
  otherwise we simplify it to $n$ or $n+1$ depending on the parity of its
  $a$-weight.
  Let us show that the result of this transformation satisfies assumption (**).
  Consider a string $C'$ and show the condition. If $C' = C$, then $C$ has
  $a$-weight~$n$, and thanks to condition (*) we know that the sum of $a$-weights are
  greater than~$n$, because the only case where we have reduced the $a$-weight
  of another string $C''$ than $C$ was to bring it down to $n$ or $n+1$, in which case
  $C''$ suffices to witness that (**) is satisfied for~$C$. If $C'$ is different from $C$,
  then its greatest possible $a$-weight is $n+1$ by construction, however, we
  know that $C$ achieves $a$-weight $n$, and thanks to the assumption that we
  have at least 3 strings containing $a$'s, we know that there is another string
  containing some~$a$, hence (**) holds for~$C'$. This establishes that (**) now
  holds after the simplifications, which concludes the proof.
\end{proof}

\section{A Coarser Dichotomy Theorem}
\label{sec:dich}
In the two previous sections, we have established some intractability and
tractability results about the constrained topological sort and constrained
shuffle problems for various languages. Remember that our end goal would be to
characterize the tractable and intractable languages, and show a dichotomy
(Conjecture~\ref{con:maincon}). This is difficult, and one reason is that the
class of tractable languages is not ``well-behaved'': while it is closed under
the union operator (Corollary~\ref{cor:closeunion}), it is 
is \emph{not} closed under intersection, complement, and other common operations. This
makes it difficult to study tractable languages using algebraic
language theory~\cite{Pin97a}.

\begin{toappendix}
  \subsection{Proofs of Proposition~\ref{prop:notclose}: Closure Counterexamples}
  \label{apx:closure}

\end{toappendix}

\begin{propositionrep}\label{prop:notclose} We have the following
  counterexamples to closure:
  \begin{itemize}
  \item \textbf{Quotient.}   There exists a word $u \in A^*$
    and a regular language $K$ such that $\pCSh{K}$ is in NL but $\pCSh{u^{-1}K}$
    is NP-hard. \label{prp:cshquotient}%
  \item \textbf{Intersection.}
  There exists two regular languages $K_1$ and $K_2$ such that $\pCTS{K_1}$ and
  $\pCTS{K_2}$ are both in PTIME but $\pCSh{K_1 \cap K_2}$ is NP-hard
  \item \textbf{Complement.}  There exists a regular language $K$ such that $\pCTS{K}$ is in NL,
  but $\pCSh{A^* \setminus K}$ is NP-hard.

  \item \textbf{Inverse of morphism.} There exists a regular language $K$ and
  morphism $\phi$ such that $\pCTS{K}$ is
  in NL but $\pCSh{\phi^{-1}(K)}$ is NP-hard.
  \end{itemize}
\end{propositionrep}

The three last results
of this proposition
also apply to the
constrained topological sort problem, but the first one does not, and in fact
$\CTS$-tractable languages \emph{are} closed under quotients. This observation
implies that there are regular languages~$K$ such that
$\pCSh{K}$ is tractable but $\pCTS{K}$ is NP-hard; one concrete example is $K
\colonequals b^* A^* + a a A^* + (ab)^*$
(see~Appendix~\ref{apx:closure}). We sketch the proof of
Proposition~\ref{prop:notclose}:

\begin{proofsketch}
  For each operation, we use $(ab)^*$ as our NP-hard
  language (by Theorem~\ref{thm:abshuffle}).
  
  For quotient, we take
  $K \colonequals b A^* + aa A^* + (ab)^*$, and $u \colonequals ab$. We have
  $u^{-1} K =
  (ab)^*$,
  but $\pCSh{K}$ is in NL because
  any shuffle instance with more than one string satisfies~$K$.
  
  For intersection, we take $K_1 \colonequals (ab)^*(\epsilon + bA^*)$
  and $K_2 \colonequals (ab)^*(\epsilon + aaA^*)$. We have $K_1 \cap K_2 = (ab)^*$, but
  $\pCSh{K_1}$ and $\pCSh{K_2}$ are in PTIME using an ad-hoc greedy
  algorithm.
  
  For complement,
  we take $K \colonequals bA^* \cup A^* a \cup A^* aa A^* \cup A^* bb A^*$. As~$K$ is a
  union of monomials, we know by Theorem~\ref{thm:monomial} that $\pCTS{K}$ is in
  NL, but we have $A^* \setminus K = (ab)^*$.
  
  For inverse of morphism,
  we take $A \colonequals \{a, b\}$ and $K \colonequals (ab)^* + A^* (a^3+b^3) A^*$.
  We know that $\pCTS{K}$ is in PTIME by Proposition~\ref{prp:trivwidth}.
  Now, defining $\phi:A^* \to A^*$ by $\phi(a) \colonequals aba$
  and $\phi(b) \colonequals bab$, we have $\phi^{-1}(K) = (ab)^*$ because no word
  in the image of~$\phi$ has three identical consecutive symbols.
\end{proofsketch}

\begin{toappendix}
  
  First, we show that tractable languages for $\CSh$ are not closed under left
  quotient.
  Recall that the \emph{left quotient} of a language $K$ by a word $u \in A^*$
  is the language $u^{-1}K \colonequals \{v \in A^* \mid uv \in K\}$; right
  quotients are defined analogously. We only consider left quotients, but of
  course the same result holds for right quotients because both our problems are
  symmetric under the reverse operator:
  \begin{proposition}
  There exists a word $u \in A^*$
  and a regular language $K$ such that $\pCSh{K}$ is tractable but
  $u^{-1}K = (ab)^*$, so that $\pCSh{u^{-1}K}$ is NP-hard by
   Theorem~\ref{thm:abshuffle}.
  \end{proposition}

\begin{proof}
  Take $A \colonequals \{a, b\}$ and $K \colonequals b A^* + aa A^* + (ab)^*$.
  Take $u \colonequals ab$. It is clear that $u^{-1} K = (ab)^*$. However,
  $\pCSh{K}$ is tractable by the following reasoning. Consider an input instance
  to $\pCSh{K}$. If there is a string that starts with~$b$, then we can clearly
  always construct a topological sort achieving $bA^*$. Hence, we can assume
  that all strings start with~$a$. If there are two strings or more, by taking
  their first letters, we can clearly always construct a topological sort
  achieving $aaA^*$. Hence, we can assume that there is only one string, and
  we can clearly check in NL whether the only possible topological sort
  achieves~$K$.
\end{proof}

  However, we point out that the tractable languages for the $\CTS$-problem
  \emph{are} closed under quotient:

  \begin{propositionrep}
    \label{prp:ctsquotient}
    For any word $u \in A^*$ and regular language $K$, there is an logspace reduction from
    $\pCTS{u^{-1}K}$ to $\pCTS{K}$.
  \end{propositionrep}
  
Thus, for the language
$K \colonequals b^* A^* + a a A^* + (ab)^*$
used in the proof of Proposition~\ref{prp:cshquotient}, we
know that $\pCTS{K}$ is NP-hard but $\pCSh{K}$ is in NL: hence, $K$ separates
the problems CSh and CTS.

  \begin{proof}[Proof of Proposition~\ref{prp:ctsquotient}]
    Fix $u \in A^*$ and $K$. Given an $A$-DAG $G$, to solve $\pCTS{u^{-1}K}$
    on~$G$, construct the DAG $G'$ obtained by adding a directed path of elements whose
    label is~$u$ and adding edges from each element of the directed path to all
    elements of~$G$. It is
    obvious that there is a topological sort of~$G'$ achieving~$K$ iff there is a
    topological sort of~$G$ achieving $u^{-1}K$, which concludes the proof.
  \end{proof}

Second, we illustrate that tractable languages are not closed under the
intersection operator, for both problems:

\begin{proposition}
  \label{prp:intersection}
  There exists two regular languages $K_1$ and $K_2$ such that $\pCTS{K_1}$ and
  $\pCTS{K_2}$ are both in PTIME, but $K_1 \cap K_2 = (ab)^*$, so that $\pCSh{K_1 \cap
  K_2}$ is NP-hard by Theorem~\ref{thm:abshuffle}.
\end{proposition}
  
Note that we do not show that $\pCTS{K_1}$ and $\pCTS{K_2}$ are in NL, although we conjecture that this should hold.

\begin{proof}[Proof of Proposition~\ref{prp:intersection}]
  We fix $A \colonequals \{a, b\}$ and
  take $K_1 = (ab)^*(\epsilon + bA^*)$
  and $K_2 = (ab)^*(\epsilon + aaA^*)$.
  It is clear that $K_1 \cap K_2 = (ab)^*$, so we only need to show that
  $\pCTS{K_1}$ and $\pCTS{K_2}$ are tractable.
  Now, observe that $a^{-1}K_1 b^{-1} = (ba)^*(\epsilon + bb A^*)$, which is the
  result of swapping the symbols $a$ and~$b$ in~$K_2$. Hence, if we establish
  that $\pCTS{K_1}$ is in PTIME, then by Proposition~\ref{prp:ctsquotient}, as
  PTIME-membership is clearly preserved by renaming the symbols, we have also shown
  that $\pCTS{K_2}$ is in PTIME. So we focus on~$K_1$.

  We will show a greedy algorithm in PTIME to solve $\pCTS{K_1}$, and explain why it
  succeeds. The algorithm has two states:

  \begin{itemize}
    \item State $\alpha$ (the initial state), where:
      \begin{itemize}
        \item being out of symbols means that we have succeeded, i.e., we have
          constructed a topological sort in~$(ab)^*$;
        \item enumerating an $a$ allows us to move to state $\beta$;
        \item enumerating a $b$ allows us to
      ``win'', i.e., that we can continue the topological sort in any way and
      remain in~$K_1$.
      \end{itemize}
    \item State $\beta$, where:
      \begin{itemize}
        \item being out of symbols means that we have failed, i.e., the word
          that we have formed is of the form $(ab)^* a$ and not in~$K_1$;
        \item enumerating an~$a$ is forbidden;
        \item enumerating a~$b$ allows us to move back to state~$\alpha$.
      \end{itemize}
    \end{itemize}

  We accordingly design the algorithm as follows:

  \begin{itemize}
    \item In state~$\beta$:
      \begin{itemize}
        \item if there is an available~$b$, enumerate any of them and move to
          state~$\alpha$;
        \item otherwise fail.
      \end{itemize}
    \item In state~$\alpha$:
      \begin{itemize}
        \item if there is an available~$b$, enumerate it and succeed;
        \item otherwise, if there is an available~$a$ such that, when
          enumerating this~$a$, there is an available~$b$ (call this a
          \emph{profitable}~$a$), then enumerate any
          one of these~$a$'s and move to state~$\beta$;
        \item otherwise, if there are no symbols left, succeed;
        \item otherwise fail.
      \end{itemize}
  \end{itemize}

  If the algorithm succeeds, then it clearly builds a suitable topological sort,
  hence we have to argue for the other direction: if there is a suitable
  topological sort then the algorithm will find it. To do so, we must justify
  that the choices made by the algorithm are without loss of generality, i.e.,
  any suitable topological sort can be modified to follow the rules of the
  algorithm, so as to witness that the algorithm succeeds.

  Let us thus consider a witnessing topological sort $\sigma$, and consider the first
  point at which $\sigma$ disagrees with the actions of the
  algorithm, and change $\sigma$ to continue like the
  algorithm did and still achieve~$K$: we can then repeat the argument until
  $\sigma$ is exactly what the algorithm does, which allows us to conclude that
  the algorithm succeeds. When the algorithm did the choice that
  disagrees with~$\sigma$, either it was in state~$\alpha$ or in state~$\beta$;
  note that
  if the algorithm had already decided that it had succeeded, then there is
  nothing left to show as indeed the topological sort is suitable no matter how
  it continues.

  If the algorithm was in state~$\beta$, as $\sigma$ is suitable, there must be
  an available~$b$. If there
  is only one available~$b$, then the algorithm and the topological sort cannot
  disagree, so the only thing to consider is the case where the algorithm picks
  one $b$-labeled element~$v$ and $\sigma$ picks another~$v'$. However, in this
  case, we can modify~$\sigma$ to pick~$v'$ and then pick~$v$ (which is
  available), and this ensures that~$\sigma$ succeeds immediately, so it is
  still suitable. So we have modified 
  modify~$\sigma$ to do like the algorithm does (and succeed immediately).

  If the algorithm was in state~$\alpha$, if there is an available~$b$, then we can
  always modify~$\sigma$ to take it and succeed.
  Likewise, if there is no available symbol, then $\sigma$ and the algorithm are
  both finished and both succeed. Hence, the only possible disagreement is
  if~$\sigma$ picks a different~$a$ than what the algorithm chose, of
  if~$\sigma$ picked an unprofitable~$a$ whereas the algorithm failed. However,
  note that, as $\sigma$ is a suitable topological sort, it cannot pick an
  unprofitable~$a$, as it will necessarily be stuck afterwards (only $a$'s, if
  anything, will be available, and we will be in state~$\beta$), so the second
  case is impossible by our assumption that~$\sigma$ is suitable. So the only case to consider is
  the first case, and we will explain how to modify~$\sigma$ to pick immediately
  the profitable~$a$ that the algorithm enumerates (call it $v$), followed by
  the~$b$ that the algorithm enumerates after~$v$ (call it $w$).
  
  To do this, consider the later moment at which $\sigma$ enumerates
  $v$. It is possible that, when $\sigma$ enumerates~$v$, it has already succeeded (meaning, there were
  two contiguous $b$'s after an~$a$ earlier in~$\sigma$), but in this case there is no constraint
  on~$\sigma$ and we can simply move $v$ and $w$ in~$\sigma$ to enumerate them at the
  moment the algorithm does, and~$\sigma$ is still suitable. If $\sigma$ has not
  already succeeded when it enumerates~$v$, then either $\sigma$ enumerates $w$ just after~$v$ (subcase~1), or it
  does not (subcase~2). If it does (subcase~1), then we can modify~$\sigma$ by moving $v$ and $w$ to the
  beginning: $\sigma$ is still a topological sort after this change (indeed we
  can enumerate $v$ and $w$ because the algorithm does it, and for the other
  elements there is nothing to show), and $\sigma$ is still suitable (we have
  added an $ab$-factor at the beginning, and removed an $ab$-factor in what
  follows but this has no effect on the fact that $\sigma$ realizes~$K_1$).
  Now, if $\sigma$ does not enumerate~$w$ immediately after~$v$ (subcase 2), then let $w'$ be the
  element that $\sigma$ enumerates instead; it must be $b$-labeled (by our
  assumption that $\sigma$ has not already succeeded when it enumerates~$v$).
  But we know from what the algorithm does that $\sigma$ can only enumerate $w$
  after having enumerated $v$, and not before it has enumerated~$v$, so $\sigma$
  must enumerate $w$ somewhere after~$w'$. We modify $\sigma$ to enumerate~$w$
  instead of~$w'$ and enumerate $w$ immediately after: this is still a
  topological sort as we just explained, and $\sigma$ is still suitable (no
  matter what happens afterwards)because
  now it enumerates two consecutive $b$'s (namely, $w$ and $w'$) immediately after
  an~$a$ (namely, $v$). We are now back to
  subcase~1, because $\sigma$ now enumerates~$w$ just after~$v$, so we can conclude as in that
  subcase.
  This concludes the correctness proof.

  Note that the algorithm described here is not in NL; we conjecture that
  $\pCTS{K_1}$ is in NL, but we do not know how this can be shown.
\end{proof}

Third, we show that tractable languages are not closed under
complement:

\begin{proposition}
  There exists a regular language $K$ such that $\pCTS{K}$ is in NL,
  but $A^* \setminus K = (ab)^*$, so that $\pCSh{A^* \setminus K}$ is NP-hard by
  Theorem~\ref{thm:abshuffle}.
\end{proposition}

\begin{proof}
  Take $K = bA^* \cup A^* a \cup A^* aa A^* \cup A^* bb A^*$. As~$K$ is a
  union of monomials, we know by Theorem~\ref{thm:monomial} that $\pCTS{K}$ is in
  NL, however by construction we have $A^* \setminus K = (ab)^*$.
\end{proof}

Fourth, we show that tractable languages are not closed under \emph{inverse
morphisms}. Recall that a \emph{morphism} from alphabet~$B$ to alphabet~$A$
is a function $\phi: B^*\to A^*$ such that $\phi(uv) = \phi(u) \phi(v)$ for all $u,
v \in B^*$; note that a morphism is completely defined by the image of each
letter of~$B$.
The \emph{inverse image} of a language $K$ over alphabet~$A$ by a morphism $\phi$ is the
language over alphabet~$B$ defined by $\phi^{-1}(K) \colonequals \{v \in B^* \mid
\phi(v) \in K\}$. We show:

\begin{proposition}
  There exists a regular language $K$ and morphism $\phi$ such that $\pCTS{K}$ is
  in NL, but $\phi^{-1}(K) = (ab)^*$, so that $\pCSh{\phi^{-1}(K)}$ is NP-hard by
  Theorem~\ref{thm:abshuffle}.
\end{proposition}

\begin{proof}
  We take $A \colonequals \{a, b\}$ and $K \colonequals (ab)^* + A^* (a^3 + b^3) A^*$, as
  in Proposition~\ref{prp:trivwidth}. We know by this proposition that $\pCTS{K}$ is
  in NL. However, let $\phi:A^* \to A^*$ be defined by $\phi(a) \colonequals aba$
  and $\phi(b) \colonequals bab$. We then have $\phi^{-1}(K) = (ab)^*$, for
  which the CSh-problem is NP-hard by Theorem~\ref{thm:abshuffle}. Indeed,
  no word in the image of~$\phi$ has three consecutive $a$'s or three
  consecutive $b$'s, so $\phi^{-1}(K) =
  \phi^{-1}((ab)^*)$, and this is clearly equal to~$(ab)^*$.
\end{proof}

\end{toappendix}

Proposition~\ref{prop:notclose} suggests that tractable languages would be
easier to study algebraically if we ensured that they were closed under all
these operations, i.e., if they formed a variety~\cite{Pin97a}. In this section,
we enforce this by moving to an alternative phrasing of the CTS and CSh
problems. This allows us to leverage algebraic techniques and show a dichotomy
theorem in this alternative phrasing, under an additional \emph{counter-free}
assumption. We first present the alternative phrasing, and then present the
additional assumption and our dichotomy result.

\subparagraph*{Alternative phrasing.}
The first change in our alternative phrasing is that the input DAG~$G$ will now be an
\emph{$A^*$-DAG}, i.e., a DAG labeled with words of~$A^*$ rather than letters
of~$A$. As before, a
topological sort $\sigma$ of~$G$ \emph{achieves} a word $\lfun(\sigma) \in A^*$
obtained by concatenating the $\lfun$-images of the vertices of~$G$ in the order
of~$\sigma$: but vertex labels are now ``atomic'' words
whose letters cannot be interleaved with anything else.
The \emph{multi-letter} CTS and CSh problems are the variants defined with
$A^*$-DAGs; intuitively, this ensures that tractable languages are  closed
under inverse morphisms.

The second change is that we will not fix one single target language, but a
\emph{semiautomaton}~\cite{holcombe_1982}, i.e.,
an automaton where initial and final states are not
specified. Formally, a semiautomaton is a tuple $(Q,A,\delta)$ where $Q$ is the set of states,
$A$ is the alphabet, and $\delta:Q\times A \to Q$
is the transition function; we extend $\delta$ to words as usual by setting $\delta(q, \epsilon)
\colonequals q$ and $\delta(q, u_1 \cdots u_{n+1})
\colonequals \delta(\delta(q, u_1), u_2 \cdots u_{n+1})$.
We will fix the target semiautomaton, and the initial and final states will be
given in the input instance (in addition to the DAG). This 
enforces closure
under quotients (by
choosing the initial and final states) and complement (by toggling the final
states). Further, to impose closure under intersection, the input
instance will specify a \emph{set} of pairs of initial-final states, with a
logical AND over them. The question is to determine whether the input DAG
achieves a word accepted by \emph{all} the corresponding automata; and this enforces
closure under intersection.

We can now summarize the formal definition of our problem variants.
The \emph{multi-letter $\CTS$-problem for a fixed semiautomaton}
$S=(Q,A,\delta)$ takes as input an $A^*$-DAG and
a set $\{(i_1,F_1),\ldots,(i_k,F_k)\}$ of initial-final state pairs, where $i_j\in Q$ and $F_j
\subseteq Q$ for all $1 \leq j \leq k$.
The input is accepted if there is a topological sort $\sigma$ of~$G$ such that,
for all $1 \leq j \leq k$,
the word~$\lfun(\sigma)$ is accepted by the automaton $(Q,A,\delta,i_j,F_j)$,
i.e., $\delta(i_j, \lfun(\sigma)) \in F_j$. The
\emph{multi-letter $\CSh$-problem for a fixed semiautomaton} is defined in the same way, imposing that the input
$A^*$-DAG is a union of directed path graphs.

\myparagraph{Dichotomy result}
Our dichotomy will apply to the multi-letter CTS and CSh problem for
semiautomata. However, we will need to make an additional assumption, namely,
that the semiautomaton is
\emph{counter-free}. This assumption means that our dichotomy will only apply
to a well-known subset of regular languages, namely, the \emph{star-free
languages}, that are better understood algebraically; it excludes in particular
the tricky case of \emph{group languages} that we will
study separately in Section~\ref{sec:group}.
Formally, a semiautomaton is \emph{counter-free} if, 
for every state $q$ and word $u \in A^*$, if $\delta(q,u^n)= q$ for some $n>1$, then we have
$\delta(q, u) = q$.
Under the counter-free assumption, we can prove the following dichotomy, using our hardness
and tractability results in Sections~\ref{sec:hardness} and~\ref{sec:sigma2}:

\begin{toappendix}
  \subsection{Proof of Theorem~\ref{thm:dich}: Coarser Dichotomy Theorem}
  Recall from the main text the definition of the \emph{transition monoid} $T(S)$
of a semiautomaton $S$.
We call \emph{transition morphism} the
morphism $\eta:A^*\to T(S)$ defined by $\eta(u)= f_u$ for all $u \in A^*$: by
construction, this morphism is surjective. Recall that our counter-free
assumption on~$S$ is equivalent to requiring that $T(S)$ is an \emph{aperiodic
monoid}: formally, it 
for all $x \in T(S)$, we have
$x^\omega=x^{\omega+1}$, where $\omega\in\NN$
is the \emph{idempotent power} of~$M$, i.e.,
the least integer $\omega\in\NN$ such that for every element $x$ in $M$, we have
$x^\omega=x^{2\omega}$.

Our characterization of tractable semiautomata in Theorem~\ref{thm:dich} is
based on the class \DA of monoids~\cite{tesson2002diamonds}, which is a subset of \A.
A monoid $M$ is in \DA iff it satisfies the equation
$(xy)^\omega x (xy)^\omega = (xy)^\omega$ for all $x, y \in M$, where $\omega$ again refers to the
idempotent power of~$M$; this implies in particular that $M$ is aperiodic.
Our dichotomy result relies on the following
characterization of \DA:

\begin{theorem}[(\cite{tesson2002diamonds}, Theorem~5 and Theorem~11)]
  \label{thm:DA}
  Let $K$ be a regular language of $A^*$. The following conditions are
  equivalent:
  \begin{itemize}
    \item $K$ is an union of \emph{unambiguous} monomials, i.e., of monomials
      $K = A_1^* a_1 \cdots A_n^* a_n A_{n+1}^*$ such that every word $u \in K$
      has a unique decomposition $u_1 a_1 \cdots u_n a_n u_{n+1}$ where $u_i \in
      A_i^*$ for all $1 \leq i\leq n+1$.
  \item There exists a monoid $M$ in \DA and a morphism $\phi:A^*\to M$ such that
    $K$ is \emph{recognized} by~$M$, meaning that 
   $K = \phi^{-1}(P)$ for some subset $P \subseteq M$.
  \end{itemize}
\end{theorem}

We will also rely on a characterization of monoids that are \emph{not} in \DA:

\begin{proposition}[(\cite{tesson2001}, Lemma 10)]
  \label{prop:DAab}
  An aperiodic monoid $M$ is not in \DA iff there exists a morphism
  $\theta:\{a,b\}^*\to M$ and $P\subseteq M$
  such that $\theta^{-1}(P)$ is either $(ab)^*$ or $(ab + b)^*$.
\end{proposition}

\begin{proof}
  This result follows from \cite{tesson2001}, Lemma~10, but the latter result is
  presented in slightly different terminology. Specifically, that result
  states that an aperiodic monoid
  is not in \DA iff it is \emph{divided} by two monoids $\mathit{BA}_2$
  and~$U$, that are respectively the syntactic monoid of $(ab)^*$
  and $(ab+b)^*$ (up to relabeling the symbols of Figure~2 of~\cite{tesson2001}).
  A monoid $N$ \emph{divides} another monoid $M$ iff there exists a submonoid $K$ of $M$ such that
$N$ is a quotient of~$K$.
  Our lemma follows from this result thanks to the well-known fact that a language $K$ is recognized
by a monoid $M$ iff its syntactic monoid {divides}
  $M$: see~\cite[Theorem V.1.3]{straubing1994}.
\end{proof}

  We are now ready to state and prove our dichotomy theorem:
\end{toappendix}

\begin{theoremrep}
  \label{thm:dich}
  Let $S$ be a counter-free semiautomaton. Then 
  the multi-letter CSh-problem and CTS-problem for $S$ are either both in NL, or
  both NP-complete. The dichotomy is effective: given $S$, it is PSPACE-complete
  to decide which case applies.
\end{theoremrep}

We conclude the section by introducing some technical tools used for this result
and for Section~\ref{sec:group}, and by giving a proof sketch. The criterion of
the dichotomy on~$S$ is phrased in terms of the \emph{transition monoid}
of~$S$, which we now define (see, e.g., \cite{Pin97a} for details). Remember that a monoid is a set that has an associative
binary operation and a neutral element. The \emph{transition monoid} $T(S)$ of a semiautomaton~$S =
(Q, A, \delta)$
is the set of functions $f: Q\to Q$ that are ``achieved'' by $S$ in the
following sense: there is a word $u \in A^*$ such that $\delta(q, u) = f(q)$ for
all~$q \in Q$. In particular, the neutral element is the identity function,
which is achieved by taking $u \colonequals \epsilon$; and the binary operation
on~$T(S)$ is function composition, which is associative. Note that the
transition monoid is finite and can be computed from~$S$.

We assumed that $S$ is counter-free, and this is
equivalent~\cite{mcnaughton-papert71} to saying that $T(S)$ is in the class \A
of \emph{aperiodic} finite monoids (formally defined by the
equation $x^{\omega + 1} = x^\omega$ where $\omega$ is the idempotent
power~\cite{Pin97a} of the monoid). Within \A, our dichotomy criterion on~$T(S)$
is based on a certain subclass of~\A, called \DA
(see~\cite{tesson2002diamonds}): 
$S$ is tractable iff $T(S)$ is in \DA, and it is
PSPACE-complete~\cite{therien1998over} to test whether this holds (using the
formal
definition of \DA by the equation $(xy)^\omega x (xy)^\omega= (xy)^\omega$).
We can now sketch the proof of 
Theorem~\ref{thm:dich}:

\begin{proofsketch}
  We first show that if $T(S)$ is in \DA then the multi-letter CTS and CSh
  problems for~$S$ are in NL. For this, we rely on one characterization of \DA
  (from~\cite{tesson2002diamonds}): if $T(S)$ is in~\DA then the regular languages
  recognized by~$S$ (for any set of initial-final states) are unions of
  \emph{unambiguous monomials}, in particular they are unions of monomials,
  so we have tractability by Corollary~\ref{cor:closeunion} and Theorem~\ref{thm:monomial}.

  For the converse direction, we use a second characterization of \DA
  (from~\cite{tesson2001}): if $T(S)$ is \emph{not} in \DA then there is a
  choice of initial-final state pairs for which $S$ computes a language~$K$ 
  whose inverse image by some morphism is either $(ab)^*$ or $(ab+b)^*$. We know
  that these languages are intractable (Theorem~\ref{thm:abshuffle} and
  Proposition~\ref{prp:abb}) so we conclude by showing a PTIME reduction from
  one of these two languages: this is
  possible in our alternative problem phrasing, in particular using the
  multi-letter labels to invert the morphism.
\end{proofsketch}

\begin{toappendix}
\begin{proof}
\newcommand{\SL}{\textrm{SL}}
  Fix the input semiautomaton~$S$.
  We wish to show that 
the multi-letter CTS-problem is tractable for~$S$
  iff the transition monoid $T(S)$
  of~$S$ is in \DA.
  We call $\SL(K)$ the set of possible languages that can be defined from $S$
  depending on the input instance, namely, depending on the set $\{(i_1, F_1),
  \ldots,  (i_k, F_k)\}$ of pairs of initial and final states.
For one direction we prove that: (a) 
  if $T(S)$ is in \DA, then
  for any language $K$ in $\SL(S)$,
  the multi-letter CTS-problem for~$K$ is in NL.
For the converse direction we prove that:
  (b) if $T(S)$ is not in \DA, then there exists a language $K$ in $\SL(S)$ whose
multi-letter CSh-problem is NP-complete, so we can show NP-hardness by
  restricting to input instances that use this language.

\textbf{Proof of (a).} Assume that
  $M \colonequals T(S)$ is in \DA. We denote by
$\eta:A^*\to M$ the transition morphism of~$S$ 
  and by $\psi:M^*\to M$ the morphism on words over the alphabet~$M$ defined by
  $\psi(m) \colonequals m$ for all $m\in M$. Intuitively, applying $\psi$ to a
  sequence of elements of~$M$ simply evaluates the sequence in~$M$.

Let $I = (G,(i_1,F_1),\ldots,(i_k,F_k))$ be an instance of the semiautomaton CTS-problem
  and let $K_j$ be the language recognized by the automaton
$(Q,A,\delta,i_j,F_j)$ for all $1\leq j\leq k$.
  We must determine whether $G = (V, E, \lfun)$ has a topological sort in~$K \colonequals \bigcap_j K_j$. We will
reduce this to our original definition of the CTS-problem for regular languages,
  with a language that we know to be in NL. 
  Specifically, we will work on the alphabet $M$ of the transition monoid, and
  the language that we will use is $K' \colonequals \{u \in M^* \mid \psi(u) =
  \eta(K)\}$. In other words, $K' = \psi^{-1}(\eta(K))$, so $K'$ is recognized
  by~$M$ which is a monoid in \DA: by Theorem~\ref{thm:DA}, we know that $K'$ is
  a union of monomials.
  
  Our goal is then to reduce to $\pCTS{K'}$. Formally,
  we construct from the $A^*$-DAG $G = (V, E, \lfun)$
  the $M$-DAG $G' = (V, E, \lfun')$ where we define $\lfun'(v) \colonequals
  \eta(\lfun(v))$ for all $v \in V$. Intuitively, we have relabeled the
  multi-letter labels of~$G$ to single-letter labels in~$M$. We claim that
  $I$ is a positive instance to the CTS-problem for~$S$
  iff $G'$ is a positive instance to $\pCTS{K'}$. This will allow us to conclude,
  because, by Theorem~\ref{thm:monomial} and Corollary~\ref{cor:closeunion},
  we know that $\pCTS{K'}$ is in NL.

  To show the equivalence, we will show that for any topological sort $\sigma$ of
  $(G, V)$, the word $\lfun(\sigma)$ achieved by~$\sigma$ in $G$ is in $K$ iff
  the word $\lfun'(\sigma)$ achieved by~$\sigma$ in~$G'$ is in~$K'$. In other words,
  letting $w_1 \cdots w_n \colonequals \lfun(\sigma)$, we must show that $w_1 \cdots
  w_n \in K$ iff $\eta(w_1) \cdots \eta(w_n) \in K'$.
  The forward direction is immediate by applying the morphism~$\eta$.
  For the backward
  direction, we have $\psi(\eta(w_1) \cdots \eta(w_n))) \in
  \eta(K)$, and the left-hand-side is $\eta(w_1) \cdots \eta(w_n)$, which is
  $\eta(w_1 \cdots w_k)$ because $\eta$ is a morphism, so applying $\eta^{-1}$
  concludes. We have shown the equivalence, so we can reduce in NL to $\pCTS{K'}$
  with $K'$ a union of monomials, which establishes NL-membership.
  
\textbf{Proof of (b).}
Assume that $T(S)$ is not in \DA. Remember that $T(S)$ is still aperiodic
because $S$ is counter-free.
Hence, we can apply Proposition~\ref{prop:DAab}:
there exists a morphism $\theta:\{a,b\}^*\to M$, a set $P\subseteq M$, 
and a regular language $H \in \{(ab)^*, (ab+b)^*\}$
such that $\theta^{-1}(P) = H$. 
Our goal is to use $\theta$ and $P$ to define a set of pairs of initial and
final states of~$S$ so that the CSh-problem for~$S$ with these states reduces in
logspace to the corresponding problem for~$H$. 
To do this, let $x \colonequals \theta(a)$ and $y \colonequals \theta(b)$.
As these are elements of the transition monoid, we can pick $u, v\in A^*$
such that $f_u = x$ and $f_v = y$, which we will use to define our reduction.

Let $G = (V, E, \lfun)$ be an instance of the CSh-problem for~$H$.
Let us build $G' = (V, E, \lfun')$ where we define $\lfun'(w) \colonequals
\theta(\lfun(w))$ for all $w\in V$.
For each function $f \in P$, let us define
an instance $I_f$ of the semiautomaton CSh-problem of $S$
by
$I_f=(G',(q_1,\{f(q_1)\}), \ldots, (q_n,\{f(q_n)\}))$
where $(q_i)_{i=1,\ldots,n}$ is an arbitrary enumeration of $Q$, the set of states
of $S$.
Note that a word $z\in A^*$ is accepted by~$S$ for the choice of initial and
final states in~$I_f$ iff $f_z = f$ in~$M$.
This construction is in NL. Let us show that $G$ is a positive instance to
$\pCSh{H}$ iff one of the $I_f$ is a positive instance to
the semiautomaton CSh-problem of~$S$, which shows that our reduction is correct
(but note that this is not a many-one reduction). 

For the forward direction, assume that we have a topological sort $\sigma$ of~$(V,
E)$
achieving a word $z \colonequals \lfun(\sigma)$ of~$H$,
and let us consider the word $\lfun'(\sigma) = \theta(z_1) \cdots \theta(z_n) =
\theta(z_1 \cdots z_n)$ because $\theta$ is a morphism. As $z \in H$ and
$\theta(H) = P$, we know that
$f \colonequals \theta(z_1 \cdots z_n)$ is in~$P$.
Hence, consider the
instance $I_f$. We know that $f_z = f$ by definition, hence $\sigma$ witnesses
that $I_f$ has a suitable topological sort.

For the backward direction, assume that there is $f \in P$
such that we have a solution of $I_f$. This means that
there is a topological sort $\sigma$ of $(V, E)$ such that the word $z
\colonequals \lfun'(\sigma)$
achieved by~$\sigma$ in~$G'$ is such that $f_z = f$. Now, we know that
$\theta^{-1}(f) \subseteq H$. Hence, the word $\lfun(\sigma)$ achieved by~$\sigma$
in~$G$ is in~$H$, so $G$ is a positive instance to $\pCSh{H}$, which
establishes the desired equivalence.

We have thus shown a reduction from $\pCSh{H}$ to the multi-letter CSh-problem
for the semiautomaton~$S$. We can then conclude that the latter problem
is NP-hard, because $\pCSh{H}$ is NP-hard: either $H = (ab)^*$ and this follows
from Theorem~\ref{thm:abshuffle}, or $H = (ab+b)^*$, in which case
we conclude from Proposition~\ref{prp:abb}.
\end{proof}
\end{toappendix}

\section{Lifting the Counter-Free Assumption for CSh}
\label{sec:group}
\begin{toappendix}
  \label{apx:group}
\end{toappendix}
Our dichotomy theorem in the previous section (Theorem~\ref{thm:dich})
was shown for an alternative
phrasing of our problems (with semiautomata and multi-letter inputs), and made
the additional assumption that the input semiautomaton is counter-free. In this section, we
study how to lift the counter-free assumption. In exchange for this, we
restrict our study to the constrained shuffle problem (CSh) rather than CTS.

To extend Theorem~\ref{thm:dich} for the CSh-problem, we will again classify
the semiautomata~$S$
based on their transition monoid $T(S)$. However, instead of \DA, we will use
the two classes \DO and \DS introduced
in~\cite{schutzenberger1976} (formally \DO is defined by the equation
$(xy)^\omega (yx)^\omega (xy)^\omega = (xy)^\omega$ and \DS by the equation
$((xy)^\omega (yx)^\omega (xy)^\omega)^\omega = (xy)^\omega$ for $\omega$ the
idempotent power). Both \DO and \DS are supersets of \DA, specifically we have
$\DA \subseteq \DO \subseteq \DS$, and we can test in PSPACE in~$S$ whether
$T(S)$ is in each of these classes~\cite{therien1998over}. Our main result is
then:

\begin{theorem}
  \label{thm:groupdich}
  Let $S$ be a semiautomaton.
  If $T(S)$ is in \DO, then 
  the multi-letter CSh-problem for $S$ is in NL. If 
  $T(S)$ is not in \DS, then it is NP-complete.
\end{theorem}

This result generalizes
Theorem~\ref{thm:dich} for the CSh-problem, because both \DO and \DS collapse to \DA for aperiodic
monoids (see~\cite{schutzenberger1976} and~\cite[Chapter 8]{almeida1994finite}); formally, $\DO \cap \A = \DS \cap \A = \DA$.
However, \DO covers more
languages than \DA: the main technical challenge 
to prove
Theorem~\ref{thm:groupdich} is to show that CSh is tractable for these
languages. One important example are the \emph{group
languages} over~$A$: these are the regular languages recognized, for some choice of
initial-final state pairs, by a semiautomaton $S$ over~$A$ such that $T(S)$ is a group.
A more general example are \emph{district group monomials}, which are the
languages of the form $K_1 a_1 \cdots K_n a_n
K_{n+1}$ where, for all~$i$, we have $a_i \in A$ and $K_i$ is a group language
over some alphabet $A_i
\subseteq A$. Note that district group monomials are more expressive than
the \emph{group monomials} defined in earlier work~\cite{pin1996} (which
set $A_i \colonequals A$ for all~$i$), and they also generalize
the monomials that we studied in Section~\ref{sec:sigma2}
(any $A_i^*$ is
trivially a group language over~$A_i$, even though it is not a group
language over~$A$).
In fact, to prove Theorem~\ref{thm:groupdich},
what we need is to generalize Theorem~\ref{thm:monomial} 
(for CSh) from monomials to district group monomials:

\begin{theorem}
  \label{thm:group}
  Let $K$ be a district group monomial. Then $\pCSh{K}$ is in NL.
\end{theorem}

Note that this theorem, like Theorem~\ref{thm:monomial}, applies to the original
phrasing of CSh, not the alternative phrasing with
semiautomata and multi-letter DAGs. Thus, 
Theorem~\ref{thm:group} implies that the original CSh-problem is tractable for many
languages that we had not covered previously, e.g., $(ab^*a+b)^* c (ba^*b+a)^*$,
the language testing whether there is one $c$ preceded by an even number of~$a$
and followed by an even number of~$b$.
The proof of 
Theorem~\ref{thm:group} is our main technical achievement, and we sketch it below (see Appendix~\ref{apx:group} for details):

\begin{proofsketch}
  We focus on the simpler case of a group language, for a finite group~$H$. The problem
  can be rephrased directly in terms of~$H$: given a tuple $I$ of strings
  over~$H$ and a target
  element $g\in H$, determine if there is an interleaving of~$I$ that
  evaluates to~$g$ under the group operation. Our approach partitions $H$ into
  the \emph{rare} elements $H_\rare$, that occur in a constant number
  of strings, and the \emph{frequent} elements $H_\freq$, that occur in sufficiently
  many strings. For the frequent elements, 
  we can build a large antichain~$C$ from the strings where they occur, with
  each element of~$H_\freq$ occuring many times in~$C$. 
  Now, as topological sorts can choose any order on~$C$, they can intuitively achieve all
  elements of the subgroup $\langle H_\freq\rangle$ generated by~$H_\freq$, except that they
  cannot change 
  ``commutative information'', e.g., the parity of the number of elements.
  We formalize the notion of ``commutative information'' using relational
  morphisms, and prove an \emph{antichain lemma} that captures our intuition
  that all elements of~$\langle H_\freq \rangle$ with the right commutative
  information can be achieved.

  For the rare elements, we can simply follow a dynamic algorithm on the
  constantly many strings where they occur. However, we must account for the
  possibility of inserting elements of~$\langle H_\freq \rangle$ from the
  other strings, and we must show that it suffices to do constantly many
  insertions, so that it was sufficient to impose a constant lower bound
  on~$\card{C}$. We formalize this as
  an \emph{insertion lemma}, which we prove using Ramsey's theorem.
\end{proofsketch}

\begin{toappendix}
  \subsection{Proof of Theorem~\ref{thm:groupdich}: Coarser Dichotomy Theorem for
CSh}
We first explain how 
Theorem~\ref{thm:groupdich} follows from Theorem~\ref{thm:group}, before dealing
with the much more difficult task of proving Theorem~\ref{thm:group}. The
overall scheme is like in Section~\ref{sec:dich}:
show that monoids in \DO can be
reduced to tractable languages (specifically, to district group monomials),
and show that
monoids not in \DS capture an intractable language. For the upper
bound, we use the following result, which is the counterpart of
Theorem~\ref{thm:DA} but for \DO rather than \DA:

\begin{theorem}[(\cite{tesson2005}, Theorem~1)]
  \label{thm:DO}
  Let $K$ be a regular language of $A^*$. The following conditions are
  equivalent:
  \begin{itemize}
    \item $K$ is an union of \emph{unambiguous} district group monomials, i.e.,
      of district group monomials
      $K = K_1 a_1 \cdots K_n a_n K_{n+1}$ such that every word $u \in K$
      has a unique decomposition $u_1 a_1 \cdots u_n a_n u_{n+1}$ where $u_i \in
      K_i$ for all $1 \leq i\leq n+1$.
  \item There exists a monoid $M$ in \DO and a morphism $\eta:A^*\to M$ such that
    $K$ is \emph{recognized} by~$M$, meaning that 
   $K = \eta^{-1}(P)$ for some subset $P \subseteq M$.
  \end{itemize}
\end{theorem}

For the lower bound, we use the following folklore result, which extends
Proposition~\ref{prop:DAab} to the non-aperiodic case:

\begin{proposition}[(\cite{almeida1994finite}, Exercise~8.1.6)]
  \label{prop:DSab}
  A monoid $M$ is not in \DS iff there exists a morphism
  $\theta:\{a,b\}^*\to M$ and $P\subseteq M$
  such that $\theta^{-1}(P)$ is either $(ab)^*$ or $(ab + b)^*$.
\end{proposition}

From these two results, we can prove Theorem~\ref{thm:groupdich} exactly like
we proved Theorem~\ref{thm:dich} in the previous section, using
Theorem~\ref{thm:group} instead of Theorem~\ref{thm:monomial}. The hard work
that remains is to prove Theorem~\ref{thm:group}.

\subsection{High-Level Presentation of the Proof of Theorem~\ref{thm:group}}
\label{apx:highlevel}
This appendix gives a high-level view of the proof of Theorem~\ref{thm:group}.
For most of the proof, we focus
on the case of \emph{group
languages}: the case of group district monomials is only presented at the very end,
in Appendix~\ref{apx:district}.
The CSh-problem for
group languages can essentially be stated directly in terms of the underlying group: we
fix a finite group~$H$ and a target element~$g$, our instance to the CSh-problem
is a tuple $I$ of strings over~$H$, and we want to test if there is an
interleaving of~$I$ which evaluates to~$g$ according to the group operation. So
we see $A \colonequals H$ as the alphabet of~$I$.

As explained in the proof sketch, given the
CSh-instance $I$, we will
split the letters of~$A$ between \emph{rare} letters $A_\rare$
and \emph{frequent} letters~$A_\freq$, which we call a \emph{rare--frequent partition}.
This will ensure that the rare letters $A_\rare$ only occur in constantly
many input strings (called the \emph{rare} strings), and the frequent letters
$A_\freq$ occur in sufficiently many different input
strings (called the \emph{frequent} strings).

For the frequent letters, the key idea 
is that we can pick many occurrences of each frequent letter in different
strings, and obtain an \emph{antichain} $C$ (subset of pairwise incomparable
elements), which
contains many occurrences of each frequent letter of~$A_\freq$.
Now, in a topological sort, we can enumerate all elements of~$C$ contiguously,
following any permutation on~$C$. Intuitively, as $C$ contains many occurrences
of each frequent letter, this should give us the freedom to create many different 
elements in the subgroup of~$H$ generated by~$A_\freq$. We cannot obtain
\emph{all} elements of this subgroup, because the number of
occurrences of each group element is fixed by
that of~$C$. To formalize this intuition, the notion of \emph{Parikh image} is helpful:

\begin{definition}
  \label{def:parikh}
Write the alphabet $A$ as $a_1, \ldots, a_k$ in some fixed order.
The \emph{Parikh image} 
of a word $w \in A^*$ is $\PI(w) \colonequals (|w|_{a_1}, \ldots,
|w|_{a_k}) \in \NN^k$,
where $|w|_a$ for $a \in A$ denotes the number of
occurrences of~$a$ in~$w$. The \emph{Parikh image} of a language~$K$ is then
the set $\PI(K)$ of the Parikh images of the words that $K$ contains: for
instance, $\PI((ab)^*) = \{(i, i) \mid i \in \NN\}$.

The \emph{Parikh image} $\PI(G)$ of an $A$-DAG $G = (V, E, \lfun)$ is
$(|G|_{a_1}, \ldots, |G|_{a_k})$, with each $|G|_{a_i}$ being $\card{\{v \in V \mid
  \lfun(v) = a_i\}}$. The \emph{Parikh image} $\PI(I)$ of a CSh-instance~$I$ is
  defined in the same way, seeing $I$ as a DAG formed of disjoint paths.
\end{definition}

As it will turn out, the Parikh image is the only constraint on what we
can generate using such an antichain~$C$.
We formalize this intuition in the
\emph{antichain lemma} (Lemma~\ref{lem:groupper}):
we show that, for any finite group,
if we have enough copies of each element, we can permute them to realize any
element of the group, up to ``commutative constraints''. Thanks to this, the CSh-problem simply
reduces to a test on the Parikh image $\PI(I)$ of the instance, under our initial
assumption.

We must now explain how to handle the rare letters~$A_\rare$. We can simply look
at the constant number of strings that contain a letter of~$A_\rare$, and handle
these strings with an approach based
on dynamic programming, as in the proof of Proposition~\ref{prp:dynamic}.
So it seems like the problem is solved: apply dynamic programming to the rare
strings, and use the antichain lemma to argue that the frequent strings can
generate any letter of the subgroup spanned by $A_\freq$, up to the
commutative constraints. However, one difficulty remains:
in a topological sort of the rare strings, we can insert elements from
the frequent strings at any point in the dynamic algorithm, and the rare strings
may be arbitrarily long; yet the frequent strings cannot create arbitrarily many
copies of each group element, because
we must use a constant bound when splitting $H$ into $A_\rare$ and $A_\freq$. We
address this by proving a result called the \emph{insertion lemma}
(Lemma~\ref{lem:ramsey}), which intuitively says that a constant number of
insertions always suffice. This is the result whose proof uses Ramsey's theorem.
Thanks to
the insertion lemma, it suffices to allow constantly many insertions of frequent
elements when performing the NL algorithm on the rare strings, which allows us
to conclude.

We give some more detail by stating the antichain lemma and insertion lemma as
standalone results (and defer their complete proof to the next sections of the
appendix, i.e., Appendices~\ref{apx:antichain} and~\ref{apx:ramsey}). We then
formalize the rare--frequent partition and
sketch the remainder of the proof of Theorem~\ref{thm:group} (the details about
the reminder of the proof are given in Appendix~\ref{apx:groupproof}).

\subparagraph*{Antichain lemma.}
Let $G$ be an $A$-DAG over some alphabet~$A$, let $C$ be an antichain of~$G$,
and let $n \in \NN$. 
We call $C$ an \emph{$n$-rich antichain} if each letter of $A$ appears at least $n$
times in~$C$.
The \emph{antichain lemma} intuitively
shows that when $G$ has a rich antichain, then it suffices to
look at commutative information of~$G$, namely, its Parikh image, to decide
whether it has a topological sort that achieves a group element. In fact, the
claim applies to any constant-length sequence of group elements, following our
needs for the insertion lemma later. Formally:

\newcommand{\lemgroupper}{
  Let $H$ be a finite group and $\mu:A^*\to H$ be a surjective morphism.
  For any integer $k>0$,
  there exists an integer $n_k$ such that, 
  for any $A$-DAG
  $G = (V, E, \lfun)$ with an $n_k$-rich antichain,
  for any elements  $g_1,\ldots,g_k$ of $H$, if
  $\PI(G)\in\PI(\mu^{-1}(g_1\cdots g_k)) $ then there is a topological
  sort $\sigma$ of~$G$ decomposable as $\sigma=\sigma_1\cdots \sigma_k$
  such that $\mu(\lfun(\sigma_i))=g_i$ for each
  $i\in\{1,\ldots,k\}$.
}
\begin{lemma}[(Antichain lemma)]
  \label{lem:groupper}
  \lemgroupper
\end{lemma}

Note that this result is not specific to the CSh-problem, and applies
to arbitrary DAGs. We now sketch its proof here; the complete proof is given
in Appendix~\ref{apx:antichain}:

\begin{proofsketch}
  We capture the ``commutative information'' contained in the Parikh image of
  the rich antichain as an element in a commutative monoid $N$ constructed from
  the commutative closure of~$H$. The elements that we can hope to reach with
  the antichain are then the images of this element of~$N$ by a so-called 
  \emph{relational morphism}~\cite{eilenberg1974} written $\tau: N \to \powerset{H}$. 
  Intuitively, for $n \in N$ capturing some ``commutative information'',
  $\tau(n)$ are the elements of~$H$ which correspond to this information.
  We then study the elements of~$N$ that use sufficiently many copies of each
  generator of~$N$, called the \emph{fully recurrent} elements, and show that
  their images by~$\tau$ all have the same cardinality. In other words, all
  antichains that are sufficiently rich can achieve the same number of elements
  of~$H$. This allows us to conclude, because making the antichain richer always
  allows us to reach more elements, so an antichain which is richer than this
  threshold always achieves the maximal possible number of elements.
\end{proofsketch}

\subparagraph*{Insertion lemma.}
We now turn to the \emph{insertion lemma}, which allows us to show that
we only need to insert group elements at a constant number of places.
More precisely, when we achieve a group element by interleaving two
sequences, we can always interleave them differently so that there are
constantly many insertions and still achieve the same element.

\newcommand{\lemramsey}{
  Let $H$ be a finite group and $\mu : A^* \to H$ be a surjective morphism.
  There exists a constant $B \in \NN$ such that,
  for any $n \in \NN$,
  for any $n$-tuple $w_1, \ldots, w_n$ of words of~$A^*$
  and $(n+1)$-tuple $w'_0, \ldots, w'_n$ of words 
  of~$A^*$,
  letting $u = w'_0 w_1 w_1' w_2 w_2' \cdots w_n w'_n$,
  there exists a set $J \subseteq \{0, \ldots, n\}$ of cardinality at most~$B$
  such that, letting $w''_j$ for $0 \leq j \leq n$ be $w'_j$ if $j \in J$ and
  the empty word otherwise, 
  letting
  $v = w''_0 w_1 w_1'' \cdots w_n w''_n$, we have $\mu(u) = \mu(v)$
  and $\mu(w'_0 \cdots w'_n) = \mu(w''_0 \cdots w''_n)$.
}
\begin{lemma}[(Insertion lemma)]
  \label{lem:ramsey}
  \lemramsey
\end{lemma}

We give a sketch of the result; the complete proof is presented in
Appendix~\ref{apx:ramsey}:

\begin{proofsketch}
  We reason on the complete graph of positions of the word $u$, coloring each
  edge by three group elements derived from the corresponding factor: the group
  element achieved when performing the insertions (from~$u$), the group element achieved
  when we do not perform them (from~$v$), and the group element achieved by the insertions
  on their own (from the $w_i'$).
  We then use Ramsey's theorem to extract a monochromatic triangle
  in this graph: we show that, in the factor spanned by this triangle, there is
  no difference between performing the insertions and not performing them. We
  can repeat this argument as long as the word has sufficiently many letters,
  so we reach a constant bound~$B$ which comes from Ramsey's theorem.
\end{proofsketch}

\subparagraph*{Putting the proof together.}
We are now ready to explain at a high level the rest of the proof of
Theorem~\ref{thm:group} in the case of group languages.
Let~$K$ be a group language on the alphabet $A = \{a_1, \ldots, a_k\}$. We let
$\mu: A^* \to H$ be the syntactic morphism of~$K$, where 
$H$ is a finite group generated by
the $\mu(a_i)$.
We consider an instance $I = (S_1, \ldots, S_n)$ to the CSh-problem,
where each $S_i$ is a string of vertices labeled with letters of the alphabet~$A$.
Let $B$ be the bound whose existence is shown in Lemma~\ref{lem:ramsey}, and,
using
Lemma~\ref{lem:groupper} for the value $k \colonequals B$,
let $R$ be the value of~$n_k$ given by this lemma. We will decompose $I$
following a \emph{rare--frequent partition}, which we now define:

\begin{definition}
  \label{def:rarefreq}
A \emph{rare--frequent} partition of $I$ consists of
a partition of $A$ into \emph{rare} letters $A_\rare$
and \emph{frequent} letters $A_\freq$, and a partition of the strings into
\emph{rare} strings $S_\rare$ and \emph{frequent} strings $S_\freq$,
where all vertices of~$S_\freq$ are labeled with letters of~$A_\freq$, and where
$S_\freq$, when seen as an subinstance of~$I$ over the alphabet~$A_\freq$,
contains an $R$-rich antichain.
\end{definition}

Note that, in a partition, rare strings may still contain arbitrarily many
frequent letters, and rare letters may still occur a unbounded number of times
overall in~$I$, as they can occur arbitrarily many times in each rare string.
We can then show the following:

\begin{lemma}
  \label{lem:rarefreq}
  For any fixed alphabet $A$ of size~$k$, given an input CSh-instance $I = (S_1, \ldots,
  S_n)$, we can compute a rare--frequent partition of $I$ in NL, represented as
  the partition $A_\freq \sqcup A_\rare$ of~$A$
  and the set of rare strings $S_\rare$, such that 
  $|S_\rare| \leq R \cdot k^2$.
\end{lemma}

\begin{proof}
  We first argue for the existence of a suitable rare--frequent partition by
  giving a naive algorithm to construct it, and then justify that we can do it
  in NL instead.

  The naive algorithm initializes $A_\rare = \emptyset$, $A_\freq = A$, $S_\rare
  = \emptyset$, $S_\freq = S$, and does the following until convergence: if a
  letter $a \in A_\freq$ occurs in less than $R\cdot k$ strings of~$S_\freq$, then
  remove $a$ from~$A_\freq$, add $a$ to~$A_\rare$, remove the $\leq R\cdot k$ strings
  that contain $a$ from~$S_\freq$, and add them to~$S_\rare$. As we perform the
  move operation at most once for each letter, it is immediate that the
  algorithm terminates, and that at the end there are at most $R \cdot k^2$ rare
  strings: now the definition of the algorithm clearly ensures that $S_\freq$
  cannot contain any letter of $A_\rare$ and that each letter of~$A_\freq$
  occurs in at least $R \cdot k$ different strings of~$S_\freq$. By picking $R$
  strings of~$S_\freq$ for each letter of~$A_\freq$ in a way that does not overlap, we see
  that $S_\freq$ contains an $R$-rich antichain for the alphabet $A_\freq$.
  Hence, a suitable rare--frequent partition exists.

  To construct the rare--frequent partition in NL, simply guess the partition
  $A_\rare \sqcup A_\freq$
  of~$A$, guess the set $S_\rare$ of rare strings of size $\leq R \cdot k^2$
  (which is constant), guess $R$ occurrences
  for each letter of~$A_\freq$, check that they are all in different strings and
  that they are not in strings of~$S_\rare$, and check that the strings which
  are not in $S_\rare$ contain only frequent letters.
\end{proof}

Hence, we assume that we have computed in NL a rare--frequent partition of~$I$,
given by $A_\rare$, $A_\freq$, $S_\rare$, and (implicitly) $S_\freq$.
We 
write $H_\freq$ for the subgroup of~$H$ equal to $\mu(A_\freq^*)$, i.e., the
subgroup spanned by $A_\freq$.
We can now sketch the remainder of the proof of
Theorem~\ref{thm:group}:

\begin{proofsketch}
  Our goal is to determine whether $I$ has some topological sort in~$K$. We
  relabel all elements of~$I$ with their image in~$H$ by~$\mu$, and equivalently
  test whether $I$ has a topological sort achieving a target group element
  $g\in H$. We do so by an NL algorithm: we perform the analogue of 
  Proposition~\ref{prp:dynamic} on the rare strings $S_\rare$, with some
  insertions of a constant number of elements from~$H_\freq$ which
  respect the constraints on the Parikh image (again formalized via
  the notion of relational morphisms). 
  To show correctness, we rely on the antichain lemma (Lemma~\ref{lem:groupper})
  to argue that any such pattern of insertions can indeed be performed
  using $S_\freq$, thanks to the rich antichain that it contains.
  To show completeness, we rely on the insertion lemma
  (Lemma~\ref{lem:ramsey}) to argue that any topological sort achieving an
  element of~$H$ can indeed be rewritten to an equivalent one where we only
  perform constantly many insertions.
\end{proofsketch}

In the rest of the appendix, we first prove the antichain lemma in 
Appendix~\ref{apx:antichain}, and then prove the insertion lemma in
Appendix~\ref{apx:ramsey}. We then complete our presentation of the 
proof of Theorem~\ref{thm:group} for group languages in 
Appendix~\ref{apx:groupproof}, using the two lemmas and some of the notions
introduced in Appendices~\ref{apx:antichain} and~\ref{apx:ramsey}. Last, we extend the proof
to district group monomials in Appendix~\ref{apx:district}.

  \subsection{Proof of Lemma~\ref{lem:groupper}: Antichain Lemma}
  \label{apx:antichain}

  To prove the antichain lemma, let us fix the finite group $H$ and morphism
  $\mu$.
Recall the definition of the Parikh image (Definition~\ref{def:parikh}), and let
us define the \emph{commutative closure} $\CCl(K)$ of a regular language~$K$ as
$\PI^{-1}(\PI(K))$, where $\PI$ denotes the
  Parikh image (Definition~\ref{def:parikh}).
  Remark
that, for any element $g \in H$, the inverse image $\mu^{-1}(g)$ is a group
language. Relying on some more standard notions from algebraic language
  theory, we will say that a language $K$ is \emph{recognized} by the
  morphism~$\mu$ if there exists $P \subseteq H$ such that $K = \mu^{-1}(P)$. We
  will also talk about the \emph{syntactic monoid} of~$K$, which is the 
  transition monoid of the minimal automaton which recognizes~$K$.
  
  We will use the following result on the group languages defined as
  $\mu^{-1}(g)$ for~$g\in H$:

\begin{lemma}[(\cite{GomezGP08}, Theorem 3.1)]
  \label{lem:groupparikh}
  The commutative closure of a group language is regular.
\end{lemma}

Remark that this result does not hold for the commutative closure of arbitrary
regular languages (e.g., $(ab)^*$), and that the commutative closure of a group
language is not necessarily a group language
(see
\cite{GomezGP08} for a counterexample).
Let us accordingly define a finite monoid~$N$, and let $\com_\mu:A^*\to N$ be a surjective morphism 
such that,
for each $g\in H$, the morphism~$\com_\mu$ recognizes $\CCl(\mu^{-1}(g))$.
We can construct $N$, for instance, by taking the direct product of the
syntactic monoids recognizing the commutative closure of each $\mu^{-1}(g)$,
using Lemma~\ref{lem:groupparikh}. Further, 
thanks to commutativity, we can choose $N$ to be a finite commutative monoid.
Let $\omega$ be a positive 
\emph{idempotent power} of~$N$, that is, a value $\omega \in \NN \setminus
\{0\}$ such that we have $p^{2\omega} = p^\omega$ for
every $p \in N$. (Such an idempotent power exists: indeed, for every $p$ in $N$,
there exists $k$ such that $p^k=p^{2k}$, and we can take $\omega$ to be the
least common multiple of the idempotent powers of all elements of~$N$.)

To characterize the ``commutative information'' of elements of~$H$, we will
study the connection between $H$ and the commutative monoid~$N$. We will do so
using relational morphisms.
A \emph{relational morphism}~\cite{eilenberg1974} between two monoids $M$ and
  $M'$ is a map from $M$ to 
the powerset $\mathcal{P}(M')$ of~$M'$, such that
for all $m\in M$ we have $\tau(m)\neq \emptyset$, and 
for all $m, m'\in M$, we have $\tau(m)\cdot\tau(m')\subseteq\tau(m m')$, 
where we extend the product operator of~$M'$ to the powerset monoid of~$M'$ in the expected
way, that is, $S \cdot S' = \{g \cdot g' \mid g \in S, g' \in S'\}$.
For any surjective morphism $\eta:A^*\to M$ and morphism $\mu:A^*\to M'$, the map $m\mapsto \mu(\eta^{-1}(m))$ is a relational
morphism. We write $\tau:M\relto M'$ if $\tau$ is a relational morphism between
  $M$ and $M'$.

We can now introduce the crucial notion of \emph{fully recurrent} elements for
our purposes, which will formalize the connection to rich antichains.
An element $p$ of a commutative monoid $N$
is said to be \emph{fully recurrent} if
  there exists a generator $S$ of~$N$ and \emph{positive} integers $r_1, \ldots, r_n$
such that $p = s_1^{r_1} \cdots s_n^{r_n}$, where $n = |S|$, and $r_i \geq
\omega$ for all $1 \leq i \leq n$.

The notion of fully recurrent elements is motivated by the following lemma:

\begin{lemma}
  \label{lem:sizeindep}
  Let $\tau:N\relto H$ be any relational morphism from a commutative monoid to a finite group.
  For any fully recurrent elements $p$ and $q$ of $N$, the sets  $\tau(p)$ and $\tau(q)$ have the same size.
\end{lemma}
\begin{proof}
  We will show the result using the following claim (*):
  for any fully recurrent element~$r$, we have
  $|\tau(r)|=|\tau(r^i)|$ for any $i \geq 1$. This suffices to conclude the
  lemma, because for any fully recurrent elements $p$ and $q$, we have $p^\omega
  = q^\omega$. Indeed, writing $p = s_1^{r_1} \cdots s_n^{r_n}$, we have
  $p^\omega = (s_1^{\omega})^{r_1} \cdots (s_n^{\omega})^{r_n} = s_1^\omega
  \cdots s_n^\omega$, and similarly for~$q$. This allows us to conclude from (*)
  because we have 
  $|\tau(p)| = |\tau(p^\omega)| = |\tau(q^\omega)| = |\tau(q)|$.

  So we simply show claim (*). Let $r$ be a fully recurrent element, and let us
  study the sequence $(x_i)$ defined by $x_i \colonequals
  |\tau(r^i)|$ for all $i \geq 1$. We must show that the sequence $(x_i)$ is
  constant. We do this in two parts: (i) we show that it is nondecreasing, and
  (ii) we show 
  that there are arbitrary large $b \in \NN$ such that $x_b = x_1$. Parts (i)
  and (ii)
  clearly imply that the sequence is constant, which establishes (*).
  
  For part (i), we show that $|\tau(r^i)| \leq |\tau(r^{i+1})|$
  for all $i \geq 1$.
  By definition of relational morphisms, we have $\tau(r^i) \tau(r) \subseteq
  \tau(r^{i+1})$.
  Now, remembering that the empty set is
  not in the image of a relational morphism, pick any $x \in \tau(r)$. 
  We know that $\tau(r^i) \cdot \{x\} \subseteq \tau(r^i) \tau(r)$. Now, as $x \in H$ and $H$
  is a group, we know that $H$ acts bijectively on any subset of~$H$, in
  particular $\tau(r)$, hence $|\tau(r^i)| = |\tau(r^i) \cdot \{x\}| \leq
  |\tau(r^i) \tau(r)| \leq |\tau(r^{i+1})|$. This shows part (i).

  We now show part (ii). To do so, let us show first that $r^{\omega+1} = r$.
  Indeed, write $r = s_1^{r_1} \cdots s_n^{r_n}$, and we simply conclude using
  the fact that $s_i^{r_i + \omega} = s_i^{r_i - \omega}
  (s_i^{\omega})^2 = s_i^{r_i - \omega} s_i^\omega = s_i^{r_i}$.
  This implies that we have $r^{j\omega + 1} = (r^\omega)^j r = r^\omega r = r$,
  for any $j \geq 0$.
  As $\omega \geq 1$, there are arbitrarily large values of $j\omega$, so this concludes part~(ii)
  and we have established claim (*), which finishes the proof.
\end{proof}

We are now ready to show the antichain lemma (Lemma~\ref{lem:groupper}). Recall
its statement:

\begin{quote}
  \textsf{\textbf{Lemma~\ref{lem:groupper}}}: \lemgroupper
\end{quote}

\begin{proof}[Proof of Lemma~\ref{lem:groupper}]
  Fix the finite group $H$, and let $\mu:A^*\to H$ be the surjective morphism.
  We fix $\gamma = \max_{g \in H} \min_{u \in \mu^{-1}(g)} |u|$: this value is
  well-defined because $\mu$ is surjective, and is finite because $H$ is finite.
  Let $\com_\mu:A^*\to N$ be 
  the surjective morphism defined as before, where $N$ is a commutative monoid,
  and let $\omega$ be the {idempotent power} of~$N$.
  Finally, let $\tau:N\relto H$
  be the relational morphism defined by $\tau(x) = \mu(\com_\mu^{-1}(x))$.
  Observe that the Parikh image assumption
  on the input $A$-DAG~$G$ and on the $g_1, \ldots, g_k$ in the statement
  of the lemma is equivalent to $\com_\mu(G) \in \com_\mu(\mu^{-1}(g_1 \cdots
  g_k))$. Indeed, the forward implication is immediate, and the converse holds
  because $\com_\mu$ recognizes $\CCl(\mu^{-1}(g_1\cdots
  g_k))$, so the rephrased condition implies that $\CCl(G) \in \CCl(\mu^{-1}(g_1
  \cdots g_k))$, which clearly implies the original condition.
  Further, by composing with~$\tau$ and simplifying using the definition
  of $\tau$, the condition rephrases to $g_1 \cdots g_k \in \tau(\com_\mu(G))$.
  We will use this equivalent rephrased condition throughout the proof.

  Let us now show the result by induction on~$k > 0$.
  For every $k$, we will choose
  $n_k \colonequals \omega + (k-1) \gamma$. 
  Let us first show the base case for $k = 1$ and $n_k = \omega$.
  Let $G = (V, E, \lfun)$ be the input $A$-DAG to the CTS-problem, and let us study the set
  $T=\{\mu(\lfun(\sigma))\mid \sigma\text{ is a topological
  sort of }G\}$. Remembering that all topological sorts of~$G$ have the same
  Parikh image, namely, $\PI(G)$, we know from the commutativity of~$N$ that all
  topological sorts of~$G$ have the same image by~$\com_\mu$, namely,
  $\com_\mu(G)$. Hence, $T$ is included in $\tau(\com_\mu(G))$.
  Our goal is to show that, when $G$ has a $\omega$-rich antichain, we have $T
  = \tau(\com_\mu(G))$. Indeed, in this case, we know that, for any $g_1$ such
  that $\PI(G) \in \PI(\mu^{-1}(g_1))$, we have $g_1 \in \tau(\com_\mu(G))$ as
  we explained above, so $g_1 \in T$ and there is a topological sort $\sigma
  \colonequals \sigma_1$
  of~$G$ such that $\mu(\lfun(\sigma_1)) = g_1$. So all that remains to show for the base
  case is that $T = \tau(\com_\mu(G))$.

  Let $C$ be a $\omega$-rich antichain of~$G$. For simplicity, let us make $C$
  \emph{maximal}: whenever some vertex $x$ of~$G$ is not in~$C$ but is
  incomparable to all vertices of~$C$, we add it to~$C$. We choose the vertices
  arbitrarily. At the end of the process, $C$ is still an antichain, and it is still
  $\omega$-rich. Further, we can partition $G$ as $G^- \sqcup C \sqcup G^+$,
  where $G^-$ contains all vertices having a directed path of positive length to
  a vertex of~$C$, and $G^+$ contains all vertices having a directed path of
  positive length from a vertex of~$C$.
  To see why this is a partition, observe that it covers $G$ because any
  counterexample vertex $x$ would contradict the maximality of~$C$. Further, $C$
  is disjoint from $G^+$, and from $G^-$, because it is an antichain, and $G^+$
  and $G^-$ are disjoint: any element in $G^+ \cap G^-$ would witness by
  transitivity a path from an element of~$C$ to an element of~$C$, contradicting
  the fact that $C$ is an antichain. 

  Let $\sigma^-$ and $\sigma^+$ be arbitrary topological sorts of~$G^-$ and $G^+$
  respectively. Our chosen partition ensures that we can build a topological
  sort of~$G$ as $\sigma^-, \sigma, \sigma^+$ where $\sigma$ is a topological sort of~$C$. Hence,
  $T'\colonequals \mu(\sigma^-)\cdot \tau(\com_\mu(C))\cdot \mu(\sigma^+)$
  is a subset of $T$, so $|T'| \leq |T|$.
  Let us now write
  $s_a\colonequals \com_\mu(a)$
  for each letter $a\in A$.
  We can write $\com_\mu(C) = \Pi_{a\in A} s_a^{i_a}$,
  where $i_a$ is the number of vertices labeled by~$a$ in~$C$. As $C$ is
  $\omega$-rich, we have $i_a \geq \omega$. Thus,
  $\com_\mu(C)$ is fully recurrent by definition. Now, it is clear that
  $\com_\mu(G)$ is also fully recurrent, because $G$ is $\omega$-rich also.
  Thus, by Lemma~\ref{lem:sizeindep}, we have $|\tau(\com_\mu(C))| = 
  |\tau(\com_\mu(G))|$. Now,
  we know that $\mu(\sigma^-)$ (resp.\ $\mu(\sigma^+)$) act bijectively on the
  left (resp.\ right) of~$H$, so we also have $|\tau(\com_\mu(C))| = |T'|$. We have
  thus shown that $|\tau(\com_\mu(G))| = |T'| \leq |T|$.
  As $T \subseteq \tau(\com_\mu(G))$, we deduce that $T = \tau(\com_\mu(G))$.
  As we have argued, this concludes the proof of the base case $k = 1$.

  \medskip

  We now prove the inductive step.  
  Assume the property holds for $k > 0$. Let $G$ be an instance of the 
  CTS-problem
  that has a $n_{k+1}$-rich antichain: as in the base case we expand
  it to a maximal such antichain, denote it by $C$, partition $G$ as $G^-
  \sqcup C \sqcup G^+$, and let $\sigma^-$ and $\sigma^+$ be arbitrary topological sorts
  of~$G^-$ and~$G^+$ respectively.
  Let us choose elements $g_1,\ldots,g_{k+1}$ of~$H$ such that
  $g_1\cdots g_{k+1} \in \tau(\com_\mu(G))$: remember that this implies that $g_1
  \cdots g_{k+1} \in \tau(\com_\mu(G))$.

  Now, let us consider $g' \colonequals g_{k+1}\cdot\mu(\sigma^+)^{-1}$.
  Let $u_{g'} \in A^*$ be a word that realises the minimum in the definition of~$\gamma$,
  and let $C_{g'}$ be a subset of~$C$ whose
  elements are labeled with the letters of~$u_{g'}$.
  As $C$ is
  $n_{k+1}$-rich, we can find such a subset, and further $C
  \setminus C_{g'}$ is still a $((k-1)\gamma+\omega)$-rich antichain, i.e., an
  $n_k$-rich antichain.
  Further, the definition of $C_{g'}$ ensures that it has a topological sort 
  $\sigma'$ that realizes the word $u_{g'}$, so that $\mu(\sigma') = g'$.
  By composing it with $\sigma^+$, we can then construct $\sigma' \sigma^+$, which is a
  topological sort of $G'' \sqcup C_{g'} \sqcup G^+$ such that
  $\mu(\lfun(\sigma' \sigma^+)) = g_{k+1}$.

  We now wish to apply the induction hypothesis for $g_1, \ldots, g_k$ on the 
  subinstance $G' \colonequals G^- \sqcup (C \setminus C_{g'})$, which still has a $n_k$-rich
  antichain.
  To do so, we must check that $\PI(G') \in \PI(\mu^{-1}(g_1 \cdots g_k))$,
  which as we argued is equivalent to $g_1 \cdots g_k \in \tau(\com_\mu(G'))$.
  As $G$ is the disjoint union of $G'$ and $G''$, we have $\com_\mu(G) = \com_\mu(G') \com_\mu(G'')$, so by composing
  by~$\tau$ and applying the definition of a relational morphism we have:
  \[
    \tau(\com_\mu(G)) \subseteq \tau(\com_\mu(G') \com_\mu(G''))
  \]
  Now, as both $G$ and $G'$ contain an antichain which is at least
  $\omega$-rich,
  we know that $\com_\mu(G)$ and $\com_\mu(G)$ are fully recurrent.
  By applying Lemma~\ref{lem:sizeindep} again, we know that
  $|\tau(\com_\mu(G))| = |\tau(\com_\mu(G')|$.
  Remember now that $\sigma' \sigma^+$ is a topological sort of~$G''$ 
  such that $\mu(\lfun(\sigma' \sigma^+)) = g_{k+1}$. Hence, $g_{k+1} \in \tau(\com_\mu(G''))$. Now, as $g_{k+1}$ acts
  bijectively on $\tau(\com_\mu(G'))$ in the group~$H$, we deduce that
  $\tau(\com_\mu(G)) = \tau(\com_\mu(G')) g_{k+1}$. Now, since we have $g_1
  \cdots g_{k+1} \in \tau(\com_\mu(G))$ by hypothesis, we deduce that indeed
  $g_1 \cdots g_k \in \tau(\com_\mu(G'))$, so we can apply the induction
  hypothesis.

  Hence, we do so and obtain a 
  topological sort $\sigma_1, \ldots, \sigma_k$ of~$G'$ such that
  $\mu(\lfun(\sigma_i))
  = g_i$ for each $i \in \{1, \ldots, k\}$. Now, letting $\sigma_{k+1} \colonequals
  \sigma' \sigma^+$, it is clear that $\sigma_1, \ldots, \sigma_k, \sigma_{k+1}$
  is a topological sort
  of~$G$, and we have $\mu(\lfun(\sigma' \sigma^+)) = g_{k+1}$, so we have shown the induction
  hypothesis. This concludes the proof.
\end{proof}

  \subsection{Proof of Lemma~\ref{lem:ramsey}: Insertion Lemma}
  \label{apx:ramsey}
  We now prove the insertion lemma (Lemma~\ref{lem:ramsey}). Recall its
  statement:

\begin{quote}
  \textsf{\textbf{Lemma~\ref{lem:ramsey}}}: \lemramsey
\end{quote}

\begin{proof}[Proof of Lemma~\ref{lem:ramsey}]
  Fix the alphabet $A$, the morphism $\mu$, and the group~$H$. By Ramsey's
  theorem, there exists a constant $B$ such that, for any complete graph
  $\Gamma$ whose edges are labeled with triples of elements of~$H$, if $\Gamma$
  has at least $B$ vertices, then it contains a monochromatic triangle, that
  is, three vertices $v_1, v_2, v_3$ such that the edges $\{v_1, v_2\}$,
  $\{v_2, v_3\}$, and $\{v_1, v_3\}$ are labeled by the same triple of elements
  of~$H$.

  Let us now show the rest of the claim by
  strong induction on $n \in \NN$. The base
  case of the induction is when $n < B$, and in this case there is nothing to
  show: we can simply take $J = \{0, \ldots, n\}$ which achieves the cardinality
  bound, and then we have $u = v$ so clearly $\mu(u) = \mu(v)$.

  Let us now show the induction step. We take an arbitrary $n \in \NN$ with $n
  \geq B$, assume that the result is true for all smaller~$n$, and show the
  result for~$n$. Fix the words $w_i$ and $w_i'$. Now, let us construct the
  complete graph $\Gamma$ with $n$ vertices $v_1, \ldots, v_n$ 
  and with edges colored by triples of elements of~$H$ in the following way: the
  edge between $v_i$ and $v_j$ for $i < j$ is colored with the triple $(g_{i,j},
  g'_{i,j}, g''_{i,j})$, where we define
  $g_{i,j} \colonequals \mu(w_i \ldots w_{j-1}))$,
  $g'_{i,j} \colonequals \mu(w_i w_i' \cdots w_{j-1} w'_{j-1})$,
  and
  $g''_{i,j} \colonequals \mu(w_i' \cdots w'_{j-1})$.

  Now, by Ramsey's theorem, as $\Gamma$ has more than $B$ vertices, it has a
  monochromatic triangle. This implies that there are $1 \leq l < m < r \leq n$
  such that $g_{l,m} = g_{m,r} = g_{l,r}$, and $g'_{l,m} = g'_{m,r} = g'_{l,r}$. 
  Now, as by definition we have $g_{l,r} = g_{l,m} g_{l,r}$, this means that
  we have $g_{l,r} = g_{l,r}^2$, and as $H$ is a group we can simplify and
  deduce that $g_{l,r} = e$, the neutral element of~$H$. We deduce in the same
  way that $g_{l,r}' = e$. Hence, we have shown $g_{l,r} = g'_{l,r}$, which
  means that (*): $\mu(w_l w_l' \cdots w_{r-1} w_{r-1}') = \mu(w_l \cdots w_{r-1})$.
  Further, we deduce in the same way that (**) $g''_{l,r} = e$.

  We will now conclude using the induction hypothesis. Let $n' = n - (r - l)$,
  and consider the $n'$-tuple $w_1, \ldots, w_{l-1}, (w_l \cdots w_{r-1}), w_r,
  \cdots, w_n$ of words of~$A^*$, and the $(n'+1)$-tuple $w'_0, \ldots,
  w'_{l-1}, w'_r, \ldots, w'_n$. Using the induction hypothesis for~$n'$, we
  deduce the existence of $J' \subseteq \{0, \ldots, n'\}$ of cardinality at
  most~$B$ such that,
  defining $w'''_j$ for all $0 \leq j \leq n'$ as
  the empty word if $j \notin J$, as $w_j$ if $j \in J'$ and $j < l$, and as
  $w_{j + (r - l)}$ if $j \in J$ and $j \geq l$,
  letting 
  \begin{align*}u' \colonequals & w'_0 w_1 w_1' \cdots w_{l-1} w'_{l-1}
  (w_l \cdots w_{r-1}) w_r w_r' \cdots w_n w_n'\\
    v' \colonequals & w'''_0 w_1 w_1'' \cdots w_{l-1} w'''_{l-1}
  (w_l \cdots w_{r-1}) w_r w_r'' \cdots w_n w_n''
  \end{align*}
  we have $\mu(u') = \mu(v')$, and we have (***) $\mu(w'_0 \cdots w'_{l-1} w'_r \cdots
  w'_n) = \mu(w'''_0 \cdots w'''_n)$.
  Let us accordingly define $J \subseteq \{0, \ldots, n\}$
  by $\{j \mid j \in J, j < l\} \sqcup \{j + (r - l) \mid j \in J, j \geq l\}$,
  which satisfies the cardinality bound. Let us show that $\mu(u) = \mu(v)$ and
  $\mu(w_0 \cdots w_n) = \mu(w'''_0 \cdots w'''_n)$ with
  $v$ and the $w''_i$ defined from this choice of~$J$.
  From the equality (*), we know that we can replace $(w_l \cdots w_{r-1})$
  by $(w_l w_l' \cdots w_{r-1} w_{r-1}')$ in $u'$ without changing its image
  by~$\mu$, so we have $\mu(u) = \mu(u')$. Second, from the fact that $J$
  does not contain any element in $\{l, \ldots, r-1\}$, we know that $w''_j$ is
  empty for all $j \in \{l, \ldots, r-1\}$, so we have $w_l \cdots
  w_{r-1} = w_l w_l'' \cdots w_{r-1} w''_{r-1})$; further, from this and our
  definition of $J$, we observe that $v = v'$, hence $\mu(v) = \mu(v')$. We thus
  deduce that $\mu(u) = \mu(v)$. Last, we can use (**) to insert in (***) the
  product $w''_l \cdots w''_{r-1}$, to establish the second required equality.
  This concludes the proof.
\end{proof}

\subsection{Proof of Theorem~\ref{thm:group} for the Case of Group Languages}
\label{apx:groupproof}

We give the complete proof of Theorem~\ref{thm:group} for the case of group
languages.

Let~$K$ be a group language on the alphabet $A = \{a_1, \ldots, a_k\}$, let
$\mu: A^* \to H$ be the syntactic morphism of~$K$, where 
$H$ is a finite group generated by
the $\mu(a_i)$.
Consider an instance $I = (S_1, \ldots, S_n)$ to the CSh-problem,
where each $S_i$ is a directed path of vertices labeled with letters of the alphabet~$A$.
Recall from Appendix~\ref{apx:highlevel} that $B$ is the bound whose existence is shown in Lemma~\ref{lem:ramsey}, and,
using
Lemma~\ref{lem:groupper} for the value $k \colonequals B$,
$R$ is the value of~$n_k$ given by this lemma.
Recall the definition of a \emph{rare--frequent} partition of $I$
(Definition~\ref{def:rarefreq}) from
Appendix~\ref{apx:highlevel}, and recall that we have used Lemma~\ref{lem:rarefreq} to compute in NL
  a rare--frequent partition of~$I$, given by $A_\rare$, $A_\freq$, $S_\rare$,
  and (implicitly) $S_\freq$. We write $H_\freq$ for the subgroup of~$H$ equal
  to $\mu(A_\freq^*)$, i.e., the subgroup spanned by $A_\freq$.

  Our goal is to determine whether $I$ has some topological sort in~$K$. This is
the case iff it has a topological sort mapped to an accepting element of~$H$
by~$\mu$, so we can equivalently test, for each accepting element of~$H$,
whether there is a topological sort that achieves it. Hence, let $g$ be the
target element. Recall that the commutative closure of the language
$\mu^{-1}(g)$ is a regular language by Lemma~\ref{lem:groupparikh}, and is
obviously commutative. Further recall the morphism $\com_\mu : A^* \to N$ from
Section~\ref{apx:antichain}, where $N$ is a commutative monoid that recognises
the inverse image of all elements of~$H$, in particular~$g$. Recall also the
relational morphism $\tau:N\relto H$ defined by $\tau(x) =
\mu(\com_\mu^{-1}(x))$.

We will state a condition, called (*), and construct an NL algorithm to check
(*). We will then show that (*) holds iff $I$ has a topological sort that
achieves~$g$. Condition (*) is: there exists a topological sort  $\rho$
of~$S_\rare$ which can be decomposed as $\rho_1 \cdots \rho_n$, and a sequence 
$g_0, \ldots, g_n$ of elements of~$H_\freq$, such that:

\begin{enumerate}
  \item $g_0 \mu(\lfun(\rho_1)) g_1 \cdots \mu(\lfun(\rho_n)) g_n = g$;
  \item $g_0 \cdots g_n \in \tau(\com_\mu(S_\freq))$;
  \item $n < B$.
\end{enumerate}

To test this condition (*), we simply nondeterministically guess a sequence $S'$ of
elements of~$H_\freq$ of size at most~$B$ (i.e., a constant)
such that the concatenation of its elements is in
$\tau(\com_\mu(S_\freq))$,
add $S'$
to $S_\rare$, and check whether the resulting CSh instance has a topological sort
using the NL algorithm of 
Proposition~\ref{prp:dynamic} (because its number of strings is at most $R
\cdot k^2 + 1$, which is constant): the language to test is $\mu^{-1}(g)$ on the
modified alphabet where the elements of~$S'$ carry labels in $H_\freq$ and stand
for themselves; note that this clearly yields a group language.

All that remains to show is that condition (*) is equivalent to the existence of
a topological sort of~$I$ that achieves~$g$. 
For the forward direction, assume that condition (*) holds.
Recall that we have defined $R \colonequals n_B$.
Focus on $S_\freq$, which has an $R$-rich antichain for~$A_\freq$, and observe
that $g_0 \cdots g_n \in \tau(\com_\mu(S_\freq)$, which is the equivalent
rephrasing of the condition $\PI(S_\freq) \in \PI(\mu^{-1}(g_0 \cdots g_n))$, as
argued at the beginning of the proof.
Using the antichain lemma (Lemma~\ref{lem:groupper}), we know that there is a
topological sort $\sigma = \sigma_0 \cdots \sigma_n$ of~$S_\freq$
  such that $\mu(\lfun(\sigma_i)) = g_i$
for each $i \in \{0, \ldots, n\}$. Now, considering the topological sort $\rho_1,
\ldots, \rho_n$ of~$S_\rare$ given by condition (*), it is clear that $\sigma_0
  \rho_1
  \sigma_1
\cdots \rho_n \sigma_n$ is a topological sort of~$I$, built by interleaving $S_\rare$
and $S_\freq$; and furthermore $\mu(\lfun(\sigma_0 \rho_1 \sigma_1 \cdots \rho_n
\sigma_n)) =
  \mu(\lfun(\sigma_0))
\mu(\lfun(\rho_1)) \mu(\lfun(\sigma_1)) \cdots \mu(\lfun(\rho_n))
\mu(\lfun(\sigma_n))$, which by (*) is equal to~$g$,
concluding the forward direction of the correctness proof.

\medskip

We now show the backward direction. Assume that there is a topological sort
$\sigma'$ 
of~$I$ achieving~$g$, i.e., $\mu(\sigma') = g$. We can decompose it as an interleaving of~$S_\rare$ and
$S_\freq$, which we write $\sigma_0 \rho_1' \sigma_1 \cdots \rho_{n'}' \sigma_{n'}$,
with $\rho_1'
\cdots \rho_{n'}'$
being a topological sort of~$S_\rare$, and $\sigma_0 \cdots \sigma_{n'}$ being a topological
sort of~$S_\freq$ (in particular, we have $\mu(\lfun(\sigma_0 \cdots
\sigma_{n'})) \in
\tau(\com_\mu(S_\freq))$, which we call condition (\#$2'$)).
We now use the insertion lemma (Lemma~\ref{lem:ramsey}) to argue that there
exists a set $w_0, \ldots, w_{n'}$ of words of~$A^*$, with $w_i =
\lfun(\sigma_i)$ for at most
$B$ values of~$i$ and being the empty word otherwise, such that $\mu(w_0
\lfun(\rho_1') w_1
\cdots \lfun(\rho_{n'}') w_{n'}) = \mu(\sigma') = g$, and 
(\#$2''$) $\mu(\lfun(\sigma_0 \cdots \sigma_{n'})) = \mu(w_0
\cdots w_n)$. We now collapse the $\rho_i'$ which are contiguous,
calling the result $\rho_1, \ldots, \rho_n$, where we have (\#3) $n < B$,
and write $g_i$ the $\mu$-image of the $i$-th $w_i$ which is non-empty: this
image is in $H_\freq$ because the strings in~$S_\freq$ are only labeled with
letters in~$A_\freq$. This
gives us a topological sort $\rho_1, \ldots, \rho_n$ of~$S_\rare$, and a sequence
$g_0, \ldots, g_n$ of elements of~$H_\freq$, such that (\#1) $g_0
\mu(\lfun(\rho_1)) g_1 \cdots
\mu(\lfun(\rho_n)) g_n = g$. By (\#1), (\#$2'$) combined with (\#$2''$), and (\#3), we have
satisfied condition (*). This concludes the backward direction, and establishes
the equivalence proof. Hence, we have shown Theorem~\ref{thm:group} in the case
of group languages.

\subsection{Proof of Theorem~\ref{thm:group} for the Case of District Group Monomials}
\label{apx:district}
We now show the complete proof of Theorem~\ref{thm:group} by adapting the 
proof of Appendix~\ref{apx:groupproof} from the case of group languages to that of district
  group monomials. We write $K = K_0 a_1 K_1 \cdots a_m K_m$, where each $a_i$ is a letter of
the alphabet (they are not necessarily distinct), and each $K_i$ is a group
language on some subset $A_i$ of the alphabet. We fix as before the instance $I
= (S_1, \ldots, S_n)$ of the CSh-problem.
A \emph{$K$-slicing} of
the instance $I$ is an $(m+1)$-tuple of instances $I_0, \ldots, I_m$, with each
$I_j$ being a $n$-tuple $(S^j_1, \ldots, S^j_n)$ of strings, and an $m$-tuple of
instances $I_1', \ldots, I_m'$, with each $I_j'$ being a $n$-tuple $((S')^j_1,
\ldots, (S')^j_n)$ as before, with the stipulation that, for each $1 \leq j \leq
m$, all $(S')^j_i$ are empty except one which is a singleton whose only element
is labeled $a_j$; and that, for each $1 \leq i \leq n$, the concatenation $S^0_i
(S')^1_i S^i_i \cdots (S')^m_i S^m_i$ is equal to~$S_i$. In other words, a
slicing is a partition of each string of $I$ in a way that respects the $a_i$.

Intuitively, we would like to guess a slicing, check the $I_i'$ in the obvious
way, and apply the previous result to the $I_j$ for odd~$j$, corresponding to
the group languages $K_j$. Unfortunately, while guessing the even $I_j$ is
immediate, we cannot afford to guess the entire slicing in NL. For this reason,
we need a more elaborate approach.

We will follow the previous proof and introduce a notion of rare--frequent
partition, generalised to slicings. 
As before, we let $B$ be the bound whose existence is shown in
Lemma~\ref{lem:ramsey}, use Lemma~\ref{lem:groupper} with $k \colonequals B$ to
obtain $n_k$, and let $R \colonequals n_k$.
Given a slicing $I_0 \ldots
I_{m}$ and $I'_1 \ldots I'_m$,
a \emph{rare--frequent partition} of the slicing consists of one partition
$A_\rare^j$, $A_\freq^j$ for all $1 \leq j \leq m$, and one \emph{global}
partition of the strings $S_1, \ldots, S_n$ into rare strings $S_\rare$ and
frequent strings $S_\freq$ (again, the frequent strings are not explicitly
represented). We require that (i) for every string $S$ of $S_\freq$,
considering its slices $S^0, \ldots, S^m$, for each $1 \leq j \leq m$, the slice
$S^j$ contains only letters of~$A^j_\freq$; that (ii) for every $1 \leq j \leq
m$, the $S^j$ for $S$ in $S_\freq$, when seen as a subinstance of~$I$ over the
alphabet $A_\freq^j$, contains an $R$-rich antichain; and that (iii) for every
$1 \leq j \leq m$, the one non-empty string of $I'_j$ is in $S_\rare$.

We can show as before that, for any slicing, we can compute a rare--frequent
partition. In fact we will only need to show that it exists, as the problem in
guessing the slicing prevents us from guessing it anyway.

\begin{lemma}
  \label{lem:rarefreq2}
  For any slicing $I_0, \ldots, I_m$, $I_1', \ldots, I_m'$, there exists a
  rare--frequent partition such that $\card{S_\rare} \leq m \cdot R \cdot k^2$.
\end{lemma}

\begin{proof}
  We apply Lemma~\ref{lem:rarefreq} to each $I_j$ for $1 \leq j \leq m$ to
  obtain one rare--frequent partition for it, written $A_\rare^j \sqcup
  A_\freq^j = A_j$ and $S_\rare^j \sqcup S_\freq^j = I_j$, except that we take
  $m \times (R+2)$ instead of~$m$. Now, the only thing that remains is to justify
  that we can take the set of rare strings to be global instead of local, and to
  satisfy condition (iii). We
  simply then take $S_\rare$ to be the union of the strings $S$ of~$I$ such that
  $S^j$ is in~$S_\rare^j$ for some~$0 \leq j \leq m$, plus the strings that are
  non-empty in some $I'_j$. We take $S_\freq$ to be
  the complement. This ensures that condition (iii) is respected by construction.
  Now, it is clear that condition (i) is respected, as, for each
  slice, the frequent strings to consider are a subset of the one given by the
  previous condition. Now, condition (ii) is respected because it was respected
  initially for the richness threshold of $m \times (R+2)$, and we have only removed
  at most $m\times (R+1)$ frequent strings in the modification: $((m+1)-1)\times R$ for
  the other slices of the form $I_j$, and $m$ for the slices of the form $I_j'$.
  Hence, we can deduce
  an $R$-rich antichain by looking at any preexisting $(m \times (R+2))$-rich
  antichain.
\end{proof}

While we cannot guess the slices, let us guess partitions $A_j = A_\rare^j
\sqcup A_\freq^j$ for $0 \leq j \leq m$ and the set $S_\rare$ of (globally) rare
strings of size at most~$R \cdot k^2$. Let us further guess the slices $S_j'$ for $1 \leq j \leq m$, i.e., we
guess elements in~$I$ with suitable order and labels. As the number of rare
strings is constant and $m$ is constant, we guess, for each string of $S_\rare$,
the $m$ points at which the slices end, i.e., we guess a slice but restricted to
the rare strings. As for the frequent strings, we will not guess the slices
globally, as there is generally a non-constant number of frequent strings.
However, we will guess the ``sequence of insertions'' to be performed using the
frequent antichains for each slice, i.e., the analogue to the sequence $g_0,
\ldots, g_n$ in condition (*) in the previous proof. Formally, we guess a
sequence $g_0^j, \ldots, g_{n_j}^j$ for all $0 \leq j \leq m$, with each $g_i^j$
being an element of $H_\freq^j$, the subgroup of $H_j$ spanned by~$A_\freq^j$.
Last, we also guess an element $\gamma_0, \ldots, \gamma_m$ of $H_0 \times
\cdots \times H_m$ to describe the accepting elements of the $H_i$ achieved in
each slice.

Intuitively, we will now do two things: first, verify that our
guesses are consistent (except for the choice of the $\gamma_i$); second, reduce
the problem to a simpler problem by replacing all strings of $S_\freq$ with an
additional string labeled directly with elements of the groups $H_i$ of the group
languages $K_i$, as in the previous proof.

First, to verify that our guesses are consistent, we check the rare strings. On
these strings, it is straightforward to verify that the sub-alphabet for each
slice is respected. Further, for the slices $I_j'$, the verification is
immediate. Now, for the frequent strings, we go over them in succession. We
maintain a state that stores, for each slice of the form $I_j$ for $0 \leq j
\leq m$, how many occurrences of each letter of~$A$ we have seen in the slice
$j$, and in how many different strings are these occurrences. Initially, each letter occurs 0 times. Now, when processing a frequent
string $S$ which is in~$S_\freq$ (i.e., not in~$S_\rare$), we guess a slicing
of~$S$, count the number of occurrences of each letter in each slice and add it
to our counter of occurrences, and add one to the counter of strings for the
symbols that did occur. At the end, we check that the value of our counters
satisfies some conditions, which will witness the existence of a suitable
slicing of the frequent strings. Specifically, we verify:

\begin{itemize}
  \item For each $0 \leq j \leq m$, for each $a \in A \setminus A^j_\freq$, that our
    choice of slicing does not contain any occurrence of $a$ in the restriction
    of the slice $I_j$ to~$S_\freq$.
  \item For each $0 \leq j \leq m$, for each $a \in A_\freq^j$, that our choice
    of slicing ensures that there are at least $R$ different strings that contain
    an occurrence of~$a$ in the restriction of slice $I_j$ to $S_\freq$,
    witnessing that it has an $R$-rich antichain for the alphabet $A_\freq^j$.
  \item For each $0 \leq j \leq m$, letting $w$ be the word containing all
    letters of the restriction of slice $S_j$ to $S_\freq$ with the correct
    number of occurrences, that $g_0^j \cdots g_{n_j}^j \in
    \tau_j(\com_{\mu_j}(w))$, intuitively checking that we have the right
    commutative image.
\end{itemize}

Second, we check the following condition (**), inspired from condition (*) in the
previous proof: for all $0 \leq j \leq m$,
there exist a topological sort $\rho_0^j \cdots \rho_{n_j}^j$ of the slice $S_\rare^j$ of
$S_\rare$ whose
concatenation, interleaved with the singleton elements of the $I'_j$, is a
topological sort of~$S_\rare$, and
$g_0^j \lfun(\rho_0^j) \cdots g_{n_j-1}^j \lfun(\rho_{n_j-1}^j) g^j{n_j} = \gamma_j$. This can be
decided in NL by adapting the algorithm of Proposition~\ref{prp:dynamic} as
previously, running it on each slice with one additional string.

Overall, our algorithm succeeds iff there is a guess of $\gamma_i$, of
$S_\rare$ (at most $R k^2$ of them), partitions $A_\freq^j \sqcup A_\rare^j$,
and sequences $g_0^j, \ldots, g_{n_j}^j$, such that the verification stage
succeeds, and condition (**) holds.

We have described our NL algorithm. We now argue that it works as intended.
There are two directions: the forward direction is to show that if the algorithm
succeeds then there is a suitable topological sort of~$I$, and the backward
direction is to show the converse.

For the forward direction, assume that the algorithm succeeds. We deduce the
existence of a set $S_\rare$ of rare strings (whose slices are written
$S^j_\rare$), and frequent strings $S_\freq$ (with the same convention for
slices), partitions $A_\freq^j \sqcup
A_\rare^j$, a slicing $I_0, \ldots, I_m$ and $I_1', \ldots, I_m'$,
a topological sort of~$S_\freq$ constituting of topological sorts $\rho_0^j \cdots
\rho_{n_j}^j$ of each $S_\rare^j$ for $0 \leq j \leq m$ interleaved with the singleton elements of the
$I_j'$ for $1 \leq j \leq m$, sequences $g_0^j, \ldots, g_{n_j}^j$ of elements
of~$H_j$ for $0 \leq j \leq m$, and an element $\gamma_0, \ldots, \gamma_m$ of
$H_0 \times \cdots \times H_m$, such that:

\begin{itemize}
  \item For all $0 \leq j \leq m$, the element $\gamma_j$ is accepting in~$H_j$.
  \item For all $0 \leq j \leq m$, for all $S \in S_\freq$, the slice $S^j$
    contains only letters from $A_\freq^j$, and contains an $R$-rich
    antichain on the sub-alphabet $A_\freq^j$.
  \item For all $0 \leq j \leq m$, for all $S \in S_\rare$, the slice $S^j$
    contains only letters from $A_j$.
  \item For all $0 \leq j \leq m$, letting $S_\freq^j$ be the slice of $S_\freq$
    defined in the expected way, we have $g_0^j \cdots g_{n_j}^j \in
    \tau_j(\com_{\mu_j}(S_\freq^j))$.
  \item (\#) For all $0 \leq j \leq m$, we have 
    $g_0^j \lfun(\rho_0^j) \cdots g_{n_j-1}^j \lfun(\rho_{n_j-1}^j) g^j_{n_j} = \gamma_j$
\end{itemize}

We claim that we can deduce from this the existence of a witnessing topological
sort. To do this, as before, we will use Lemma~\ref{lem:groupper} in
the~$S_\freq^j$ for all $0\leq j \leq m$. From our definition of $R$, as $n_j
< B$,
as $S_\freq^j$ contains an $n_k$-rich antichain (seen as an instance on the
sub-alphabet $A_\freq^j$), as $g_1^j \cdots g_{n_j}^j \in
\tau_j(\com_{\mu_j}(S^j_\freq))$, there is a topological sort $\sigma_1^j \ldots
\sigma_{n_j}^j$ of~$S_\freq^j$ such that $\mu_j(\lfun(\sigma_i^j)) = g_i^j$ for each $0 \leq j
\leq m$ and $1 \leq i \leq n_j$. This allows us to deduce our witnessing
topological sort of~$I$, consisting of a topological sort of each slice $I_j$
of~$I$ achieving $\gamma_j$, interleaved with the trivial topological sorts of
the $I_j'$ that achieve the required $a_j$: the topological sort of~$I_j$ is
formed of the guessed topological sort $\rho_0^j \cdots \rho_{n_j}^j$ of $S_\rare^j$
interleaved with the topological sort $\sigma_1^j, \ldots, \sigma_{n_j}^j$ of
$S_\freq^j$, each $v_i^j$ achieving $g_i^j$, so that the topological sort
of~$I_j$ indeed achieves $\gamma_j$ by point (\#).

\medskip

We now show the backward direction. We show that if there is a suitable
topological sort, then the algorithm succeeds. The witnessing topological sort
must define a slicing of~$I$ such that each $I_j$ for $0 \leq j \leq m$ has a
topological sort achieving an element $\gamma_j$ which is accepting for~$H_j$.
We now use Lemma~\ref{lem:rarefreq2} to argue that there exists a rare--frequent
partition consisting of a partition $S_\rare \sqcup S_\freq$ of the strings, and
$A_\rare^j \sqcup A_\freq^j$ of the alphabets $A_j$, such that $\card{S_\rare}
\leq m \cdot R \cdot k^2$. In each slice, the witnessing topological sort must
consist of a topological sort of the $S_\rare^j$ interleaved with topological
sorts of the $S_\freq^j$. As in the previous proof, we now use
Lemma~\ref{lem:ramsey} to argue that we can assume that there are at most $n_j$
such insertions, without changing the $\mu_j$-image of the result or the
$\mu_j$-image of the inserted elements. Now, we define the $g_1, \ldots,
g_{n_j}^j$ as the $\mu_j$-images of these insertions. We now consider the run of
the algorithm where we guess the right rare--frequent partition, the right slices
in the rare strings, the right topological sort of the rare strings.

We first
check that the verification phase of the algorithm does not fail. This is the
case: the first condition is by definition of a witnessing topological sort (for
$A \backslash A_j$) and of a rare--frequent partition (for $A_j \backslash
A_\rare^j$); the second condition is by definition of a rare--frequent partition;
the third condition is by definition of $g_1, \ldots, g_{n_j}^j$ being achieved
as a topological sort of~$S_\rare^j$. We next explain why the second phase
works, by explaining why condition (**) is satisfied. This can be seen by
considering when the insertions of the $S_\freq^j$ are performed in the
$S_\rare^j$: we perform the same additions with the additional string. Hence,
this run of the algorithm succeeds. This concludes the backwards direction of
the correctness proof, so our NL algorithm is correct. This concludes the proof
of Theorem~\ref{thm:group}.

\end{toappendix}

\begin{toappendix}
  \subsection{Proof of Proposition~\ref{prp:aabbaabb}: Example in
  $\DS\setminus\DO$}
  \label{apx:limitations}
\end{toappendix}
We close the section by commenting on the two main limitations of
Theorem~\ref{thm:groupdich}.
The first limitation is that
it is not a dichotomy: it does not cover the semiautomata with
transition monoid 
in $\DS \setminus \DO$.
We do not know if the corresponding languages are tractable or not; we have not
identified
intractable cases, but we can show tractability, e.g.,
for $(a^+b^+a^+b^+)^*$,
the language of words with an even number of subfactors of the form $a^+b^+$.

\begin{toappendix}
  We show the side result on the language in $\DS \setminus \DO$. Note that the
  fact that this language is indeed in \DS and not in \DO can be simply checked
  from the equations that define \DS and \DO, as can be performed, e.g.,
  using~\cite{papermansemigroup}.
\end{toappendix}

\begin{propositionrep}
  \label{prp:aabbaabb}
  Let $K = (a^+b^+a^+b^+)^*$. Then $\pCSh{K}$ is in NL.
\end{propositionrep}

\begin{proof}
  Consider an
  input instance $I$ to the $\CSh$-problem for~$K$. Observe first that, if $I$
  has no string whose first element is~$a$, then clearly no topological sort
  of~$I$ achieves~$K$. Likewise, if $I$ has no string whose last element is~$b$,
  then clearly no topological sort of~$I$ achieves~$K$. We can check these two
  conditions in NL and fail if one of them does not hold, so in the sequel we
  assume that $I$ has a string whose first element is~$a$ and a string whose
  last element is~$b$.
  
  Recall that a $3$-rich antichain
  for~$A$ in~$I$ is an antichain containing at least $3$ elements labeled by~$a$
  and $3$ elements labeled by~$b$. We show that if $I$ contains a $3$-rich
  antichain then it is necessarily a positive instance to $\pCSh{K}$. Of course,
  note that we can easily test in NL if such a $3$-rich antichain exists.

  To show the claim, let $C''$ be such an antichain, and $C'$ be a subset
  of~$C''$ containing
  exactly three occurrences of each letter; it is still an antichain.
  We now define $C$ as a subset of~$C'$ containing exactly two occurrences of
  each letter, and ensuring that there is an $a$-labeled element $v_a$ which is the
  first element of a string and is not in a string of~$C$, and likewise there is a
  $b$-labeled element $v_b$ which is the last element of a string and is not
  in a string of~$C$: we can ensure this because we can choose which $a$-labeled
  element and which $b$-labeled element we remove from~$C'$ to construct~$C$.
  
  Now, consider a topological
  sort $\sigma_1$ of~$I$ formed by concatenating $v_a$, a topological sort~$\sigma_-$
  of the ancestors of elements of~$C$ and of the elements incomparable to~$C$
  except~$v_a$ and~$v_b$, a topological sort $\sigma_1'$
  of~$C$ achieving the word $aabb$, a topological sort~$\sigma_+$ of the
  successors of~$C$, and $v_b$. The word $w_1$ achieved by $\sigma_1$ starts with~$a$
  and ends with~$b$, so it must be of the form $(a^+b^+)^*$.
  Let $n_1$ be
  the number of repetitions of $a^+b^+$ in~$w_1$.
  Now, consider the topological
  sort $\sigma_2$ obtained by concatenating $v_a$, $\sigma_-$, $\sigma_2'$,
  $\sigma_+$, and $v_b$, where $\sigma_2'$ is a topological sort of~$C$
  achieving the word $abab$. Again, the word~$w_2$ achieved by~$\sigma_2$ must
  be of the form $(a^+b^+)^*$: let $n_2$ be the number of repetitions of
  $a^+b^+$ in~$w_2$. We claim that $n_2 = n_1 +
  1$. Indeed, consider the subfactor $a^+b^+$ that contains $\sigma_1'$
  in~$\sigma_1$. In $\sigma_2$, the other subfactors are unchanged, and this
  subfactor is split into two subfactors, one ending at the first~$b$
  of~$\sigma'_2$, the other one starting at the second~$a$ of~$\sigma'_2$. So indeed
  $n_2 = n_1 + 1$. Hence, one of $n_1, n_2$ is even, and the corresponding
  $\sigma_i$ witnesses that $I$ is a positive instance to $\pCSh{K}$.

  Hence, it suffices to handle the case where $I$ has no $3$-rich antichain.
  This implies that there is one symbol $\alpha \in A$ which occurs in at most
  two strings $S$ and~$S'$, which means that the other strings $S_1, \ldots, S_m$ only
  contain elements labeled with the other symbol $\beta \neq \alpha$ of~$A$.
  Now, it is easy to see that we obtain exactly the same topological sorts by
  merging together the $S_1, \ldots, S_m$ to one string $S''$ of elements labeled
  $\beta$ whose length is $\sum_i \card{S_i}$. Hence, we can reduce the problem
  in NL to the instance $\{S, S', S''\}$. As it has three strings, we can conclude in NL
  using Proposition~\ref{prp:dynamic}. Hence, we have indeed shown that
  $\pCSh{K}$ is in NL.
\end{proof}

However, it would be difficult to show tractability for all of \DS, because
\DS is still poorly understood in algebraic language theory.
For instance, characterizing the languages with a syntactic monoid in \DS has
been open for over 20 years \cite[Open problem 14, page 442]{almeida1994finite}.

The second limitation of Theorems~\ref{thm:groupdich} and~\ref{thm:group}
is that they only apply to CSh.
New problems arise with CTS: for instance, an
\mbox{$\{a,b\}$-DAG}~$G$ may contain large antichains~$C_a$ and~$C_b$ of
$a$-labeled and $b$-labeled vertices, and yet contain no antichain with
many $a$-labeled and $b$-labeled vertices (e.g., if $G$ is the series composition
of $C_a$ and~$C_b$). The missing proof ingredient seems to be an 
analogue of Dilworth's theorem for \emph{labeled} DAGs (see
also~\cite{amarilli2016generalization}).

\section{Conclusion and Open Problems}
\nosectionappendix
\label{sec:conclusion}
We have studied the complexity of two problems, constrained topological sort
(CTS) and constrained shuffle (CSh): fixing a regular language $K$, given
a labeled DAG (for CTS) or a tuple of strings (for CSh), we ask if the input DAG
has a topological sort achieving~$K$. We have shown tractability and
intractability for several regular languages using a variety of techniques.
These results yield a coarser dichotomy (Theorem~\ref{thm:dich}) in an alternate
problem phrasing that imposes some closure assumptions.

Our work leaves the main dichotomy conjecture open
(Conjecture~\ref{con:maincon}). Even in the alternate problem phrasing of
Theorem~\ref{thm:dich}, our dichotomy only covers counter-free semiautomata:
the restriction is lifted in Section~\ref{sec:group} but only for CSh,
and with a gap between tractability and intractability.
In the original phrasing, there are many concrete languages 
that we do not understand: Does Proposition~\ref{prp:abaaaa} extend to
$(ab)^* + A^* a^i A^*$ for $i > 2$? Does Proposition~\ref{prp:aab} extend to
$(a^i+b)^*$ for $i>2$, or to CTS rather than CSh? Can we show
Conjecture~\ref{con:fstar}?

Another direction would be to connect CSh and CTS to the framework of
\emph{constraint satisfaction problems} (CSP)~\cite{feder1998computational}, which
studies the complexity of homomorphism problems for fixed ``constraints''
(right-hand-side of the homomorphism). If this were possible, it could
lead to a better understanding of our tractable and hard cases. However, CTS
does not seem easy to rephrase
in CSP terms: topological sorts and regular language
constraints seems hard to express in terms of homomorphisms, even in 
extensions such as \emph{temporal} CSPs~\cite{bodirsky2010complexity,bodirsky2016discrete}.

One last question would be to
investigate CTS and CSh for \emph{non-regular} languages.
The simplest example is the Dyck language, which appears to be
NP-hard for CTS (at least in the multi-letter
setting), but tractable for CSh, via a connection to scheduling; see
\cite{garey1979computers}, problem~SS7.
More generally, CTS and CSh could be studied,
e.g., for context-free languages, where the complexity landscape may be
equally enigmatic.

\bibliography{main}
\end{document}